%% file: main.tex
\def\confversion{0}
\def\ifconf{\ifnum\confversion=1}
\def\ifnotconf{\ifnum\confversion=0}

\documentclass[11 pt, fleqn, reqno]{article}

\def\showauthornotes{0}
\def\showkeys{0}
\def\showdraftbox{0}

\input{macros}

\input{abbrev}

\usepackage{tikz}
\usetikzlibrary{shapes.symbols}
\usepackage{float}
\usepackage{graphicx}
\usepackage{pstyle}

\usepackage[top=1in, bottom=1in, left=1.25in, right=1.25in]{geometry}

\newcommand{\cK}{\mathcal{K}}
\newcommand{\cZ}{\mathcal{Z}}

\DeclareMathOperator{\PRP}{\operatorname {PRP}}

\newcommand{\fnote}[1]{\textcolor{blue}{ {\textbf{(Fernando: #1)}}}}

\newcommand{\nnote}[1] {\textcolor{purple}{ {\textbf{Nir: #1}}}}

\begin{document}
\sloppy

\title{Coherence in Property Testing: Quantum-Classical\\
        Collapses and Separations}

 \author{Fernando Granha Jeronimo\thanks{{\tt University of Illinois Urbana-Champaign}. {\tt granha@illinois.edu}. Supported as a Google Research Fellow.}
  \and Nir Magrafta\thanks{{\tt Weizmann Institute of Science}. {\tt nir.magrafta@weizmann.ac.il}. Supported by the Israel Science Foundation (Grant No.\ 3426/21), and by the European Union Horizon 2020 Research and Innovation Program via ERC Project REACT (Grant 756482)}
  \and Joseph Slote\thanks{{\tt Caltech}. {\tt jslote@caltech.edu}. Supported by Chris Umans Simons Investigator Grant.} 
  \and Pei Wu\thanks{{\tt Penn State University}. {\tt pei.wu@psu.edu}.  Supported by ERC Consolidator Grant VerNisQDevS (101086733).}}

\date{\today}



\maketitle
\draftbox
\thispagestyle{empty}
\input{abstract}

\newpage

\ifnotconf
\pagenumbering{roman}
\tableofcontents
\clearpage
\fi

\pagenumbering{arabic}
\setcounter{page}{1}

\input{intro}

\input{prelim}

\input{strategy}

\input{q-indist}

\input{c-cert-testing}

\input{deMerlinization}

\input{supp_prot}

\input{hierarchy}

\section*{Acknowlegdments}

The authors are very thankful for the amazing support provided by the Simons Institute.
In particular, we would like to highlight the importance to us of the programs and clusters:
``HDX and Codes'', ``Analysis and TCS: New Frontiers'', ``Error-Correcting Codes: Theory and Practice'',
and ``Quantum Algorithms, Complexity, and Fault Tolerance''.
FGJ thanks Venkat Guruswami for kindly hosting him in his fantastic research group.
FGJ thanks Umesh Vazirani for helping with the bureaucracy to bring him to the Simons
Institute (an example of the great generosity Sandy Irani described during the UmeshFest).

\newpage
\bibliographystyle{alpha}
\bibliography{macros,references}

\input{appendix}

\end{document}

%% file: macros.tex
\usepackage{xspace,xcolor,enumerate}
\usepackage{amsmath,amssymb}
\usepackage{amsthm}
\usepackage[toc,page]{appendix}
\usepackage{thmtools}
\usepackage{thm-restate}
\usepackage{color,graphicx}
\usepackage{boxedminipage}
\usepackage{makecell}
\usepackage{tabularx}
\usepackage{enumitem}
\usepackage[linesnumbered,ruled,vlined]{algorithm2e}
\ifnum\showkeys=1
\usepackage[color]{showkeys}
\fi

\definecolor{darkred}{rgb}{0.5,0,0}
\definecolor{darkgreen}{rgb}{0,0.35,0}
\definecolor{darkblue}{rgb}{0,0,0.55}

\usepackage[pdfstartview=FitH,pdfpagemode=UseNone,colorlinks,linkcolor=darkblue,filecolor=darkred,citecolor=darkgreen,urlcolor=darkred,pagebackref]{hyperref}

\usepackage[capitalise,nameinlink]{cleveref}
\usepackage[T1]{fontenc}
\usepackage{mathtools,dsfont,bbm}

\ifnum\showauthornotes=1
\newcommand{\Authornote}[2]{{\sf\small\color{red}{[#1: #2]}}}
\newcommand{\Authorcomment}[2]{{\sf \small\color{gray}{[#1: #2]}}}
\newcommand{\Authorfnote}[2]{\footnote{\color{red}{#1: #2}}}
\else
\newcommand{\Authornote}[2]{}
\newcommand{\Authorcomment}[2]{}
\newcommand{\Authorfnote}[2]{}
\fi

\ifnum\showdraftbox=1
\newcommand{\draftbox}{\begin{center}
  \fbox{%
    \begin{minipage}{2in}%
      \begin{center}%
        \begin{Large}%
          \textsc{Working Draft}%
        \end{Large}\\
        Please do not distribute%
      \end{center}%
    \end{minipage}%
  }%
\end{center}
\vspace{0.2cm}}
\else
\newcommand{\draftbox}{}
\fi

\newtheorem{theorem}{Theorem}[section]

\newtheorem{definition}[theorem]{Definition}

\newtheorem{lemma}[theorem]{Lemma}
\newtheorem{remark}[theorem]{Remark}
\newtheorem{proposition}[theorem]{Proposition}
\newtheorem{corollary}[theorem]{Corollary}
\newtheorem{claim}[theorem]{Claim}
\newtheorem{fact}[theorem]{Fact}

\newtheorem{problem}[theorem]{Problem}

\theoremstyle{remark}
\newtheorem{algo}[theorem]{Algorithm}

\def\FullBox{\hbox{\vrule width 6pt height 6pt depth 0pt}}

\def\qedsketch{\ifmmode\Box\else{\unskip\nobreak\hfil
\penalty50\hskip1em\null\nobreak\hfil$\Box$
\parfillskip=0pt\finalhyphendemerits=0\endgraf}\fi}

\def\to{\rightarrow}

\def\epsilon{\varepsilon}

\def\cal{\mathcal}

\def\implies{\Rightarrow}

\newcommand{\ie}{i.e.,\xspace}
\newcommand{\eg}{e.g.,\xspace}
\newcommand{\etal}{et al.\xspace}

\newcommand{\mper}{.\,}
\newcommand{\mcom}{,\,}

\newcommand{\R}{{\mathbb R}}
\newcommand{\E}{{\mathbb E}}
\newcommand{\C}{{\mathbb C}}

\newcommand{\cA}{\mathcal{A}}

\newcommand{\cH}{\mathcal{H}}
\newcommand{\cD}{\mathcal{D}}
\newcommand{\cF}{\mathcal{F}}
\newcommand{\cE}{\mathcal{E}}

\newcommand{\cS}{\mathcal{S}}

\usepackage{nicefrac}

\newcommand{\abs}[1]{\ensuremath{\left\lvert #1 \right\rvert}}

\newcommand{\norm}[1]{\ensuremath{\left\lVert #1 \right\rVert}}

\makeatletter

\def\ProbabilityRender#1#2{
  \@ifnextchar\bgroup%
  {\renderwithdist{#1}{#2}}
   {\singlervrender{#1}{#2}}
}
\def\singlervrender#1#2{%
   \ensuremath{\mathchoice
       {{#1}\left[ #2 \right]}
       {{#1}[ #2 ]}
       {{#1}[ #2 ]}
       {{#1}[ #2 ]}
   }
}
\def\renderwithdist#1#2#3{%
   \@ifnextchar\bgroup
   {\superfancyrender{#1}{#2}{#3}}
   {\ensuremath{\mathchoice
      {\underset{#2}{#1}\left[ #3 \right]}
      {{#1}_{#2}[ #3 ]}
      {{#1}_{#2}[ #3 ]}
      {{#1}_{#2}[ #3 ]}
     }
   }
}
\def\superfancyrender#1#2#3#4#5{
   \ensuremath{\mathchoice
      {\underset{#1}{{#1}}\left#4 #3 \right#5}
      {{#1}_{#2}#4 #3 #5}
      {{#1}_{#2}#4 #3 #5}
      {{#1}_{#2}#4 #3 #5}
   }
}
\makeatother

\newcommand{\conv}[1]{\mathrm{conv}\inparen{#1}}

\newfont{\inhead}{eufm10 scaled\magstep1}

\newcommand{\calA}{{\cal A}}
\newcommand{\calB}{{\cal B}}
\newcommand{\calC}{{\cal C}}
\newcommand{\calD}{{\cal D}}

\newcommand{\calP}{{\cal P}}

\newcommand{\calS}{{\cal S}}

\newcommand{\poly}{{\mathrm{poly}}}
\newcommand{\polylog}{{\mathrm{polylog}}}

\DeclareMathOperator\supp{supp}

\DeclareMathOperator{\Span}{\operatorname {span}}

\newcommand{\classfont}[1]{\textup{\textsf{#1}}}

\newcommand{\classP}{\classfont{P}\xspace}

\newcommand{\inparen}[1]{\left(#1\right)}             

\newcommand{\ket}[1]{\lvert #1\rangle}
\newcommand{\braket}[2]{\left\langle #1 \,\middle\vert\, #2\right\rangle}
\newcommand{\QMA}{\classfont{QMA}}
\newcommand{\BQP}{\classfont{BQP}}

\newcommand{\NP}{\classfont{NP}}
\newcommand{\MA}{\classfont{MA}}

\newcommand{\AM}{\classfont{AM}}
\newcommand{\NEXP}{\classfont{NEXP}}

\newcommand{\IP}{\classfont{IP}}

\newcommand{\propBPP}{\classfont{propBPP}}
\newcommand{\propBQP}{\classfont{propBQP}}
\newcommand{\propQMA}{\classfont{propQMA}}

\newcommand{\propEXP}{\classfont{propEXP}}
\newcommand{\propMA}{\classfont{propMA}}
\newcommand{\propMAexp}{{\classfont{propMA}_\textup{exp}}}
\newcommand{\propQMAexp}{{\classfont{propQMA}_\textup{exp}}}
\newcommand{\propAM}{\classfont{propAM}}
\newcommand{\propIP}{\classfont{propIP}}
\newcommand{\propMIP}{\classfont{propMIP}}

%% file: abbrev.tex
\DeclareSymbolFont{extraup}{U}{zavm}{m}{n}
\DeclareMathSymbol{\varheart}{\mathalpha}{extraup}{86}
\DeclareMathSymbol{\vardiamond}{\mathalpha}{extraup}{87}

\DeclareMathOperator{\Tr}{Tr}

\DeclarePairedDelimiter\set{\lbrace}{\rbrace}

\usepackage{tikz}



%% file: abstract.tex
\begin{abstract}
  Understanding the power and limitations of classical and quantum information, and how
  they differ, is an important endeavor. 
  On the classical side, property testing of distributions is a fundamental task: 
  a tester, given samples of a distribution  over a typically large domain such as
  $\{0,1\}^n$, is asked to verify properties of the distribution.
  A key property of interest in this paper is the \emph{support size} both of distributions,
  a central problem classically [Valiant and Valiant STOC'11], as well, as of quantum states.
  Classically, even given $2^{n/16}$ samples, no tester can distinguish between 
  distributions of support size $2^{n/8}$ from $2^{n/4}$ with probability better than $2^{-\Theta(n)}$, even
  with the promise that they are flat distributions.

  In the quantum setting, quantum states can be in a coherent superposition of many states of
  $\{0,1\}^n$, providing a global description of probability distributions.
  One may ask if coherence can enhance property testing. A natural way to encode a flat
  distribution is via the \emph{subset states}, $\lvert \phi_S \rangle = 1/\sqrt{\lvert S \rvert} \sum_{i \in S} \ket{i}$.
  We show that coherence alone is not enough to improve the testability of support size.
  \begin{itemize}
    \item[1.] \textbf{Coherence limitations.} Given $2^{n/16}$ copies, no tester can distinguish between subset states of size $2^{n/8}$
               from $2^{n/4}$ with probability better than $2^{-\Theta(n)}$.
  \end{itemize}  
  Our result is more general and establishes the indistinguishability between the subset states and the Haar random states
  leading to new constructions of pseudorandom and pseudoentangled states, resolving an open problem of [Ji, Liu and Song, CRYPTO'18].

  The hardness persists even when allowing multiple public-coin AM provers for a classical tester. 
  \begin{itemize}
  \item[2.] \textbf{Classical hardness with provers.} Given $2^{O(n)}$ samples from a classical distribution and $2^{O(n)}$ communication with multiple independent AM provers,
            no classical tester can estimate the support size up to factors $2^{\Omega(n)}$ with probability better than $2^{-\Theta(n)}$. Our hardness result is tight.
  \end{itemize}
 In contrast, coherent subset state proofs suffice to improve testability exponentially,
  \begin{itemize}
    \item[3.] \textbf{Quantum advantage with certificates.} With polynomially many copies and subset state proofs, a tester can
                 approximate the support size of a subset state of arbitrary size.
  \end{itemize}
  Some structural assumption on the quantum proofs is required since we show that 
  \begin{itemize}
      \item[4.] \textbf{Collapse of $\QMA$.} A general proof cannot 
  improve testability of \emph{any}  quantum property whatsoever.
  \end{itemize} 
  
  Our results highlight both the power and limitations of coherence in
  property testing, establishing exponential quantum-classical separations across various parameters. We also show several connections
  and implications of the study of property testing, in particular, in establishing quantum-to-quantum state transformation lower bounds,
  and to disentangler lower bounds.
\end{abstract}

%% file: intro.tex
\section{Introduction}

Testing whether a given object has a desired property, or is \emph{far} from it, is a fundamental task
both in the classical and quantum settings~\cite{G17,MW16}, possessing myriad important applications,
\eg \cite{Dinur07pcp,DELLM22,PK22}.
In this context, understanding the resources (e.g., number of copies of a state or samples, efficiency of tester, etc) needed to
test a given property is a central goal of property testing. A key property of interest in this paper is the
\emph{support size} of both distributions, a central property classically~\cite{V11,VV11,HR22label-inv},
as well as of quantum states~\cite{AKKT20,JW23}.

In testing properties of classical probability distributions, one is given access to a distribution via its samples. The study on what properties can be understood with a small number of samples can be traced back to the work of Fisher \etal~\cite{fisher1943relation} and Turing~\cite{good1953population}.
In many computer science problems, we are faced with high-dimensional distributions, which can be seen as assigning probabilities to $\set{0,1}^n$.
A rich theory of property testing of distributions has emerged, and we now know the sample complexity of many properties
of interest~\cite{R10,V11,R12,G17,C20,C22}. There, one quickly learns that several properties require $2^{\Omega(n)}$ samples to be testable, \eg
distinguishing the support size between two families of distributions can require exponentially many samples even if they have vastly different
support sizes and are promised to be flat distributions.

\begin{theorem}[Failure of Classical Testing (Informal)]
  Even given $2^{n/16}$ copies, no tester can distinguish between flat distributions of size $2^{n/8}$
  from $2^{n/4}$ with probability better than $2^{-\Theta(n)}$.
\end{theorem}

This kind of strong lower bound is pervasive in property testing of distributions
\cite{BFFKRW01,BDKR02,RRSS07,VV11,VV17,R12, HR22label-inv},
and it establishes severe limitations on our ability to test classical information. Roughly speaking, this is not
surprising since by accessing a probability distribution via samples, we do not get a ``global'' hold on it, but rather,
we just get random local pieces of it. In contrast, quantum mechanics allows us to manipulate objects that are global
in the sense they are in superposition of possibly many different states of $\set{0,1}^n$. This phenomenon is
known as \emph{coherence}, and it is one of the fundamental pillars of quantum mechanics.
We can then ask what improvements the setting of quantum information can provide, more specifically, whether this global
nature of coherence can lead to substantial improvements in distinguishing vastly different support sizes.

\begin{center}
  \em How much can coherence help property testing?
\end{center}

We show that coherence alone cannot help. A quantum analog of a flat probability distribution is a subset state, namely,
a quantum state of the form $1/\sqrt{\abs{S}} \sum_{i \in S} \ket{i}$ for some $S \subseteq \set{0,1}^n$. In words, this state is
a uniform superposition over some set $S$. These states are natural in their own right and they are commonly used in quantum
complexity~\cite{VW16,JW23}. More precisely, we prove the following result analogous to the classical case.

\begin{theorem}[Failure of Testing with Coherence (Informal)]\label{theo:failure_coherence}
  Even given $2^{n/16}$ copies, no tester can distinguish between subset states of
  size $2^{n/8}$ from $2^{n/4}$ with probability better than $2^{-\Theta(n)}$.
\end{theorem}

We obtain the above result from a more general theorem about subset states. In particular, we show that subset
states are actually indistinguishable from Haar random states, provided their support is not too small
nor too big.

\begin{theorem}\label{thm:subset-state-PRS}
  Let $\cH=\C^d$ be a Hilbert space of dimension $d \in \mathbb{N}$, $\mu$ be the Haar measure on $\mathcal{H}$, and $S\subseteq [d]$ of size $s$. Then for any $k\in \mathbb{N}$,
  \begin{align*}
    \left\| \int{\psi^{\otimes k} d\mu(\psi)} - \mathop{\mathbb{E}}_{S\subseteq[d], |S|=s} \phi_S ^{\otimes k} \right\|_1 \le O\left(\frac{k^2}{d} + \frac{k}{\sqrt{s}} + \frac{s k}{d}\right),
  \end{align*}
  where  $\phi_S =  \left(\frac{1}{\sqrt{s}}\sum_{i\in S}\ket i\right)\left(\frac{1}{\sqrt{s}}\sum_{i\in S} \bra i \right)$.
\end{theorem}

The above theorem leads to a new construction of pseudorandom states (PRS), which is an important primitive with broad applications
in quantum cryptography~\cite{kretschmer2021quantum,kretschmer2023quantum}, resolving an open problem from the seminal work of Ji, Liu and Song
\cite{ji2018pseudorandom}. It also leads to a new construction of pseudoentangled states~\cite{aaronson2024quantum}. At the technical level,
the proof~\cref{thm:subset-state-PRS} goes via spectral graph theory by analyzing some matrices in the so-called
Johnson association scheme~\cite{delsarte1975association}.\footnote{These matrices also naturally arise in the study of complete high-dimensional expanders.}

Given that both classical and quantum property testing models fail spectacularly for our task, one can ask if there are other
approaches to property testing.

\begin{center}
  {\em How to go beyond the standard property testing models?}
\end{center}

Very much like $\NP$ enhances $\classP$ (and $\QMA$ enhances $\BQP$) with adversarial \emph{proofs}, one can enhance
the standard property testing models with proofs (or certificates, i.e., ``structured'' proofs in this paper), namely, additional adversarial information intended
to help testability. Here, we will consider the power and limitations of proofs and also interaction with provers
in the context of property testing.
One can imagine that a powerful untrustworthy entity prepares samples (or copies) together with certificates
so that a less powerful entity can be convinced of a property, ideally using substantially fewer resources.

Classically, we show that even with exponentially many samples and interacting with exponentially many independent public-coin $\AM$ provers
for exponentially many rounds, classical property testing still fails,

\begin{theorem}[Failure of Classical Testing with Certificates (Informal)]\label{theo:failure_am_testing}
  Even given $2^{\Omega(n)}$ samples of a classical flat distribution and interaction with $2^{\Omega(n)}$ $\AM$ provers in $2^{\Omega(n)}$ rounds,
  no classical tester can estimate the support size up to factors $2^{\Omega(n)}$ with probability better than $2^{-\Theta(n)}$.
\end{theorem}

At the heart of our proof of the above lower bound is a connection to fast mixing of high-dimensional
expanders~\cite{AJKPV22}. This lower bound technique is quite general and holds even given any promise (say intended to make
verification easier) on families of certifying distributions, which, in particular, captures the above public-coin $\AM$
lower bound with multiple independent provers with multiple rounds of interactions.

In fact, we show in the classical case, for distinguishing flat distributions of different support sizes, the proof provides no power---an optimal strategy
for the honest provers is just to provide more samples by proofs.

\begin{theorem}[Classical Certification Offer No Advantage (Informal)]\label{theo:tightness_observation}
  Given any public-coin $\AM$ protocol of communication cost of $p$ bits, let the sample complexity distinguishing flat distribution of support size $s$ and $2s$
  with a constant advantage be $t$. Let $t'$ be the sample complexity without proofs. Then, $p+t\ge \Omega( t')$.
\end{theorem}

In sharp contrast to the classical case, the presence of polynomially many flat adversarial certificates (\ie subset states) can
dramatically reduce the number of copies for a property to be testable showing that coherence can also be extremely powerful.

\begin{theorem}[Effective Quantum Certified Testing (Informal)]\label{theo:subset_advantage}
  With just polynomially many (\ie $n^{O(1)}$) copies and subset state proofs (\ie certificates of flat amplitudes), a tester can
  with high probability either approximate the support size of a subset state of arbitrary size, or detect that
  the certificates are malicious.
\end{theorem}

Note that the above model corresponds to a $\QMA$ type tester with structured proofs (subset states in this case). It is natural to ask
if the same result can be achieved with a general (adversarial) proof instead of assuming additional structure on the proofs. Surprisingly,
the answer is an emphatic no, and this holds for any quantum property whatsoever. More precisely, we show the following severe
information theoretic limitation on property testing with a general proof.

\begin{theorem}[Informal]\label{theo:limitation_prop_test}
  A general (adversarial) proof cannot improve quantum property testing. More precisely, a general quantum proof
  can be replaced by at most polynomially many extra input states.
\end{theorem}

\cref{theo:limitation_prop_test} is obtained using the de-Merlinization ideas of Aaronson~\cite{A06,HLM17}, so we do not claim technical novelty,
but rather just make explicit their surprising implication to quantum property testing with a general proof. We point out that this result is an information theoretic result since the above process of replacing a proof can incur an exponential increase in the running time of a tester. 

We also show how the study of property testing with structured proofs, sometimes even with very strong promises, can have interesting consequences to quantum-to-quantum state transformation lower bounds.
One example of a quantum-to-quantum state transformation lower bound that can be deduced from our work is the following.

\begin{theorem}[Hardness of Absolute Amplitudes Transformation (Informal)]\label{theo:hardness_of_abs_value_main}
  Any transformation that takes $k$ copies of an arbitrary $n$-qubit quantum state $\ket{\psi} = \sum_{x \in \set{0,1}^n} \alpha_x \ket{x}$ and produces
  a single $n$-qubit output state at least $0.001$ close to $\sum_{x \in \set{0,1}^n} \abs{\alpha_x} \ket{x}$ requires $k = 2^{\Omega(n)}$.
\end{theorem}

We also discuss how the study of quantum property testing has implications to disentangler lower bounds in~\cref{sec:hierarchy}. This gives a concrete
approach to attack Watrous' disentangler conjecture~\cite{ABDFS08}.


This paper is organized as follows.
In~\cref{subsec:additional_background}, we survey related work on quantum and classical property testing.
In~\cref{subsec:pseudorandomness}, we provide more details on the connection of our results with pseudorandomness and pseudoentanglement.
In~\cref{sec:prelim} and~\cref{sec:prop_test}, we recall basic notation and terminology, and formally introduce some of the property testing models studied in this work.
In~\cref{sec:strategy}, we give an overview on some of our technical results explaining, in particular, the role of the spectral analysis
in the Johnson scheme and the use of fast mixing of high-dimensional expanders mentioned above.
In~\cref{sec:quantum-indist}, we establish the limitations of coherence, proving~\cref{theo:failure_coherence} and~\cref{thm:subset-state-PRS}.
We also provide additional results on indistinguishability of ensemble of quantum states that is not via Haar random states.
In~\cref{sec:classical-lower-bound}, we show the limitations of testing classical distributions leading to~\cref{theo:failure_am_testing}
and~\cref{theo:tightness_observation}.
In~\cref{sec:prop_propqma_collapse}, we show that a general proof cannot improve quantum property testing leading to~\cref{theo:limitation_prop_test}.
In~\cref{sec:flat-cert}, we show that subset state proofs can substantially improve testability leading to~\cref{theo:subset_advantage}.
In~\cref{sec:hierarchy}, we conclude with a description of how quantum property testing complexity relates to standard computational complexity, and
to problems with quantum input.

\subsection{Related Work on Property Testing}\label{subsec:additional_background}

There is a vast body of work studying property testing both in the classical and quantum settings, e.g., see the surveys~\cite{MW16, G17, C20}.
In the quantum setting, property testing of oracles has been more extensively investigated.
There is by now a diverse and powerful set of query lower bound techniques which includes
the polynomial method~\cite{BBCMW01}, adversary method~\cite{A00}, generalized adversary
method~\cite{HLS07}, and others~\cite{AA18, AKKT20, ABKRT21}.
In the area of property testing of oracles, the analogous problem to testing support size is known as \emph{approximate counting}.
Approximately counting the weight of a classical oracle with a quantum $\QMA$ proof was considered by Aaronson \etal in~\cite{AKKT20}.
They focused on distinguishing oracle weight $w$ from $2w$ and established strong lower bounds using Laurent polynomials.
They also considered the setting without proofs, but with access to subset states encoding the support of the classical oracle
or access to a unitary that can produce it, also obtaining lower bounds.
Subsequently, Belovs and Rosmanis~\cite{BR20} established lower bounds for approximate counting in the setting without proofs in the
regime $w$ versus $(1+\epsilon)w$ for small $\epsilon \in (0,1)$.
In~\cite{DGRT22}, Dall'Agnol \etal investigate the power of adversarial quantum proofs in the study of property testing of unitaries.
More recently, Weggemans~\cite{W24} showed additional lower bounds for testing some unitary properties with proofs and advice.
The literature on property testing of quantum states rather than oracles (or unitaries) seems to be much sparser.

Property testing of classical probability distributions has been extensively studied~\cite{R10,V11,R12,C20,C22}.
The case of testability of support size of distributions, or closely related notions such as entropy, uniformity, essential
support, etc, are very natural, and they have been investigated in many works, including for flat distributions. Valiant
and Valiant's $\Theta(N/\log N)$ bounds on testing support size of general distribution of domain size $N$ is widely considered
as a cornerstone of the area~\cite{VV11}.
%

The notion of proofs is pervasive in theoretical computer science, and it was also studied
in many forms in property testing.
In~\cite{CG18}, Chiesa and Gur investigate the power of adversarial certificates for property
testing of probabilities distributions.
They considered analogues of $\NP$ and $\MA$, where certificates are bit strings. They also considered analogues of
single prover interactive proofs $\AM$ and $\IP$. They show that their corresponding $\NP$, $\MA$, and $\AM$ models can
at most provide a quadratic advantage in general, whereas $\IP$ can give an exponential improvement.
Following Chiesa and Gur's work, there are several works focusing on the upper bounds on property testing of general
distribution including support size for various interactive proof models~\cite{HR22label-inv, HR23IP}.
Here in our classical lower bounds, we consider multiple provers, whose certificates are probability distributions
satisfying any desired (convex) promise intended to help testability (see~\cref{sec:prop_test} for the precise details).
In particular, we consider interaction with multiple independent $\AM$ provers. Our technique based on fast mixing of
high-dimensional expanders yields tight lower bounds.

\subsection{Quantum Pseudorandomness from Our Results}\label{subsec:pseudorandomness}
We now explain the implications of~\cref{thm:subset-state-PRS} above for quantum pseudorandomness and pseudoentanglement.
We recall some of the context about these concepts along the way.

\paragraph{Pseudorandom States.} Pseudorandom quantum states (PRS) are a keyed family of quantum states that can be efficiently
generated and are computationally indistinguishable from Haar random states, even when provided with polynomially many copies.
PRSs have a wide range of applications including but not limited to statistically binding quantum bit commitments \cite{morimae2022quantum} and
private quantum coins \cite{ji2018pseudorandom}. Notably, for certain applications like private quantum coins, PRSs represent the weakest primitive
known to imply them.
Moreover, PRSs imply other quantum pseudorandom objects, such as one-way state generators (OWSGs) \cite{ji2018pseudorandom,morimae2022one} and EFI pairs
(\textbf{e}fficiently samplable, statistically \textbf{f}ar but computationally \textbf{i}ndistinguishable pairs of quantum states)
\cite{brakerski2022computational,morimae2022quantum}. Although all the existing PRS constructions rely on quantum-secure pseudorandom functions (PRFs)
or pseudorandom permutations (PRPs),  PRSs may be weaker than PRFs~\cite{kretschmer2023quantum}.

Since the initial proposal of pseudorandom states~\cite{ji2018pseudorandom}, various constructions have been investigated ~\cite{brakerski2019pseudo, ananth23binary, brakerski2020scalable, aaronson2024quantum,behera2023pseudorandomness}. 
Randomizing the phase was essential in the security proofs for all of these constructions. It is then natural to ask if it is possible to construct PRS without varying the phases, and indeed, Ji, Liu, and Song raised this question and conjectured that PRS can be constructed using \emph{subset states}.

Consider a  $n$-qubit system, represented by a $2^n$ dimensional Hilbert space.
For any function $t(n)=\omega(\poly(n))$ and $t(n) \le s \le 2^n/t(n)$ and $k=\poly(n)$, the distance between $k$ copies of a Haar random
state and $k$ copies of a random subset state of size $s$ is negligible as per the above theorem. This range for the subset's size is tight,
as otherwise  efficient distinguishers exist between copies of a Haar random state and a random subset state.

An immediate corollary of~\cref{thm:subset-state-PRS} above is the following:
\begin{corollary}[Pseudorandom States]\label{cor:sub-state-are-prs-informal} 
  Let $\{\PRP_k:[2^n]\to[2^n]\}_{k\in \cK}$ be a quantum-secure family of pseudorandom permutations. Then the family of states
  $
  \left\{\frac{1} {\sqrt{s}}\sum_{x\in {[s]}} \ket{\PRP_k(x)} \right\}_{k\in \cK}
  $
  is a PRS on $n$ qubits for $t(n) \le s \le 2^n/t(n)$ and any $t(n)=\omega(\poly(n))$.
\end{corollary}

Tudor and Bouland~\cite{bouland23subset} indepdentently discovered a subset state PRS construction. Their analysis uses representation theory.

\paragraph{Pseudoentanglement.} A closely related notion to the PRSs is that of \emph{pseudoentangled} states studied recently by~\cite{aaronson2024quantum}.
Here we call a PRS $h(n)$-pseudoentangled if for any state $\ket\phi$ from the PRS, $\ket\phi$ in addition satisfies that its entanglement entropy across all
cut is $O(h(n)).$\footnote{In \cite{aaronson2024quantum}, another notion was considered. Roughly speaking, a pseudoentangled state ensemble consists of two efficiently
constructible and computationally indistinguishable ensembles of states which display a gap in their entanglement entropy across all cuts.} Note that for a
Haar random state, the entanglement entropy is near maximal across all cuts~\cite{hayden2006aspects}. It's observed in~\cite{aaronson2024quantum} that a subset state with respect to set
$S$ has entanglement entropy at most $O(\log |S|)$ across any cut for some function $h(n):\mathbb{N}\to\mathbb{N}$, since the Schmidt rank of a subset state is
at most $|S|$ across any cut. Therefore the subset states with small set size are good candidates for pseudoentangled states, which they left as an open problem. As a corollary to our PRS result regarding subset state, we resolve this open problem.

\begin{corollary}[Pseudoentangled States]
  Let $\PRP_k:[2^n]\to[2^n]$ be a quantum-secure family of pseudorandom permutations. For any $h(n)= \omega(\log n)$ and $h(n)=n-\omega(\log n)$,
  we have the following $h(n)$-pseudoentangled state from subset state of size $s=2^{h(n)}$,
  $$
  \left\{\frac{1}{\sqrt{s}}\sum_{x\in {[s]}} \ket{\PRP_k(x)}\right\}_{k\in\cK}\mper
  $$
\end{corollary}

For a PRS, it is easy to see that if for some cut the entanglement entropy of a state $\ket\phi$ is $O(\log n)$, then $\ket\phi$ can be distinguished from Haar random
states with polynomially many copies using swap test. In this sense, the above pseudoentangled state construction is optimal.

%% file: prelim.tex
\section{Preliminaries}\label{sec:prelim}
\paragraph{General.}
We adopt the Dirac notation for vectors representing quantum states, e.g., $\ket\psi, \ket\phi$, etc. All the vectors of the form $\ket\psi$ will be unit vectors. Given any pure state $\ket\psi$, we adopt the convention that its density operator is denoted by the Greek letter without the ``ket'', e.g. $\psi = \ket\psi\bra\psi$. The set of density operators in an arbitrary Hilbert space $\cH$ is denoted $\fD(\cH)$, and the set of pure states is denoted by $\fP(\cH)$. For a mixed state denoted by capital letters, e.g., $\Psi,\Phi$, we quite often treat it as a \emph{set} of states together with some underlying distribution on the set. For mixed state $\Psi$, there can be many different ways to express it as a distribution on pure states. Normally, we fix some explicit set that will be clear from the context.
So $\Psi,\Phi$ will have both the set interpretation and the density matrix interpretation. The Haar measure is referred to the uniform measure on the unit sphere of $\C^d$. For Hermitian matrices $A, B$, we adopt the notation $A\preceq B$ or $B\succeq A$ to denote the Loewner order, i.e., 
\[
    A\preceq B ~\iff~ B\succeq A ~\iff~ \langle \psi | B-A |\psi \rangle \ge 0, \forall \psi.
\]

Given any matrix $M\in \C^{n\times n}$ denote by $\|M\|_1$ the trace norm, which is the sum of the singular values of $M$. The trace distance between two quantum states $\sigma,\rho$, denoted $\TD(\sigma, \rho): = \|\sigma-\rho\|_{\trace}$ is $\|\rho-\sigma\|_1/ 2$. Two states with small trace distance are indistinguishable to quantum protocols.
\begin{fact}\label{fact:trace_norm_acc}
If a quantum protocol accepts a state $\phi$ with probability at most $p$, then it accepts $\psi$  with
  probability at most $p + \|\phi-\psi\|_{\trace}$.
\end{fact}
We write $x\lesssim y$ to denote that there is a small constant $c\ge 1$, such that $x\le c y$. 
For any set $S$, let
$A(S,k):=\{(i_{1},i_{2},\ldots,i_{k})\in S{}^{k}:i_{j}\not=i_{j'}\text{ for }j\not=j'\}.$ We also adopt the notation $S_n$ for the symmetric group. For two disjoint sets $A, B$, we use $A\sqcup B$ to denote their union, emphasizing that $A$ and $B$ are disjoint. 

Let $n^{\underline{k}}:=n(n-1)\cdots(n-k+1).$ 
A simple calculation reveals that for $k=O(\sqrt n)$,
\begin{equation*}
    \frac{n^{\underline k}}{ (n-k+1)^{\underline k}} -1 \le  \left(\frac{n+k-1}{n-k+1}\right)^k -1
    =
    O\left(\frac{k^2}{n}\right).
\end{equation*}
We will use this bound without referring to this calculation. Given any pure quantum state $\rho$ over systems $A,B$, the entanglement entropy over the cut $A:B$ is the von Neumann entropy of the reduced density matrix of system $A$ (or $B$), i.e., $-\trace (\rho_A \log \rho_A)$.

\paragraph{Entropy and KL-divergence.}
Consider some discrete space $\Omega$ and a probability measure $\gamma$ over $\Omega$.
If the random variable $X$ is drawn from $\gamma$, we denote it
by $X\sim\gamma$. 
We let $\ln x$ and $\log x$ stand for the natural logarithm of $x$
and the logarithm of $x$ to base $2,$ respectively. For any distribution
$\gamma$ over some discrete space $\Omega$, the entropy function
\[
H(\gamma)=\Exp_{x\in\Omega}\gamma(x)\log\frac{1}{\gamma(x)}.
\]
Recall that the Kullback-Leibler divergence (KL-divergence) between
two distributions $\mu_{0},\mu_{1}$ over $\Omega$ is defined by
the following formula
\[
\KL{\mu_{0}}{\mu_{1}}=\sum_{x\in\Omega}\mu_{0}(x)\log\frac{\mu_{0}(x)}{\mu_{1}(x)}.
\]
If two random variables $X_{0},X_{1}$ obey $\mu_{0}$ and $\mu_{1}$,
respectively, we also use $\KL{X_{0}}{X_{1}}$ to denote the KL-divergence
between the two distributions. 
The KL-divergence satisfies the following
chain rule:
\[
\KL{X_{0}Y_{0}}{X_{1}Y_{1}}=\KL{X_{0}}{X_{1}}+\Exp_{x\sim X_{0}}\left[\KLfrac{Y_{0}\mid X_{0}=x}{Y_{1}\mid X_{1}=x}\right].\footnotemark
\]
\footnotetext{Here, we use the fraction-like notation to also denote the
KL-divergence for aesthetics, as we are comparing two conditional distributions.
The numerator in the fraction-like notation corresponds to the first argument
in the standard notation.}

\paragraph{Flatness.}
A distribution $\mu$ is called \emph{flat} if it is uniform over its support. There is one natural way to encode a classical distribution over a discrete domain $X$ into a pure quantum state over the Hilbert space $\C^X$. That corresponds to the so-called \emph{subset state}.
\begin{definition}[Subset States/Flat States]\label{def:subset_flat}
  We say that $\ket{\psi} \in \C^d$ is a subset state (or, equivalently, a flat state) if $\ket\psi$ is the uniform superposition over some subset $S \subseteq [d]$,
  \[\ket{\psi} = \frac{1}{\sqrt{\abs{S}}} \sum_{i\in S} \ket{i}.\]  
\end{definition}

\section{Property Testing Models}\label{sec:prop_test}
We now formally state some definitions and models for property testing we will use. Note that the definitions are made focusing on information theoretic measures like sample complexity. Some remarks on time complexity will be discussed in the final section.

\paragraph{Classical Property Testing Models.}
We start with the classical setting. In the property testing for classical distribution, one is given an unknown probability distribution $\nu$ where the only way to access $\nu$ is to draw independent samples. The goal is to test whether the unknown distribution $\nu$ satisfies certain property, e.g., whether $\nu$ has large support size or not. 

Let $\Delta_d$ be the probability simplex in $\R^d$, \ie 
\[
    \Delta_d := \left\{(p_1,\ldots,p_d) \in \R^d :\, \sum_{i=1}^d p_i = 1 \textup{ and }  p_1,\ldots, p_d \ge 0 \right\}.
\]
Analogously, for a finite set $S$, we denote by $\Delta_{S}$ the probability simplex in $\R^S$.
Recall that a property of classical probability distributions is defined as follows.

\begin{definition}[Property of Classical Distributions]\label{def:cprop}
  A property is any family of probability distributions $\calP = \sqcup_d \calP_d$ where $\calP_d \subseteq \Delta_d$. 
\end{definition}

\begin{definition}[Standard Classical Property Testing Model]
  For $d\in\mathbb N$, let $k=k(d): \mathbb{N} \to \mathbb{N}$, $1 \ge a > b \ge 0$. A property $\calP = \sqcup \calP_d$ belongs to $\propBPP_{a,b}[k]$
  if there exists a verifier $V$ such that for every $\nu\in \Delta_d$,
  \begin{enumerate}
    \item if $\nu \in \calP_d$, then $V$ with $k$ independent samples from $\nu$ accepts with probability at least $a$, and
    \item if $\nu\in \Delta_d$ is $\epsilon$-far in statistical distance from $\calP_d$, then $V$ with $k$ independent samples from $\nu$ accepts with probability at most $b$.
  \end{enumerate}
\end{definition}


Even though we use $\propBPP$ in our notation for the model,  we just take the ``bounded-error'' and ``probabilistic'', there is nothing polynomially bounded here. The sample complexity can be unbounded and the running time can be unbounded. This is a somewhat common abuse of notation to avoid introducing too many notations.

Next, we proceed to discuss the classical property testing model enhanced with classical certificates. 
In the most standard notion of the Merlin-Arthur ($\MA$) type proof, fix some property $\calP$, one expect that for any $\nu\in \calP_d$, there is an honest proof $\pi\in \{0,1\}^{p}$ that convince the verifier to accept with high probability after making $k$ samples. On the other hand, for $\nu\in \Delta_d$ far from $\calP$, no proof should fool the verifier to accept with high probability. 

\begin{definition}[Classical Property Testing with $\MA$ Proofs]\label{def:propMA-model}
  For $d\in \mathbb N$, let $k=k(d),p=p(d) : \mathbb{N} \to \mathbb{N}$, $1 \ge a > b \ge 0$. A property $\calP = \sqcup \calP_d$ belongs to $\propMA_{a,b}[k,p]$ 
  with respect to certificates set $\calS=\{0,1\}^p$ if there exists a verifier $V$ such that
  for every $\nu \in \Delta_d$,
  \begin{enumerate}
  \item if $\nu \in \calP_d$, then there exist $\pi \in \calS$ such that
    \begin{align*}
      \Pr_{x_1,\ldots,x_k \sim \nu^{\otimes k}}\left[V(x_1,\ldots,x_k,\pi) \textup{ accepts} \right] \ge a\mcom
    \end{align*}   
    \item if $\nu$ is $\epsilon$-far from $\calP_d$ in statistical distance, then for every $\pi \in \calS$,
    \begin{align*}
      \Pr_{x_1,\ldots,x_k \sim \nu^{\otimes k}}\left[V(x_1,\ldots,x_k,\pi) \textup{ accepts} \right] \le b\mper
    \end{align*}
  \end{enumerate}
\end{definition}

A much stronger notion usually referred to as the public-coin Arthur-Merlin ($\AM$) model, in its greatest generality, involves $m$ provers and $r$ rounds of communication. In each round, Arthur sends $m$ independently uniformly random bit strings of length $p$ to each of the $m$ Merlins who have no information about each other's communication, and each Merlin responds with a single bit (without loss of generality). 
After the $r$ rounds of communication, Arthur should be able to decide whether the unknown distribution $\nu\in\calP$ or far from it. 
Any such $\AM$ protocol $\Pi$ gives rise to some distribution  on the communication transcript $\pi\in \set{0,1}^{rm(p+1)}$ between Arthur and Merlins, satisfying that in each round, the random bits send are uniformly random; and the communications with each Merlin are independent of each and only depends on communication history between the current Merlin and Arthur. 
Fix any valid $\AM$ protocol $\Pi$, we can let $\mu_\Pi\in \Delta_{\set{0,1}^{rm(p+1)}}$ be the distribution on the communication transcript generated by $\Pi.$

\begin{definition}[Classical Property Testing with $\AM$ Proofs]\label{def:propAM-model}
  For $d\in \mathbb N$, let $k=k(d),m = m(r), p=p(d),r=r(d) : \mathbb{N} \to \mathbb{N}$, $1 \ge a > b \ge 0$. A property $\calP = \sqcup \calP_d$ belongs to $\propAM(m, r)_{a,b}[k, r(mp+m)]$,
  \begin{enumerate}
  \item if $\nu \in \calP_d$, there exists a $r$-round, $m$-prover, $\AM$-protocol $\Pi$, where in each round Arthur can send $p$ uniformly random bits and get 1 bit answer, such that
    \begin{align*}
      \Pr_{x_1,\ldots,x_k \sim \nu^{\otimes k}, \pi \sim \mu_\Pi}\left[V(x_1,\ldots,x_k,\pi) \textup{ accepts} \right] \ge a\mcom
    \end{align*}   
    \item if $\nu$ is $\epsilon$-far from $\calP_d$ in statistical distance, for any $r$-round, $m$-prover, $\AM$-protocol $\Pi$, where in each round Arthur can send $p$ uniformly random bits and get 1 bit answer, 
    \begin{align*}
      \Pr_{x_1,\ldots,x_k \sim \nu^{\otimes k}, \pi \sim \mu_\Pi}\left[V(x_1,\ldots,x_k,\pi) \textup{ accepts} \right] \le b\mper
    \end{align*}
  \end{enumerate}
\end{definition}

Analogously, we can define the private-coin Arthur Merlin model, whose more familiar name is the interactive proof system (IP). Consider $\propIP(m,r)_{a,b}[k,p]$ which is like $\propAM(m,r)_{a,b}[k,p]$, wherein the verifier has private random coins. Consequently, in each round, the verifier's message can depend on the private coin and the communication transcript so far. Abbreviate
$\propIP_{a,b}[k,p] = \propIP(1,\poly(n))_{a,b}[k,p]$, the standard single prover, polynomial-rounds interactive proofs.

The principle of our notations is that the parenthesis $(~)$ are for the parameters inherent to the proof system, including number of provers and number of rounds of communication allowed, in the given order; while the square brackets $[~]$ are for the complexity-related parameters including number of samples, and the total communications between the provers and the verifier, in the given order order; and finally the subscripts are for the completeness and soundness parameters in the given order, that are very sometimes omitted for completeness being $2/3$ and soundness being $1/3$.

\paragraph{Quantum Property Testing Models.}
We now move to the quantum setting. Recall that $\fP(\C^d)$ denotes the set of pure states in $\C^d$. First, we review the notion of properties of quantum states.

\begin{definition}[Property of Quantum States]\label{def:qprop}
  A property is any family of subsets $\calP = \sqcup_d \calP_d$ where $\calP_d \subseteq \fP(\mathbb{C}^d)$.
\end{definition}

The standard quantum property testing model is defined as follows. 

\begin{definition}[Standard Quantum Property Testing Model]\label{def:bqp-prop}
  For $d\in\mathbb N$, let $k=k(d): \mathbb{N} \to \mathbb{N}$, $1 \ge a > b \ge 0$. A property $\calP = \sqcup \calP_d$ belongs to $\propBQP_{a,b}[k]$
  if there exists a verifier $V$ such that for every $\ket{\psi} \in \mathbb{C}^d$,
  \begin{enumerate}
    \item if $\ket{\psi} \in \calP_d$, then $V(\ket{\psi}^{\otimes k})$ accepts with probability at least $a$, and
    \item if $\ket{\psi}$ is $\epsilon$-far from $\calP_d$ in trace distance, then $V(\ket{\psi}^{\otimes k})$ accepts with probability at most $b$.
  \end{enumerate}
\end{definition}

This model can be enhanced with certificates as follows.

\begin{definition}[Quantum Property Testing with Certificates]\label{def:qma-prop}
  For $d\in \mathbb N$, let $k=k(d),p=p(d) : \mathbb{N} \to \mathbb{N}$, $1 \ge a > b \ge 0$. A property $\calP = \sqcup \calP_d$ belongs to $\propQMA(m)_{a,b}[k,mp]$
  if there exists a verifier $V$ such that for every $\ket{\psi} \in \mathbb{C}^d$,
  \begin{enumerate}
    \item if $\ket{\psi} \in \calP_d$, then there exist $m$ certificates $\ket{\phi_1},\ldots,\ket{\phi_m} \in \mathbb{C}^{2^p}$ such that $V(\ket{\psi}^{\otimes k}\otimes \ket{\phi_1} \otimes \cdots \otimes \ket{\phi_m})$ accepts with probability at least $a$, and
    \item if $\ket{\psi}$ is $\epsilon$-far from $\calP_d$ (in trace distance), then then for every $\ket{\phi_1},\ldots,\ket{\phi_m} \in \mathbb{C}^{2^p}$, $V(\ket{\psi}^{\otimes k}\otimes \ket{\phi_1} \otimes \cdots \otimes \ket{\phi_m})$
          accepts with probability at most $b$.
  \end{enumerate}
\end{definition}

We will also consider the promised version of these models in which we have a pair of properties $(\calP_{\textup{YES}}, \calP_{\textup{NO}})$. Then $\calP_{\textup{YES}}$ will be the property of interest, and the item (ii) in the above definitions will only care about distributions or states from $\calP_{\textup{NO}}$. For simplicity, one can think of $\calP_{\textup{NO}}$ as a subset of distributions or states $\epsilon$-far from $\calP_{\textup{YES}}$. In reality, $\calP_{\textup{NO}}$ is often relaxed so that with respect to some probability measure $\mu$, $\calP_{\textup{NO}}$ has measure close to $1$ being $\epsilon$-far from $\calP_{\textup{YES}}$.

%% file: strategy.tex
\usetikzlibrary{matrix, decorations.pathreplacing}
\usetikzlibrary{shapes, arrows.meta, positioning, calc, shapes.geometric, matrix}

\section{Technical Overview of Our Results}\label{sec:strategy}

Our results intersect many fields: classical and quantum property testing; the power and limitations of classical and quantum proofs;
the nature of classical versus quantum information via the role of coherence; cryptography via PRSs and entanglement via
pseudoentanglement. At the heart of some of our technical results are connections to spectral graph theory via the Johnson scheme
and also to fast mixing of high-dimensional expanders. We now give an overview of some of these connections, and precise technical
details are left to the relevant sections.

\subsection{Johnson Scheme and Limitations of Coherence}

We will now discuss how subset states of the form $\ket{\phi_S}$, devoided of negative and imaginary phases by definition,
can be indistinguishable from Haar random states over the sphere $\C^d$. This indistinguishability happens whenever the
subset size $\abs{S}$ is not too small nor too large. In turn, this implies that subset states of vast different support
sizes are indistinguishable since they are both indistinguishable from Haar random states.

The key technical contribution is realizing a connection to the so-called Johnson scheme and conducting a spectral analysis
using it to obtain the above indistinguishability. Recall that the matrices of the Johnson scheme $\mathcal{J}([d],k)$ have rows
and columns indexed by sets from $\binom{[d]}{k}$. Moreover, given a matrix $\calD$ on the scheme, each entry $\calD(A,B)$ 
only depends on the size of the intersection $A \cap B$. For $t \in \set{0,1,\ldots,k}$, one defines a basis matrix $\calD_t$,
whose rows and columns are indexed by elements in $\binom{[d]}{k}$ as follows
\begin{align*}
  \calD_t(A,B) = \begin{cases}
                  1   & \textup{if } \abs{A\cap B} = t\\
                  0   & \textup{otherwise}
                 \end{cases}
\end{align*}
for every $A, B \in \binom{[d]}{k}$. It follows that any matrix $\calD$ in the Johnson scheme can be decomposed as
$\calD = \sum_{t=0}^d \alpha_t \calD_t$, where each $\alpha_t$ is a scalar. This association scheme has a variety of remarkable
properties,\footnote{It forms a commutative algebra of matrices under addition and matrix multiplication.} but we will be
mostly concerned with the spectral properties of $\calD_t$.

Now, we proceed to give an idea of how the Johnson scheme arises in the analysis. The starting point is the well-known fact that
the expectation over Haar random states $\int \psi^{\otimes k} d \mu$ is equal to\footnote{A normalized version of the projector onto the symmetric subspace.}
\begin{align*}
 \binom{d+k-1}{k}^{-1}\frac{1}{k!}\sum_{\pi\in S_{k}}\sum_{\vec{i}\in[d]^{k}}|\vec{i}\rangle\langle\pi(\vec{i})| \approx \binom{d+k-1}{k}^{-1}\frac{1}{k!}\sum_{\pi\in S_{k}}\sum_{\vec{i}\in A([d],k)}|\vec{i}\rangle\langle\pi(\vec{i})| \eqqcolon \tilde{\Psi}, 
\end{align*}
where the last approximation assumes $k \ll \sqrt{d}$. Although the matrix on the RHS is not in the Johnson scheme, note that its only non-zero entries have a fixed value and
occur on entries indexed by row $(i_1,\ldots,i_k)$ and column $(j_1,\ldots,j_k)$ if and only if $\abs{\set{i_1,\ldots,i_k} \cap \set{j_1,\ldots,j_k}} = k$. This means that
this matrix can be written as $\calD_k \otimes J$, where $J$ is a $k! \times k!$ all-ones matrix.

Next, we turn our attention to uniform average of subset states of a fixed size $s$. As before, it will be convenient to work with tuples
of distinct indices as
\[
\tilde{\Phi} = \Exp_{S:|S|=s}\left[\frac{1}{s^{\underline{k}}}\sum_{\vec{i},\vec{j}\in A(S,k),}|\vec{i}\rangle\langle\vec{j}|\right]\mcom
\]
and the approximation error is small provided $k \ll \sqrt{s}$. Note that
\begin{align}
  \tilde{\Phi}((i_1,\ldots,i_k),(j_1,\ldots,j_k)) & =\frac{1}{s^{\underline{k}}}\Pr_{|S|=s}[\vec{i},\vec{j}\in A(S,k)]=\frac{1}{s^{\underline{k}}}\frac{\binom{d-2k+t}{s-2k+t}}{\binom{d}{s}}\mcom
\end{align}
where $t = \abs{\set{i_1,\ldots,i_k} \cap \set{j_1,\ldots,j_k}}$. Since $s$ and $k$ are fixed, this means that the entries only depend on the size of the intersection of the set of elements in the tuples,
and thus $\tilde{\Phi} = \sum_{t=0}^k \alpha_t \calD_t \otimes J$, for scalars $\alpha_t$. To compute the trace distance between $\tilde{\Psi}$ and $\tilde{\Phi}$ in order to conclude that these states
are close, we rely on the spectral properties of the Johnson scheme. We also show the indistinguishability of some ensembles with dense support (see~\cref{subsec:dense_indist}), but this time, it is not
via indistinguishability from Haar.

\subsection{Fast Mixing of High-dimensional Expanders and Classical Limitations}

Suppose our goal is to distinguish flat distributions of support size $s$, the yes-case,
from those with support size $w \gg s$, the no-case. Suppose these distributions are on $[N]$,
where $N = 2^n$. We can imagine that we have a complete simplicial complex $X = \cup_{i=1}^w X(i)$,
with $X(i) = \binom{[N]}{i}$. Taking $t$ independent samples from a flat distribution with support
$S \in X(s)$ is approximately the same as taking a uniform subset of size $t$ from $S$ provided
$t \ll \sqrt{s}$. In the yes-case, each flat distribution, represented by a set $S$, has an associated certifying
distribution $\pi_S$. We can think that we choose a set $S \in X(s)$ uniformly at random with
its corresponding certificate. This naturally gives rise to a pair of coupled random variables
$(\mathcal{S}=S,\Pi=\pi_S)$. Now sampling from $\Pi$ induces a conditional distribution on $\mathcal{S} \vert \Pi$;
equivalently, we have a distribution $\mu$ on $X(s)$. This possibly very sparsely-supported distribution
corresponds to how much we learned from the proof. Now, we take a uniformly random subset $T$ of size $t$ from
$\calS$. The connection to fast mixing of high-dimensional expanders now emerges. We recall that the complete
simplicial complex $X$ is a very strong HDX. A well known random walk on $X$ is the Down walk $D_{i\to i-1}$, defined for every
$i \in \set{2,\ldots,w}$, as walk operator from $X(i)$ to $X(i-1)$
\begin{align*}
   D_{i\to i-1}(A,B) = \begin{cases}
                \frac{1}{i} & \textup{ if } A \subseteq B\\
                0           & \textup{ otherwise}
              \end{cases}
\end{align*}
In this language, observe that $T$ is distributed as $\mu D_{s \to s-1} D_{s-1 \to s-2} \ldots D_{t+1 \to t}$. The closeness of $T$ to
uniform on $X(t)$ is given by how fast the down random walk mixes starting with distribution $\mu$ on $X(s)$.
Roughly speaking, since sampling $T$ as above is statistically similar to sampling a uniform set from $X(t)$,
this means that the certificate was not very informative, and this will allow us to deduce lower bounds on distinguishing the
yes, and no cases. We illustrate this fast mixing from $\mu$ on $X(s)$ to close to uniform on $X(t)$ in~\cref{fig:fast_mixing}.

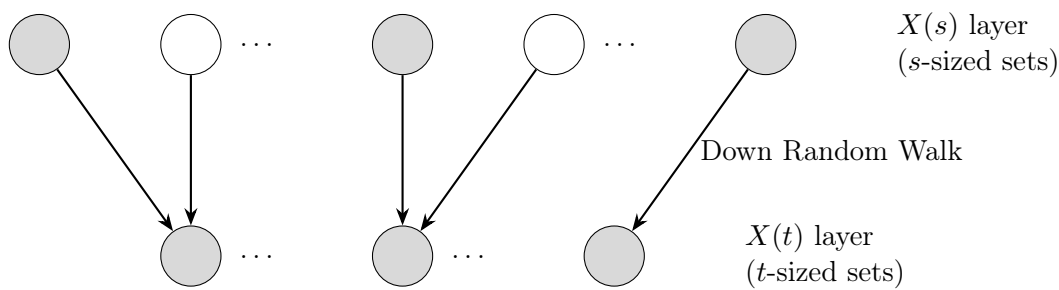
\begin{figure}[h!]
\begin{tikzpicture}[set/.style={circle, draw, minimum size=0.8cm}, 
                    grayset/.style={circle, draw, fill=gray!30, minimum size=0.8cm},
                    arrow/.style={-Stealth, thick},
                    node distance=1.2cm] 

  \node[grayset] (S1) {};
  \node[set] (S2) [right = of S1] {};
  \node at ($(S2.east)+(0.5,0)$) {\ldots}; 
  \node[grayset] (S3) [right = 2cm of S2] {}; 
  \node[set] (S4) [right = of S3] {};
  \node at ($(S4.east)+(0.5,0)$) {\ldots}; 
  \node[grayset] (S5) [right = 2cm of S4] {}; 

  \node[grayset] (T1) [below = 2cm of S2] {}; 
  \node at ($(T1.east)+(0.5,0)$) {\ldots}; 
  \node[grayset] (T2) [right = 2cm of T1] {}; 
  \node at ($(T2.east)+(0.5,0)$) {\ldots}; 
  \node[grayset] (T3) [right = 2cm of T2] {}; 

  \draw[arrow] (S1) -- (T1); 
  \draw[arrow] (S2) -- (T1);
  \draw[arrow] (S3) -- (T2);
  \draw[arrow] (S4) -- (T2);
  \draw[arrow] (S5) -- (T3) node[midway, right] {Down Random Walk};

  \node[align=left, right=of S5] {$X(s)$ layer\\($s$-sized sets)};
  \node[align=left, right=of T3] {$X(t)$ layer\\($t$-sized sets)};
\end{tikzpicture}
\caption{Mixing to uniform measure on $X(t)$ with Down random walk starting from $X(s)$. Gray vertices on top indicate support of initial measure on $X(s)$,
         whereas gray vertices on the bottom indicate the support on $X(t)$.}\label{fig:fast_mixing}
\end{figure}

It is easy to see that if $\mu$ was just a delta distribution on a single set $S$, then mixing cannot happen.
Therefore, we need to also ensure that sampling $\Pi$ leaves us with enough entropy on $\mu$ so that mixing happens.
In trying to help the verification protocol as much as possible, we can assume that the certifying distributions 
only come from a desired promised convex set, and this proof technique is oblivious to this choice. In particular, we
can consider that the certifying distributions come from multiple $\AM$ provers and the lower bound still
applies.

\subsection{Coherence Strikes Back via Certificates}

As discussed above, coherence alone is not enough to imply distinguishability between subset states of
vastly different support size. In contrast, coherence in the form of additional (adversarial) certifying subset
states will enable us to give a multiplicative approximation to the support size of an arbitrary subset state
$\ket{\psi_S}$ regardless of the size of $S \subseteq \set{0,1}^n$.

An honest prover will consider an arbitrary sequence of nested sets such that $\set{0,1}^n = S_0 \supseteq S_1 \supseteq S_2 \supseteq \cdots \supseteq S_\ell = S$
and $\abs{S_{i}}/\abs{S_{i-1}} = 1/2$ for every $i \in [\ell]$ (we assumed that $\abs{S}$ is a power of two for convenience). For each $S_i$, the prover will
send multiple copies of $\ket{\phi_{S_i}}$ and $\ket{\phi_{S_i^c \cap S_{i-1}}}$. We show that with multiple purported copies of $\ket{\phi_{S_{i-1}}}$, $\ket{\phi_{S_i}}$,
$\ket{\phi_{S_i^c \cap S_{i-1}}}$ it is possible to test if
$$
\frac{\abs{S_i}}{\abs{S_{i-1}}} ~=~ \frac{1}{2} \pm \delta \mcom
$$
or reject if the prover is dishonest. We illustrate the expected nested sequence of subset states  below.

\begin{figure}[h!]
\centering
\scalebox{.6}{
\begin{tikzpicture}[
    node distance=0.8cm and 0.8cm, 
    every node/.style={
        circle, 
        draw, 
        align=center,
        inner sep=1pt,
        text width=1.6cm, 
        font=\small\sffamily,
    },
    arrow/.style={-Stealth}
]

\node (root) {$\ket{\phi_{S_0}}$};
\node[below left=of root] (l1) {$\ket{\phi_{S_1}}$};
\node[below left=of l1,draw=none] (l2) {\reflectbox{$\ddots$}};
\node[below left=of l2] (l3p) {$\ket{\phi_{S_{\ell-1}}}$};
\node[below left=of l3p] (l3) {$\ket{\phi_{S_\ell}}$};
\node[below right=of l3p] (r2) {$\ket{\phi_{S_{\ell}^c \cap S_{\ell-1}}}$};
\node[below right=of l1] (r1) {$\ket{\phi_{S_2^c \cap S_1}}$};
\node[below right=of root] (r0) {$\ket{\phi_{S_1^c \cap S_0}}$};

\draw[arrow] (root) -- (l1);
\draw[arrow] (l1) -- (l2);
\draw[arrow] (l2) -- (l3p);
\draw[arrow] (l3p) -- (l3);
\draw[arrow] (root) -- (r0);
\draw[arrow] (l1) -- (r1);
\draw[arrow] (l3p) -- (r2);
\end{tikzpicture}
}
\end{figure}

Coherence allows us to recursively apply this with enough control on the errors so that we obtain the telescoping conclusion
\begin{align*}
  \abs{S} ~=~  \abs{S_0} \frac{\abs{S_1}}{\abs{S_0}} \frac{\abs{S_2}}{\abs{S_1}} \cdots \frac{\abs{S_\ell}}{\abs{S_{\ell-1}}}  ~=~ \left(\frac{1}{2} \pm \delta \right)^\ell 2^n\mper
\end{align*}
Since $t$ is at most $n$, we obtain a non-trivial multiplicative approximation to $\abs{S}$ with $\delta=1/\poly(n)$.

%% file: q-indist.tex
\section{Testing without Certificates: A Fiasco}\label{sec:quantum-indist}
In this section, we illustrate some natural examples of untestable quantum properties.

\subsection{The Lower Bound Meta-Technique}
To start, we review the generic lower bound technique on the power of a tester in distinguishing two
collections of quantum states $\calA$ and $\calB$,  representing states with and without certain abstract properties, respectively. If there is a tester that distinguishes
between any state in $\calA$ from any state in $\calB$, then by linearity, it should
distinguish any distributions over states from $\calA$ and $\calB$. In other words, it
implies the distinguishability of ensembles. By the contrapositive, if the distinguishability of
ensembles fails for any pair of ensembles, this means that there is no tester
distinguishing $\calA$ and $\calB$. This meta-technique for indistinguishability
is standard in the quantum as well as the classical setting.\footnote{Suitably stated for probability distributions
in the classical setting.} The key innovations are the ensembles for which we show
indistinguishability results.

\begin{lemma}[Pointwise Distinguishability Implies Ensemble Distinguishability]\label{lem:point_dist_ens_dist}
Let $\calA,\calB \subseteq \C^d$ and $1 \ge a > b \ge 0$. If there exists a measurement $M$ such that
\begin{align*}
    &\forall \ket{\phi} \in \calA,~ \Tr(M \phi^{\otimes k}) \ge a,  \text{ and,}
    \\
    &\forall \ket{\phi} \in \calB,~ \Tr(M \phi^{\otimes k}) \le b. 
\end{align*}
Then for any distributions $\mu_{\calA},\mu_{\calB}$ on $\calA$ and $\calB$, respectively, we have 
  \begin{align*}
    \Tr(M \rho_\calA) \ge a \quad \textup{ and } \quad \Tr(M \rho_\calB) \le b\mcom
  \end{align*}
  where $\rho_\calA = \E_{\mu_\calA} \phi^{\otimes k}$ and $\rho_\calB = \E_{\mu_\calB} \phi^{\otimes k}$.
\end{lemma}

\begin{proof}
  By linearity, we have
  \begin{align*}
    \Tr(M \rho_\calA) = \E_{\mu_\calA} \Tr(M \phi^{\otimes k}) \ge a,
  \end{align*}
  where the last inequality follows from the assumption. An analogous computation establishes that
  $\Tr(M \rho_\calB)  \le b$, concluding the proof.
\end{proof}

By the contrapositive of~\cref{lem:point_dist_ens_dist}, one obtains the following  lemma asserting that it suffices to find indistinguishable ensembles to rule out the existence of a property tester.

\begin{lemma}[Ensemble Indistinguishability Implies Pointwise Indistinguishability]\label{lem:ens_indist_point_indist}
  Let $\calA,\calB \subseteq \C^d$. If there exist distributions $\mu_{\calA},\mu_{\calB}$ on $\calA$ and $\calB$, respectively,
  such that  
  \begin{align*}
    \norm{\rho_\calA -  \rho_\calB}_1 < \epsilon,
  \end{align*}
  where $\rho_\calA = \E_{\mu_\calA} \phi^{\otimes k}$ and $\rho_\calB = \E_{\mu_\calB} \phi^{\otimes k}$.
  Then there is no measurement $M$ satisfying
  \begin{align*}
    \Tr(M \phi^{\otimes k}) \ge a,~ \forall \ket{\phi} \in \calA, \quad \textup{ and } \quad \Tr(M \phi^{\otimes k}) \le b, ~\forall \ket{\phi} \in \calB,
  \end{align*}
  with $b-a \ge \epsilon$.
\end{lemma}

\subsection{Warm-ups: Product States v.s. Nonproduct States}
Here, we show that testing \emph{productness}, i.e., given a state $\ket\psi$ determine if it's a (multipartite) product state or $(1-\epsilon)$-far from being a product state, is impossible.
Based on the previous discussion, to rule out a property tester, one needs to come up with indistinguishable ensembles.  For example, consider the following two ensembles:
\begin{align*}
    \cE_1^k &= \{\ket\psi^{\otimes k}: \ket\psi \in \fP(\C^d) \}; 
    \\
    \cE_0^k &= \Big\{\frac{1}{\sqrt {k!}}\sum_{\pi \in S_k} \ket{\psi_{\pi(1)}} \cdots \ket {\psi_{\pi(k)}}: \ket{\psi_{1}},\ldots, \ket {\psi_{k}} 
    \\
    &\qquad\qquad\qquad\qquad\text{ are the first $k$ columns of a Haar random unitary } U \in \C^{d\times d} \}.
\end{align*}
Note that states from the first ensemble are $k$-partite product, while states from the second ensemble have small overlap with any product state. We consider the Haar measure on states for $\cE_1^k$ 
and the Haar measure on unitaries for $\cE_0^k$.
Note that $\Exp _{\psi\in\cE_1^k} \psi$ and $\Exp_{\psi\in\cE_0^k} \psi$ are both invariant under the action of $U^{\otimes k}$ for any unitary $U$
since the corresponding Haar measures are invariant.
By Schur's lemma from representation theory~\cite{serre1977linear}, we have
\begin{align*}
    \Exp _{\psi\in\cE_1^k} \psi = \Exp_{\psi\in\cE_0^k} \psi.
\end{align*}
By~\cref{lem:ens_indist_point_indist}, no tester has any advantage testing product states and those far from being product using a single copy. We remark that it is known that insisting on perfect completeness means accepting everything because the product states span the entire space. This example rules out any other possible test giving up perfect completeness.

\subsection{Subset States v.s. Haar Random States}
Testing productness is impossible with one copy, but is achievable with two copies of the state~\cite{HM13}. In this section, we demonstrate a property impossible to test even with polynomially many copies, that is subset states of fixed size. In particular, we show that testing the size $s$ of the subset is impossible by considering the naive ensembles for different support size $s, \ell\in [N]$:
\begin{align*}
\cE_{1} & =\left\{ \phi_{S}=\frac{1}{\sqrt{|S|}}\sum_{i\in S}\ket i\,:\,|S| = s \right\} ,\\
\cE_{0} & =\left\{ \phi_{S}=\frac{1}{\sqrt{|S|}}\sum_{i\in S}\ket i\,:\,|S| = \ell \right\} .\\
\end{align*}
The key technical result will be:
\begin{theorem}[\cref{thm:subset-state-PRS} restated]
  Let $\cH=\C^d$ be a Hilbert space of dimension $d \in \mathbb{N}$, $\mu$ be the Haar measure on $\mathcal{H}$, and $S\subseteq [d]$ of size $s$. Then for any $k\in \mathbb{N}$,
  \begin{align*}
    \left\| \int{\psi^{\otimes k} d\mu(\psi)} - \mathop{\mathbb{E}}_{S\subseteq[d], |S|=s} \phi_S ^{\otimes k} \right\|_1 \le O\left(\frac{k^2}{d} + \frac{k}{\sqrt{s}} + \frac{s k}{d}\right),
  \end{align*}
  where  $\phi_S =  \left(\frac{1}{\sqrt{s}}\sum_{i\in S}\ket i\right)\left(\frac{1}{\sqrt{s}}\sum_{i\in S} \bra i \right)$.
\end{theorem}

In the above theorem, the subset state is compared with Haar random states, which by triangle inequality translates to a comparison between subset state of different subset sizes. The theorem itself implies new construction of quantum pseudorandom states as explained in the introduction.
Note that the theorem is optimal in the following sense: When the subset size $s=O(\poly(n))$ is small, collision attacks illustrate that measuring some subset state in the computational basis of support size $s$ for $O(\sqrt{s})$ times, a collision will be observed, that distinguishes subset state from Haar random state. When the support size is large, in particular, if $s = \Omega(d/\poly(n))$, then the overlap between the subset state and the uniform superposition of the computational basis will be significant, i.e., $1/\poly(n)$. Then with polynomially many copies, the subset state will be distinguishable from the Haar random state.

As a corollary of~\cref{lem:ens_indist_point_indist} and~\cref{thm:subset-state-PRS}, we have the following. 
\begin{theorem}[Failure of Standard Testing]
  Even given $\lceil 2^{n/16} \rceil$ copies, no tester can distinguish between subset states of size $\lceil 2^{n/8} \rceil$
  from $\lceil 2^{n/4} \rceil$ with probability better than $O(2^{-n/16})$.
\end{theorem}

So for the task of distinguishing subset state of very different support size, even with exponentially many copies, the advantage is still exponentially small.

\begin{proof}
  Let $d = 2^{n}$, $k=\lceil 2^{n/16} \rceil$, $s=\lceil 2^{n/8} \rceil$, and $s' =\lceil 2^{n/4} \rceil$. Set
  \begin{align*}
    \calA = \left\{ \phi_S \mid S \subseteq [d], \abs{S} = s \right\} \quad \text{ and } \quad \calB = \left\{ \phi_S \mid S \subseteq [d], \abs{S} = s' \right\}.
  \end{align*}
  Let $\mu_\calA$ and $\mu_\calB$ be uniform distributions on $\calA$ and $\calB$, respectively.
  By~\cref{thm:subset-state-PRS} and triangle inequality, we obtain
  \begin{align*}
    \norm{\rho_\calA - \rho_\calB}_1 &~\le~ \norm{\rho_\calA -\int{\psi^{\otimes k} d\mu(\psi)}}_1 + \norm{\int{\psi^{\otimes k} d\mu(\psi)}-\rho_\calB}_1\\
    &~\le~ O\left(\frac{k^2}{d} + \frac{k}{\sqrt{s}} + \frac{s k}{d} +\frac{k}{\sqrt{s'}} + \frac{s' k}{d}\right)\\
    &~\le~ O\left(\frac{1}{2^{n/16}}\right)\mcom
  \end{align*}
  where the last inequality follows from our choices of $d$,$k$,$s$, and $s'$. Now, applying~\cref{lem:ens_indist_point_indist} to
  $\rho_\calA$ and $\rho_\calB$, we conclude the proof.
\end{proof}

We now set off to prove~\cref{thm:subset-state-PRS}. There will be three steps: 1. Give an approximant of the Haar random states; 2. Give an approximant of the random subset state; 3. Show that the two approximants are indistinguishable.

\paragraph{Approximate the Mixture of Haar Random States.}
First, let's look at the Haar random state. 
A well-known fact by representation theory gives an explicit formula
for the mixture of Haar random states, $\Psi = \int \psi^{\otimes k} d \mu $, where $\mu$ is the Haar measure. For a detailed proof, see for example \cite{harrow2013church}.
\begin{fact}
\begin{align*}
\int \psi^{\otimes k} d \mu & =\binom{d+k-1}{k}^{-1}\cdot\frac{1}{k!}\sum_{\pi\in S_{k}}\sum_{\vec{i}\in[d]^{k}}|\vec{i}\rangle\langle\pi(\vec{i})|.
\end{align*}
\end{fact}
Instead of working with $\Psi$ directly, we look at the operator $\tilde \Psi = \Pi \Psi \Pi$, where $\Pi$ is the projection onto the subspace of
\[ \Span\{\ket{\vec{i}} : \vec i \in A([d],k)\} \subseteq \cH^{\otimes k}.\]
Recall that $A([d],k)$ is the set of k-tuples of $[d]$ without repeated elements. Immediately,
\begin{align}
\tilde{\Psi}=\binom{d+k-1}{k}^{-1}\cdot\frac{1}{k!}\sum_{\pi\in S_{k}}\sum_{\vec{i}\in A([d],k)}|\vec{i}\rangle\langle\pi(\vec{i})|.
\end{align}
As long as $k$ is small, we expect that $\Psi \approx \tilde \Psi$. This is simple and known. For completeness, we present a proof.

\begin{proposition}\label{prop:haar-pi-haar}
$\|\Psi -\tilde\Psi\|_{1}=O(k^{2}/d).$
\end{proposition}

\begin{proof}
Consider the following decomposition of $\Psi:=\tilde\Psi +\cR$. Note that 
\[
\cR = (I-\Pi) \Psi (I-\Pi).
\]
It's clear that $\tilde{\Psi}$ and $\cR$ are both positive semi-definite, and
$\tilde{\Psi}\cR= 0.$ Therefore, the nonzero eigenspaces of $\Psi$ correspond to those of $\tilde \Psi$ and $\cR$, respectively. Consequently,
\begin{align*}
\left\Vert \Psi-\tilde{\Psi}\right\Vert _{1} & =1-\|\tilde{\Psi}\|_{1}=1-\frac{d^{\underline{k}}}{(d+k-1)^{\underline{k}}}\\
 & \qquad=1-\frac{d}{d+k-1}\cdot\frac{d-1}{d+k-2}\cdots\frac{d-k+1}{d}\\
 & \qquad\le O\left(\frac{k^{2}}{d}\right).\qedhere
\end{align*}
\end{proof}

\paragraph{Approximate the Mixture of Random Subset State.}
Next, we turn to random subset states. Let $\Phi = \Exp_{|S|=s} \phi_S ^{\otimes k}$, and consider $\Pi \Phi \Pi$, but normalized.\footnote{Although in the case of Haar random state we didn't normalize, this doesn't really matter. Our choice is for simplicity of proof.} In particular,
\[
\tilde{\Phi}=\Exp_{S:|S|=s}\left[\frac{1}{s^{\underline{k}}}\sum_{\vec{i},\vec{j}\in A(S,k),}|\vec{i}\rangle\langle\vec{j}|\right].
\]
Analogous to Proposition~\ref{prop:haar-pi-haar}, we have
\begin{proposition}\label{prop:subset-pi-subset}
$\|\Phi-\tilde{\Phi}\|_{1}=O(k/\sqrt{s}).$
\end{proposition}

\begin{proof} Let $\gamma$ be the uniform distribution over subset $S\subseteq [d]$  of size $s$,
\begin{align*} 
 & \left\Vert \int_S \left(\frac{1}{s^{k}}\sum_{\vec{i},\vec{j}\in S^{k}}|\vec{i}\rangle\langle\vec{j}|-\frac{1}{s^{\underline{k}}}\sum_{\vec{i},\vec{j}\in A(S,k)}|\vec{i}\rangle\langle\vec{j}|\right) d\gamma \right\Vert _{1}\\
 & \qquad\le\int_{S}\left\Vert \frac{1}{s^{k}}\sum_{\vec{i},\vec{j}\in S^{k}}|\vec{i}\rangle\langle\vec{j}|-\frac{1}{s^{\underline{k}}}\sum_{\vec{i},\vec{j}\in A(S,k)}|\vec{i}\rangle\langle\vec{j}|\right\Vert _{1} d\gamma\le O\left(\frac{k}{\sqrt{s}}\right).\qedhere
\end{align*}
\end{proof}
All that is left to do is to show that $\|\tilde{\Phi}-\tilde{\Psi}\|_{1}$ is small.
Fix any $\vec{i},\vec{j}\in A([d],k)$, and let $\ell=\ell(\vec{i},\vec{j})$
be the total number of distinct elements in the union of the elements of the vectors  $\vec{i},\vec{j}$. Then the $(\vec{i},\vec{j})$'th entry of $\tilde{\Phi}$ is
\begin{align}
\tilde{\Phi}(\vec{i},\vec{j}) & =\frac{1}{s^{\underline{k}}}\Pr_{|S|=s}[\vec{i},\vec{j}\in A(S,k)]=\frac{1}{s^{\underline{k}}}\frac{\binom{d-\ell}{s-\ell}}{\binom{d}{s}}=\frac{s^{\underline{\ell}}}{s^{\underline{k}}\cdot d^{\underline{\ell}}}.\label{eq:entry-value}
\end{align}

\paragraph{Comparing the Approximants.}
\begin{proposition}\label{prop:haar-vs-subset-appr}
For any $k\ll s\le d,$ it holds that \[\|\tilde{\Phi}-\tilde{\Psi}\|_{1}=O\left(\frac{sk}{d}\right).\]
\end{proposition}

\begin{proof}
Let 
\[
\cD=\tilde{\Phi}-\frac{(d+k-1)^{\underline{k}}}{d^{\underline{k}}}\tilde{\Psi}.
\]
This factor is chosen so that $\cD(\vec{i},\vec{j})=0$ for any $\vec{i}$
and $\vec{j}$ such that $\vec{j}=\pi(\vec{i})$ for some permutation
$\pi.$ By triangle inequality,
\[
\|\tilde{\Phi}-\tilde{\Psi}\|_{1} \le 
    \left\|
    \cD
    \right\|_1
    + \left\|
            \tilde{\Psi}-\frac{(d+k-1)^{\underline{k}}}{d^{\underline{k}}}\tilde{\Psi}
    \right\|_1,
\]
where the second term is bounded by $O(k^2/d)$.

We turn to $\cD$. Let $\vec{j}\sim\vec{i}$ to denote that $\vec{j}$ is a permutation
of $\vec{i}$. Note that for any $\vec{i}\sim\vec{j}$, $\cD(\vec{i},\cdot)=\cD(\vec{j},\cdot)$,
and similarly $\cD(\cdot,\vec{i})=\cD(\cdot,\vec{j}).$ Therefore
$\cD=\tilde{\cD}\otimes J$
where $J\in\C^{k!\times k!}$ is the all 1 matrix and $\tilde{\cD}\in\C^{\binom{[d]}{k}\times\binom{[d]}{k}},$
s.t. for any $A,B\in\binom{[d]}{k},$ 
\begin{align*}
\tilde{\cD}(A,B)=\cD(\vec{i},\vec{j}), &  & \vec{i},\vec{j}\text{ contain }A\text{ and }B,\text{ respectively.}
\end{align*}
Next, decompose $\tilde{\cD}:=\sum_{t=0}^{k-1}\alpha_{t}\cD_{t},$ where in view of (\ref{eq:entry-value}),
\begin{align*}
 & \alpha_{t}=\frac{(s-k) \cdots (s-2k+t+1)}{d \cdots (d-2k+t+1)},\\
 & \cD_{t}(A,B)=\begin{cases}
1 & |A\cap B|=t,\\
0 & \text{otherwise.}
\end{cases}
\end{align*}
$\cD_{t}$ is the adjacency matrix for the well-studied generalized
Johnson graphs~\cite{delsarte1973algebraic}. In particular,
we will need the following fact (explained in Appendix).
\begin{fact}\label{fact:johnson-trace}
For any $0\le t\le k-1$, and for $k=O(\sqrt{d})$
\[
\|\cD_{t}\|_{1}\lesssim \binom{d-k}{k-t}\binom{d}{t}2^{k-t}.
\]
\end{fact}
Assisted by the above fact, we can bound $\|\cD\|_{1}$ for $sk=O(d)$ as below, 
\begin{align*}
    \|\cD\|_{1}=k! & \|\tilde{\cD}\|_{1}\le k!\sum_{t=0}^{k-1}\alpha_{t}\|\cD_{t}\|_{1}\lesssim k!\sum_{t=0}^{k-1}\frac{s^{k-t}}{d^{2k-t}}\cdot\frac{d^{k}}{t!(k-t)!}\cdot2^{k-t}\\
     & =\sum_{t=0}^{k-1}\left(\frac{2s}{d}\right)^{k-t}\binom{k}{t}= \left(1+\frac{2s}{d}\right)^{k}-1\lesssim O\left(\frac{2sk}{d}\right). \qedhere
\end{align*}
\end{proof}
Theorem~\ref{thm:subset-state-PRS} follows by triangle inequality on Propositions ~\ref{prop:haar-pi-haar}-\ref{prop:haar-vs-subset-appr}.

\subsection{Indistinguishability of Support Size in the Dense Regime}\label{subsec:dense_indist}
We have discussed that subset state with small support of fixed size are information-theoretically indistinguishable from Haar random states. This fact rules out property testing distinguishing general states with support size $s_0 = \omega(\poly(n))$ and $s_1 = 2^n / \omega(\poly(n)).$ 
Then one may hope that property testing for large support size (constant density) may be possible.\footnote{In the sparse regime for $s=O(\poly(n))$, one can learn the state.} In this section, we adapt a similar proof to show that this is also impossible.

In particular, we consider the following two ensembles for parameters $s,t,p\in(0,1)$: 
\begin{align*}
\cE_{1} & =\left\{ \phi_{S}=\frac{1}{\sqrt{|S|}}\sum_{i\in S}\ket i\,:\,|S|\approx p d\right\} ,\\
\cE_{0} & =\left\{ \phi_{S,T}=\frac{1}{\sqrt{2|S|}}\sum_{i\in S}\ket i-\frac{1}{\sqrt{2|T|}}\sum_{j\in T}\ket j\,:\,|S|\approx s d,|T|\approx t d,S\cap T=\varnothing\right\} .
\end{align*}
To be precise, the underlying distributions of the two ensembles are: 
\begin{enumerate}
    \item For $\cE_1$, sample state $\phi_{S}$ by letting $i\in S$ w.p. $p$ for all $i\in [d]$ independently at random;
    \item For $\cE_0$, sample state $\phi_{S,T}$ by the following process: For each $i\in [d]$ independently, sample a uniformly random $r\in [0,1],$ then let $i\in S$ if $r<s$; let $i\in T$ if $s\le r < s+t $.
\end{enumerate} 
Choose $p, s, t \in (0,1)$ to be some constant such that 
\begin{equation}\label{eq:ell-1-cond}
    \sqrt{\frac{s}{2}} - \sqrt{\frac{t}{2}} = \sqrt{p} .
\end{equation}
The choice is made so that states from the two ensembles have about the same overlap with $\ket\mu=\sum_{i\in [d]} \ket i / \sqrt{d}.$ Otherwise, comparing with $\ket\mu$ will be a valid attack distinguishing the two ensembles. On the other hand, this overlap condition is the only thing matters: In $\cE_0$, there are positive part and negative part, it is totally fine to have both parts positive, the analysis works equally well. Such examples explain in general why it is hard to test density in the dense regime.

For concreteness, one can set $s=8p,$ and $t=2p$ in (\ref{eq:ell-1-cond}). The support size of a random state from $\cE_1$ will be  $pd \pm \epsilon d$ almost surely for arbitrarily small constant $\epsilon>0$; while the support size of a random state from $\cE_0$ will be $(s+t)d \pm \epsilon d$ almost surely. Note that $s+t = 10 p$, i.e., states in $\cE_0$ has 10 times larger support size than $\cE_1$. A slight abuse of notation, we also use $\cE_0, \cE_1$ to denote the mixed state for states of the average state $\cE_0, \cE_1$, respectively. Our goal is to show that
\[
    \|\cE_0 - \cE_1 \| = \negl(n).
\]

\paragraph{Approximant of $\cE_1$.}
We consider approximant of the average of state from $\cE_1$, 
\begin{center}
\begin{tikzcd}
    \cE_1 \arrow[r, "\Pi"] & \cE_1' \arrow[r,"\text{flatten}"] &[1em] \cE_1''.
\end{tikzcd}
\end{center}
Recall $\Pi$ be the projection onto the subspace  $\Span\{\ket{\vec i}: \vec i\in A([d],k)\}$, then $\cE_1' = \Pi\cE_1\Pi$. In the approximant $\cE_1''$, we pretend that the ``amplitudes'' do not depend on $|S|$, in other words we ``flatten'' the distribution that we sample the state. In particular,
\begin{align*}
    \cE_{1} & = \Exp_{S}\left[|S|^{-k}\left(\sum_{i\in S}\ket i\right)^{\otimes k}\left(\sum_{i\in S}\bra i\right)^{\otimes k}\right];\\
    \cE'_{1} & = \Exp_S \left[|S|^{-k}\sum_{\vec i, \vec j \in A(S,k)} \ket{\vec i}\bra{\vec j}\right];\\
    \cE''_{1} & = \Exp_S \left[(pd)^{-k}
        \sum_{\vec i, \vec j \in A(S,k)}  \ket{\vec i}\bra{\vec j}
    \right].
\end{align*}
We compare $\cE_1$ and $\cE_1'$ as follows
\begin{align*}
    \|\cE_1 - \cE_1'\| &
        \le \Exp_S |S|^{-k} \left\|\sum_{\vec i, \vec j \in S^k} \ket{\vec i}\bra{\vec j}
            - \sum_{\vec i, \vec j \in A(S,k)} \ket{\vec i}\bra{\vec j}\right\| 
            \\
    & \le O\left(\Exp_S \sqrt{\frac{k^2}{|S|}}\right) = O\left(\frac{k}{\sqrt d}\right).
\end{align*}
Next compare $\cE'_1$ and $\cE''_1$, we claim
\begin{align}   \label{eq:cE-one-error-1}
    \|\cE_1'' - \cE_1' \| \le O\left(\frac{k}{d^{2/5}}\right).
\end{align}
Consider the interval $L=pd \pm d^{3/5}$. By Chernoff Bound, the probability that $|S|\not\in L$ is $\exp(-\Omega(d^{1/5})).$ Then (\ref{eq:cE-one-error-1}) is a direct conclusion from the following two bounds.

1. For $S\in L$, 
\begin{align*}
    \left\|
        (|S|^{-k} - (pd)^{-k})\sum_{\vec i, \vec j \in S^k} \ket{\vec i}\bra{\vec j}
    \right\|
    & = |S|^{\underline k}\cdot \frac{|S|^k - (pd)^k}{ (|S|pd)^k} 
    \\
    & \le \left(1+\frac{k^2}{|S|}\right) \left(\left( 1+\frac{1}{pd^{2/5}}\right) ^k-1\right)
    \\
    & \le O\left(\frac{k}{d^{2/5}}\right).
\end{align*}

2. For $S \not\in L$,
    \begin{align*}
            \sum_{S: |S|\not\in L} \Pr[S] \left\|(pd)^{-k}\sum_{\vec i, \vec j \in S^k} \ket{\vec i}\bra{\vec j}\right\|
            ,\;  \sum_{S: |S|\not\in L} \Pr[S] \left\||S|^{-k}\sum_{\vec i, \vec j \in S^k} \ket{\vec i}\bra{\vec j}\right\| 
            \le \exp(-\Omega(d^{1/5})).
    \end{align*}
If follows that 
\begin{align*}
    \|\cE_1 - \cE_1''\| \le O\left(\frac{k}{d^{2/5}}\right).
\end{align*}

For any $\vec i, \vec j \in A([d], k)$, compute the entry of $\cE''(\vec i, \vec j)$ explicitly. Note the entry depends only on $\ell = |\vec i \cup \vec j|$,
\begin{align}
    \bra{\vec i} \cE'_1 \ket{\vec j} = p^{\ell-k}\cdot d^{-k}. \label{eq:cE-one}
\end{align}

\paragraph{Approximant of $\cE_0$.}
We consider approximant of the average of state from $\cE_0$ completely analogously (and a lot more tedious) to $\cE_1$, 
\begin{center}
\begin{tikzcd}
    \cE_0 \arrow[r, "\Pi"] & \cE_0' \arrow[r,"\text{flatten}"] &[1em] \cE_0''.
\end{tikzcd}
\end{center}
In particular,
\begin{align*}
    \cE_{0} & = \Exp_{S, T}\left[\left(\frac{1}{\sqrt{2|S|}}\sum_{i\in S}\ket i-\frac{1}{\sqrt{2|T|}}\sum_{j\in T}\ket j\right)^{\otimes k} \left(\frac{1}{\sqrt{2|S|}}\sum_{i\in S}\bra i-\frac{1}{\sqrt{2|T|}}\sum_{j\in T}\bra j\right)^{\otimes k} \right];\\
    \cE'_{0} & = \Pi \cE_0 \Pi;\\
    \cE''_{0} & = \Pi\left(\Exp_{S, T}\left[\left(\frac{1}{\sqrt{2sd}}\sum_{i\in S}\ket i-\frac{1}{\sqrt{2td}}\sum_{j\in T}\ket j\right)^{\otimes k} \left(\frac{1}{\sqrt{2sd}}\sum_{i\in S}\bra i-\frac{1}{\sqrt{2td}}\sum_{j\in T}\bra j\right)^{\otimes k} \right]
    \right) \Pi.
\end{align*}

The analysis would also be completely analogous to that of $\cE_1$. We omit the calculations here,
\[
    \|\cE_0 - \cE_0''\| \le O\left(\frac{k}{d^{2/5}}\right).
\]

Now we compute the entry for $\cE_0''$ explicitly. Fix any $\vec{i},\vec{j}\in A([d],k),$ let $\ell=\ell(\vec{i},\vec{j})$
be the total number of distinct elements in vectors $\vec{i}$ union
$\vec{j}$. Let $a:=2(\ell-k),b:=2k-\ell$. So $a$ is the number
of elements that appeared only in $\vec{i}$ or $\vec{j}$, while
$b$ is the number of elements that appeared in both $\vec{i}$ and
$\vec{j}$. Then
\begin{align}
    \bra{\vec i}\cE_0'' \ket{\vec j}
    & = \sum_{\ell_{1}=0}^{a}\sum_{\ell_{2}=0}^{b}\binom{a}{\ell_{1}}\binom{b}{\ell_{2}}\left(\frac{1}{\sqrt{2sd}}\right)^{\ell_{1}+2\ell_{2}}s^{\ell_{1}+\ell_{2}}\left(\frac{-1}{\sqrt{2td}}\right)^{a-\ell_{1}+2(b-\ell_{2})}t^{a-\ell_{1}+b-\ell_{2}}
    \nonumber\\
    & = \sum_{\ell_{1}=0}^{a}\sum_{\ell_{2}=0}^{b}\binom{a}{\ell_{1}}\binom{b}{\ell_{2}}
    \left(\frac{s}{\sqrt{2sd}}\right)^{\ell_{1}}
    \left(\frac{1}{\sqrt{2d}}\right)^{2\ell_{2}}
    \left(\frac{-t}{\sqrt{2td}}\right)^{a-\ell_{1}}
    \left(\frac{1}{\sqrt{2d}}\right)^{2(b-\ell_{2})}
    \nonumber\\
    & = \sum_{\ell_{1}=0}^{a}
    \binom{a}{\ell_{1}}
    \left(\frac{s}{\sqrt{2sd}}\right)^{\ell_{1}}
    \left(\frac{-t}{\sqrt{2td}}\right)^{a-\ell_{1}}
    \sum_{\ell_{2}=0}^{b}
    \binom{b}{\ell_{2}}
    \left(\frac{1}{{2d}}\right)^{b}
    \nonumber\\
    & =\left(\frac{s}{\sqrt{2sd}}-\frac{t}{\sqrt{2td}}\right)^{a}\left(\frac{1}{d}\right)^{b}
    \nonumber\\
    & =\left(\sqrt{\frac{p}{d}}\right)^{a}\left(\frac{1}{d}\right)^{b}
    \nonumber\\
    &=\frac{p^{\ell-k}}{d^{k}}. \label{eq:cE-zero}
\end{align}

\paragraph{Indistinguishability of the Two Ensembles}
Note, $\cE''_0 = \cE''_1 $ by (\ref{eq:cE-one}) and (\ref{eq:cE-zero}). By triangle inequality,
\[
    \|\cE_0 - \cE_1 \| \le O\left(\frac{k}{d^{2/5}}\right) = \negl(n).
\]

\subsection{Quantum and Classical Property Testing}
\label{subsec:quantum-encoding-of-distr}
Property testing for classical distribution can be viewed as a degenerated version of property testing for quantum states. To make this point precise, we encode classical distributions as states and show that lower bounds on quantum property testing imply lower bounds for property
testing of classical  distributions. 

Let $\calA \subseteq \Delta_N$ be a collection of probability distributions in $\R^N$. We say that $\ket\psi$ has classical shadow in $\calA$ if $\ket{\psi} = \sum_{i=1}^N \alpha_i \ket{i}$ satisfying $(\abs{\alpha_1}^2,\ldots,\abs{\alpha_N}^2) \in \calA$. This definition generalizes to mixed states, i.e., $\psi$ has classical shadow in $\calA$ if $\psi$ can be expressed as some mixture of pure states where each pure state has a classical shadow in $\calA$. Now a further generalization to reflect more than one ``samples'' from classical distribution,
a general state $\rho$ has a classical $k$-shadow in $\calA$
if for some distribution $\lambda$ on the $\ell_2$-unit sphere of $\C^N$ and $k \in \mathbb{N}$, the state $\rho$ can be expressed as
\begin{align*}
  \rho = \E_{\ket{\psi} \sim \lambda} \left(\ket{\psi} \bra{\psi}\right)^{\otimes k},
\end{align*}
with every $\ket{\psi}\sim\lambda$ has a shadow in $\calA$.
So quantum states with $k$-shadow of $\calA$ corresponds to the natural quantum counterparts for a mixture $\cD$ of distributions in $\calA$ from where $k$ samples will be drawn. Precisely, let $\Lambda$ be the channel that measures each copy $\ket{\psi}\bra{\psi}$ in the standard computational basis. Then the effect of $\Lambda$ is to make $k$ samples from a random distribution $\nu\in\cD$. 

\begin{claim}\label{claim:shadow_proj}
  If $\rho$ has a classical $k$-shadow in $\calA \subseteq \Delta_N$, then
  \begin{align}
    \Lambda(\rho) = \E_{\nu \sim \calD} \left(\sum_{i=1}^N \nu_i \ket{i} \bra{i}\right)^{\otimes k}\mcom
  \end{align}
  where $\calD$ is some distribution on $\calA$.
\end{claim}

Using~\cref{thm:subset-state-PRS} and~\cref{claim:shadow_proj}, we deduce the following lower bound.

\begin{corollary}[Failure of Standard Classical Testing]
  Even given $\lceil 2^{n/16} \rceil$ samples of a flat distribution, no tester can distinguish between support size $\lceil 2^{n/8} \rceil$ from $\lceil 2^{n/4} \rceil$
  with probability better than $O(2^{-n/16})$.
\end{corollary}

\begin{proof}
  Consider two ensemble of subset states $\rho_{\calA}$ and $\rho_{\calB}$ with subset sizes $\lceil 2^{n/8} \rceil$
  and $\lceil 2^{n/4} \rceil$, respectively.
  From~\cref{thm:subset-state-PRS}, we deduce that $\norm{\rho_{\calA} - \rho_{\calB}}_1 \le O(2^{-n/16})$.
  Note that $\Lambda$ does not increase the trace distance, so $\norm{\Lambda(\rho_{\calA}) - \Lambda(\rho_{\calB})}_1 \le O(2^{-n/16})$ concluding the proof. 
\end{proof}

%% file: c-cert-testing.tex
\section{Testing with Classical Certificates: Another Fiasco}\label{sec:classical-lower-bound}

In this section, we discuss property testing for classical distribution in the presence of classical certificates. We deliberately choose the word certificates to contrast the notion of proofs in the definition of $\propMA$. 

Our investigation in this section focuses on the concrete problem: establishing lower bounds for certifying the support size, the classical counterpart to the subset state problem.
We consider the problem $\textup{GapSupp}_{s,\ell}$, that is to distinguish the following two ensembles:

\begin{definition}[The Gap Support Size Problem]
The $\textup{GapSupp}_{s,\ell} = (\calP_{\textup{yes}}, \calP_{\textup{no}})$, where
    \begin{description}        
        \item (YES) $\cE_1=\{\text{uniform distribution on support }S : S\subseteq[N], |S|=s\}$
        \item (NO) $\cE_0=\{\text{uniform distribution on support } S: S\subseteq[N], |S|=\ell\}$.
    \end{description}
\end{definition}

Since $\textup{GapSupp}_{s,\ell}$ can be thought of as a special case of quantum state property due to~\cref{subsec:quantum-encoding-of-distr}, we think $N$ as some exponentially large quantity, and let $n=\log N$.

\subsection{Lower Bounds on Classical Testing with Proofs via HDX Fast Mixing}
Our main result in the section is the following.
\begin{theorem}[Classical Indistinguishability for Subset Distribution with Certificates]\label{thm:subset-distribution-proof}
    For some parameter $s=\omega(\poly(n))$, 
    given any proof of length $p$ and allowing $t$ samples, then the verifier can distinguish $\textup{GapSupp}_{s,2s}$ with an advantage at most
    \[
        O\left(\sqrt{\frac{tp}{s}}+\frac{t^2}{s}\right).
    \]
\end{theorem}

\begin{remark}[Optimality]\label{re:tightness-classical}
    The above lower bound is tight in general. To see this, note that one obvious strategy that the honest prover can do is to send set $T$ consisting of $p$ elements from the support of size $s$ at the cost about $p\log (N/p)$ communication complexity. In the yes case, the probability of seeing a collision with $T$ sampling $s/p$ elements is $(1-p/s)^{s/p}$; while in the no case, the probability seeing a collision with $T$ in the samples is $(1-p/2s)^{s/p}$. The two probabilities can differ by $\Omega(1)$.
\end{remark} 

Our proof relies on the connection to the fast mixing of high-dimension expander, which is just a simplicial complex, i.e., a downward closed set system. A random Down walk from a vertex $v$ in a high-dimensional expander representing some set $S$ roughly corresponds to make a small number of samples of the flat distribution on $S$. Fast mixing on a high-dimension expander roughly says that the Down walk from some vertex $v$ at some high level mixes very fast, in another words, if the number of samples is small, it looks like random samples from a uniform distribution. To obtain the tight bound, it is crucial analyze precisely how fast the Down walk mixes.

\begin{proof}
Suppose we are assisted with a proof of length $p$. For any string $\Pi\in\{0,1\}^p$, let $\cF_\Pi$ denote the set of $S\in \cE_1$ such that the faithful prover will provide the proof $\Pi$. Now there will be two situations, one with a faithful prover, one with the adversarial prover:
\begin{description}
    \item (Yes) Consider the YES distribution indicated by its support $S$ chosen randomly from $\cE_1$. Let $\Pi$ be the faithful proof associated with $S$. The verifier will sample $t$ elements from the distribution. Overall, the verifier observes  $X_1, X_2, \ldots, X_t$ together with the proof $\Pi$.
    \item (NO) Suppose a uniformly random distribution indicated by its support $S'$ is chosen from $\cE_0$. The adversary will send a proof $\Pi'$ to the verifier with probability $|\cF_{\Pi'}| / \binom{N}{s}$, independent to $S'$.  The verifier samples $t$ elements from the uniform distribution on $S'$ together with the adversary proof $\Pi'$. So the verifier sees $Y_1, Y_2, \ldots, Y_t$ together with an adversarial proof $\Pi'$. 
\end{description}

Let $\nu_1$ and $\nu_0$ denote the distribution on samples together with the proofs that the verifier sees in the YES and NO case, respectively.

Now, note that for $t\ll s$, the probability that $X_1,X_2, \ldots, X_t$ consist some collision is $O(t^2/s)$. Therefore in YES case, we can alternatively think of the distribution of $X_1, X_2, \ldots, X_t, \Pi$, as follows

\vspace{1.5mm}
\noindent\fbox{\begin{minipage}[t]{1\columnwidth - 2\fboxsep - 2\fboxrule}%
\textbf{\uline{Modified Faithful Process}} (To generate $\tilde\nu_1$ on $X_1,X_2,\ldots,X_t,\Pi$): \vspace{0.5mm}
\begin{enumerate}
    \item Sample $\Pi$ with probability $|\cF_\Pi|/\binom{N}{s}$;
    \item Sample $S\in\cF_\Pi$ at random;
    \item Sample a subset $T\subseteq S$ of size $t$ uniformly at random;
    \item Let $X_1, X_2, \ldots, X_t$ be some uniformly random permutation $\tau$ on $T$.
\end{enumerate}
\end{minipage}}\vspace{1.5mm}
\begin{center}
\begin{tikzcd}
    \nu \arrow[r, "\text{collisionless}"] &[3em] \tilde \nu .
\end{tikzcd}
\end{center}

The new distribution, denoted $\tilde\nu_1$, differs from the old $\nu_1$ in statistical distance by $O(t^2/s)$. Analogously, consider the distribution on $\tilde\nu_0$ which differs from $\nu_0$ in statistical distance at most $O(t^2/s)$ described below,

\vspace{1.5mm}
\noindent\fbox{\begin{minipage}[t]{1\columnwidth - 2\fboxsep - 2\fboxrule}%
\textbf{\uline{Modified Adversarial Process }} (To generate $\tilde\nu_0$ on $Y_1,Y_2,\ldots,Y_t,\Pi'$): \vspace{0.5mm}
\begin{enumerate}
    \item Sample $\Pi'$ based on $|\cF_{\Pi'}|/\binom{N}{s}$;
    \item Sample $R\in \cE_0$ uniformly at random; 
    \item Sample a subset $S'\in R$ uniformly at random of size $s$;
    \item Sample a subset $T'\subseteq S'$ of size $t$ uniformly at random;
    \item Let $Y_1, Y_2, \ldots, Y_t$ be some uniformly random permutation $\tau'$ on $T$.
\end{enumerate}
\end{minipage}}\vspace{1.5mm}

In the above description for $\tilde\nu_0$, step (iii) is redundant as sampling $t$-subset $T'$ from an $s$-subset $S'$ that itself is a random subset of $R$ is the same as sampling a $t$-subset $T'$ from $R$ directly. Furthermore, the overall distribution of $T'$ is uniform over $t$-subset of $[N]$ (and $S'$ will be a uniform $s$-subset of $[N]$). This redundancy is introduced for the purpose of analysis.

Therefore, $\tilde\nu_1$ corresponds to essentially what the verifier reads in the YES case, and $\tilde\nu_0$  corresponds to what the verifier reads in the NO case. 
Based on the discussion so far, to prove our claimed bound in the theorem 
\[
    \|\nu_0 - \nu_1 \| \le \sqrt{\frac{tp}{2s}}+O(t^2/s) ,
\]
it suffices to show
\begin{equation}\label{eq:classical-indist}
    \|\tilde\nu_0 - \tilde\nu_1\|_1 \le \sqrt{\frac{tp}{2s}} . 
\end{equation}
Now we justify the inequality~(\ref{eq:classical-indist}).
\begin{align}
    2\|\tilde\nu_0 - \tilde\nu_1 \|^2 
        &~\stackrel{(1)}{\le}~ \KL{\nu_0}{ \nu_1} 
        =\KL{\Pi T \tau}{\Pi' T' \tau'} 
        \nonumber \\
        & ~\stackrel{(2)}{=}~ \KL{\Pi \tau}{\Pi' \tau'} + \Exp_{\pi, \sigma}\KLfrac{T \mid \Pi=\pi, \tau=\sigma}{T' \mid \Pi'=\pi, \tau'=\sigma}
        \nonumber \\
        & ~\stackrel{(3)}{=}~ \Exp_\pi \KLfrac{T\mid\Pi=\pi }{T'\mid \Pi'=\pi}
        \nonumber \\
        & ~\stackrel{(4)}{=}~ \Exp_{\pi} \KLfrac{T \mid \Pi = \pi}{T'}
        \nonumber \\
        &~\stackrel{(5)}{\le}~ \Exp_{\pi} \frac{t}{s}\KLfrac{S\mid \pi}{S'}
        \nonumber \\
        & ~\stackrel{(6)}{\le}~  \frac{tp}{s},
\end{align}
where (1) uses Pinsker's inequality, and note there is a natural bijection between $\Pi, T, \tau$ and $\Pi, X_1, X_2, \ldots, X_t$ (analogously for $\Pi', T', \tau'$ and $\Pi', Y_1, Y_2, \ldots, Y_t$); (2) is by Chain rule for KL-divergence; (3) holds because $\Pi\tau$ and $\Pi'\tau'$ have the same distribution and the random permutation $\tau$ ($\tau'$) is independent from $T, \Pi$ ($T', \Pi'$); 
(4) holds because in the adversary case the proof is independent with $T$; (5) invokes a divergence contraction result~\cref{lem:div-contraction} that we explain later; and (6) holds because $S'$ is uniform over $\binom{[N]}{s}$ as $S$, and by definitions of mutual information and KL-divergence,
\begin{align*}
    \Exp_\pi\KL{(S\mid\pi)}{S} = I(S;\Pi)\le H(\Pi) \le p.&\qedhere
\end{align*}
\end{proof}

The missing technical component for the above proof is the following divergence contraction lemma, which is studied in the theory of higher-dimensional expanders. In particular, it is an application of the more general theorems proved in~\cite{CGM19, AJKPV22}. 
\begin{lemma}[Divergence contraction]\label{lem:div-contraction}
    Let $\mu_0$ be the uniform distribution over $\binom{[N]}{s}$, and $\mu_1$ be some distribution over $\binom{[N]}{s}$. Consider the following random process for $i\in \{0,1\}$:
    \begin{enumerate}
        \item Sample $S$ from $\mu_i$,
        \item Sample subset $T$ of size $t$ from $S$ uniformly at random.
    \end{enumerate}
    The above random process introduces a distribution $\lambda_i$. Then
    \begin{equation}\label{eq:div-contraction}
       \KL{\lambda_1}{\lambda_0} \le  \frac{t}{s}\cdot\KL{\mu_1}{\mu_0}.
    \end{equation}
\end{lemma}

For completeness, we provide a self-contained proof using a language consistent with our discussion so far, where (\ref{eq:div-contraction}) is replaced with a slightly weaker bound:
\begin{equation*}
    \KL{\lambda_1}{\lambda_0} \le  \frac{t}{s-t+1}\cdot\KL{\mu_1}{\mu_0}.
\end{equation*}
Note that for our application, $t \le O(\sqrt s)$, thus, $t/s$ and $t/(s-t+1)$ are essentially the same. For the tighter bound, see~\cite[Thoerem 5]{AJKPV22}.

\begin{proof}
Consider the random variables $X_1, X_2, \ldots, X_s$ which are obtained by drawing a random subset $S$ from $\mu_1$, and then  permute the elements by a random permutation $\tau$. Similarly, the random variables $Y_1, Y_2, \ldots, Y_s$ will be obtained by first drawing a random subset $S'$ from $\mu_0$, then randomly order the elements in the subset by $\tau'$. Note that
\begin{align}
    \KL{X_1X_2\ldots &X_s}{Y_1Y_2\ldots Y_s}
    = \KL{\tau S}{\tau' S'}
    \nonumber \\
    &=  \KL{\tau}{\tau'} + \Exp_\sigma \KLfrac{S\mid\tau=\sigma}{S'\mid\tau'=\sigma} = \KL{S}{S'} = \KL{\mu_1}{\mu_0},
    \nonumber
\end{align}
where the second step follows the chain rule, and third step holds as the permutations $\tau,\tau'$ are independent of $S, S'$. 
Analogously,
\begin{align}
    \KL{X_1X_2\ldots &X_t}{Y_1Y_2\ldots Y_t}
    = \KL{\lambda_1}{\lambda_0}.
    \nonumber
\end{align}
By chain rule, for any $\ell \le s$,
\begin{align}
    \KL{X_1X_2\ldots X_\ell}{Y_1Y_2\ldots Y_\ell}
    = \sum_{i=1}^\ell\Exp_{x\in A([N],s)} \KLfrac{X_i\mid X_{<i}=x_{<i}}{Y_i\mid Y_{<i}=x_{<i}}.
    \nonumber
\end{align}
Now we need the following claim.
\begin{claim}[cf. Theorem~4~\cite{AJKPV22}]\label{claim:chain-bound}
For any  $x\in A([N],s)$ and $1\le i\le s$,
    \begin{equation}
        \KLfrac{X_i\mid X_{<i}=x_{<i}}{Y_i\mid Y_{<i}=x_{<i}} \le \frac{1}{s-i+1}\cdot \KLfrac{X_i X_{i+1} \ldots X_s \mid X_{<i}=x_{<i} }{Y_i Y_{i+1} \ldots Y_s \mid Y_{<i}=x_{<i}}.
    \end{equation}
\end{claim}
The proof of the claim is deferred to the appendix. With the claim, we can finish the proof.
\begin{align}
    \KL{\lambda_1}{\lambda_0}
        & = \sum_{i=1}^t \Exp_{x\in A([N],s)} \KLfrac{X_i\mid X_{<i}=x_{<i}}{Y_i\mid Y_{<i}=x_{<i}}.
        \nonumber\\
        & \le \sum_{i=1}^t \frac{1}{s-i+1}\cdot \Exp_{x\in A([N],s)}\KLfrac{X_i X_{i+1} \ldots X_s \mid X_{<i}=x_{<i} }{Y_i Y_{i+1} \ldots Y_s \mid Y_{<i}=x_{<i}}
        \nonumber \\
        & \le \sum_{i=1}^t \frac{1}{s-i+1} \cdot \KL{X_1 X_2\ldots X_s}{Y_1 Y_2 \ldots Y_s}
        \nonumber \\
        & \le \frac{t}{s-t+1} \cdot \KL{\mu_1}{\mu_0}.\nonumber\qedhere
\end{align}
\end{proof}

In~\cref{thm:subset-distribution-proof}, the lower bound is stated for $\textup{GapSupp}_{s,2s}$, the YES case corresponds to the distribution of the smaller support size, and the NO case corresponds to the distribution of the larger support size. The support size of the NO case is somewhat arbitrary, and one can consider $\textup{GapSupp}_{s,\ell}$ for  $\ell > 2s$, or simply $\textup{GapSupp}_{s,N}$ where the NO case is simply the uniform distribution over $[N]$ (no ensemble at all). This a-priory makes the distinguishing task easier, but the same analysis holds and results in the same asymptotic bound, i.e., allowing $t$ samples it's impossible to distinguish YES case from uniform distribution with an advantage better than $O(\sqrt{tp/s} + t^2/s)$.

One can also consider $\textup{GapSupp}_{2s,s}$, i.e., the NO case having the smaller support size. This case is also captured by the same analysis. However, when the NO case has a smaller support size, it admits a much stronger lower bound that trivializes the problem. We discuss this case in~\cref{subsec:one-sideness-GapSupp}.

\subsection{A Generalization: Testing with Structured Classical Certificates}\label{subsec:structured-classical-certificate}
Note that in our proof to~\cref{thm:subset-distribution-proof}, regardless of whether the proof $\Pi$ associated with each subset $S$ is fixed or randomized, the analysis is exactly the same as long as (i) in the modified faithful process, the joint distribution of the proof and the distribution indicated by $S$ are set correctly; (ii) in the modified adversarial process, the marginal distribution of $\Pi'$ is set correctly. Therefore, our analysis is robust. We discuss the implication of the robustness formally.

Stated in the most general way, our lower bound technique works for certificates coming from any promised (intended to help testability) convex subset  $\calC$ of
the probability simplex $\Delta_\calS$ in $\R^{\calS}$, where $\calS$ is an arbitrary finite set representing all the possible certificates. 
In view of the proof to \cref{thm:subset-distribution-proof}, $\Pi$ can be any random variable depending only on the distribution $\mu_1\in \mathcal E_1$ to test. 
In other words, an element from $\Delta_\calS$ may represent some certificate on a YES input $\mu_1$. A certificate can be the proof from the prover in the case of $\propMA$ model stated in~\cref{thm:subset-distribution-proof}, and can be the entire communication transcript in the case of the (public coin) $\propAM$ model.

%

Starting from the trivial example, the delta-distributions on $\calS=\{0,1\}^p$, i.e., distributions supported on a single element in $\calS$. Such distributions correspond to the case where for a fixed input, there is a fixed  proof of length $m$. Therefore, we can let $\Delta_p^{\MA}$ denote the set of delta-distributions on $\calS$,
\[
    \Delta_p^\MA = \left\{\text{singleton distribution}\in \Delta_\calS : \calS = \{0,1\}^{p}\right\}.
\]
%
Allowing the convex hull of $\Delta^\MA$, we obtain all distributions $\Delta_\calS$---indeed a trivial example. This model captures $\MA$ proofs: Because for any $\MA$ protocol, in the yes case there is  a delta-distribution corresponding to the honest $\MA$ proof that will be accepted with high probability. In the no case, no proof will be accepted with high probability. Therefore for any mixed strategy from $\conv{\Delta^\MA}$, the verifier will reject with high probability. Consequently, we can give an alternative definition of the $\MA$ type property testing model.

\begin{definition}[Classical Property Testing with $\MA$ Type Certificates]\label{def:alt-propMA}
  Let $k=k(d),p=p(d) : \mathbb{N} \to \mathbb{N}$, $1 \ge a > b \ge 0$. A property $\calP = \sqcup \calP_d$ belongs to $\propMA_{a,b}(k,p)$ 
  with respect to certificates $\calC = \conv{\Delta_p^\MA}$ if there exists a verifier $V$ such that
  for every $\mu \in \Delta_d$,
  \begin{enumerate}
  \item if $\nu \in \calP_d$, then there exist $\mu \in \calC$ such that
    \begin{align*}
      \Pr_{x_1,\ldots,x_k \sim \mu^{\otimes k}, y \sim \mu}\left[V(x_1,\ldots,x_k,y) \textup{ accepts} \right] \ge a\mcom
    \end{align*}   
    \item if $\nu$ is $\epsilon$-far from $\calP_d$ (in statistical distance), then for every certificate $\mu \in \calC$ such that
    \begin{align*}
      \Pr_{x_1,\ldots,x_k \sim \mu^{\otimes k}, y \sim \mu}\left[V(x_1,\ldots,x_k,y) \textup{ accepts} \right] \le b\mper
    \end{align*}
  \end{enumerate}
\end{definition}

Replacing $\conv{\Delta^{\MA}}$ in~\cref{def:alt-propMA} by any other convex set $\calC\subseteq \Delta_p$, one obtains $\calC$ type certificates. Then~\cref{thm:subset-distribution-proof} can be stated in a more generalized way.

\begin{theorem}[Indistinguishability for Subset Distribution with Structured Certificates]\label{thm:subset-distribution-structured-certificates}
    For some parameter $s=\omega(\poly(n))$, 
    given any certificates of type $\calC \subseteq \Delta_p$ and allowing $t$ samples,  $\textup{GapSupp}_{s,2s}$ with an advantage at most
    \[
        O\left(\sqrt{\frac{tp}{s}}+\frac{t^2}{s}\right).
    \]
\end{theorem}

We now instantiate $\calC$ to certificates arising from a valid $\AM$ protocol---the communication transcripts. To illustrate, suppose $\cZ = \set{0,1}^p$ where $p=r n +r$, and consider $r$-round $\AM$ protocols where
the verifier sends $n$ uniformly random bits to the prover and receives $1$ answer bit at each round. A distribution $\mu$ on $\cZ$ naturally defines random variables
$ R_1  A_1\ldots  R_r  A_r$, where each $ R_i$ takes values in $\set{0,1}^n$ and $ A_i \in \set{0,1}$,
such that $\Pr[( R_1=r_1,  A_1=a_1, \ldots,  R_r = r_r,  A_r = a_r)] = \mu(r_1,a_1,\ldots,r_r,a_r)$. The collection
of distributions encoding valid $\AM$ protocols of this form is given by
\begin{align*}
\Delta_p^\AM = \left\{\mu \in \Delta_\cZ ~:~  
            \begin{aligned}
             \forall i,r_1,\ldots,r_{i}, \, 
             (R_i \mid & R_1=r_1\ldots R_{i-1}=r_{i-1} \textup{ is uniform})\\
            &\text{and } (A_i\mid R_1=r_1\ldots R_i=r_i \textup{ is fixed})
            \end{aligned}
            \right\}\mper
\end{align*}
Analogous to the discussion in the last paragraph, taking $\conv{\Delta_p^\AM}$ captures the power of $\AM$ provers.

For another more structured example, suppose $\calS$
is equal to the Cartesian product $\mathcal{Z} \times \cdots \times \mathcal{Z}$ with $m$ copies, we can take
$\calC = \conv{ \{\mu^{\otimes k} \mid \mu \in \Delta_{\mathcal{Z}}\} }$. In words, $\calC$ is the convex hull of of i.i.d. distributions.
Using this notation, we can for instance take $ \conv{\set{ \mu^{\otimes m} : \mu \in \Delta^\AM}}$,
or $\conv{\set{ \mu_1^{\otimes m_1} \otimes \cdots \otimes  \mu_\ell^{\otimes m_\ell} : \mu_1,\ldots, \mu_\ell \in \Delta^\AM}}$. This can be thought of as capturing $\AM$ with multiple independent provers.
Analogously, $\propAM(m)$ for multiple provers.

It then follows that
\begin{corollary}\label{cor:AM-k-lower-bound}
    For $\propAM$ tester with even unbounded independent provers and unbounded rounds, to solve $\textup{GapSupp}_{s,\ell}$ with a constant advantage, where $\ell >s$, the sample complexity $t$ and the proof complexity $p$, must satisfy
        \[t (p + t) = \Omega(s).\]
\end{corollary}

\subsection{An $\IP$ Protocol}
\label{subsec:private-coin-protocol}
In the previous subsection, we proved a tight lower bound for $\textup{GapSupp}$ with classical certificates. The lower bounds holds even when allowing public-coin $\AM$ type certificates. In this section, we point out that the public-coin restriction is important as there is a very efficient private-coin $\AM$ tester, i.e., two-turn $\IP$ tester. This tester is adapted from~\cite{HR22label-inv}.

\begin{theorem}[cf. \cite{HR22label-inv}]\label{thm:IP-upper-bound}
    There is a private-coin  $\AM$ tester, e.g. Algorithm~\ref{proc:gapsupp-tester}, for $\textup{GapSupp}_{\frac{N}{3},\frac{2N}{3}}$ using $O(1)$ samples and $O(1)$ communication.
\end{theorem}

\noindent\fbox{
\begin{myalg}[A Private-coin $\IP$ Tester for $\textup{GapSupp}_{N/3, 2N/3}$]\label{proc:gapsupp-tester}\ignorespacesafterend
\textbf{Input:} Unknown distribution $\mu$ 
\vspace{3mm}

\textbf{Arthur:} 
\begin{enumerate}
    \item Make $k$ samples from $\mu$ for some large enough constant $k$, denote the set of elements sampled from $\mu$ by $S$;
    \item Make $k$ samples uniformly at random from $[N]$, denote the set $R$;
    \item Send the set $M:=S\cup R$ to Merlin in a random order.
\end{enumerate} 

\vspace{3mm}
\textbf{Merlin:} Merlin return a subset $M'\in M$ such that $M' = M\cap \supp \mu$

\vspace{3mm}
\emph{Accept} if $|M'| \le 1.5 k$ and $S\subseteq M'$. 

\end{myalg}}

\begin{proof}
For large enough $N$, almost surely, the set $M$ contains $2k$ elements. If the prover is honest, $M' = M\cap \supp \mu$. If $\supp\mu = N/3$, by chernoff bound, $|M'| \le |S| + |R|/2$ with probability at least $1-\exp(-\Omega(k))$. This establishes the completeness case.

In the soundness case, $\supp \mu \ge 2N/3$. Then with probability $1-\exp(-\Omega(k))$, $|M\cap \supp\mu|\ge |S|+7|R|/12 = 1.5k + k/12$. To fool Arthur, Merlin needs to send $|M'|\le 1.5k$, thus with probability $1-\exp(-\Omega(k))$,
\[
   \left| (M\cap \supp \mu) \setminus M' \right| \ge \frac{k}{12}.
\]
However, the verifier has no information about which element in $M\cap \supp \mu$ is in $S$, that is $M'$ is determined by $M\cap \supp \mu$ and independent of $S$. Note that $S$ is just a uniformly random subset of $M\cap \supp \mu$ of size $k$. Thus, with probability at most $\exp(-\Omega(k))$, $S\subseteq M'$. It concludes that in the soundness case, Arthur accepts with probability at most $\exp(-\Omega(k)).$
\end{proof}

Note that simply treating the quantum object subset state as a distribution, the above protocol implies an $\IP$ protocol for testing the subset state with small support size.

\subsection{One-sidedness of the $\textup{GapSupp}$}
\label{subsec:one-sideness-GapSupp}
So far, we considered the upper and lower bounds for $\textup{GappSupp}_{s,\ell}$,  where $s < \ell$. In words, we wanted to test that the flat distribution has a small support. We now justify this choice. In particular, suppose we want to test whether a given distribution has large support, e.g., consider $\text{GapSupp}_{\ell, s}$ for $\ell > s$, then the proof does not improve the testability at all no matter how long it is.

\begin{theorem} \label{theorem:yes_case_large}
    Consider two ensembles of distributions for some parameters $\ell \gg s$.
    Given $t=o(\sqrt{s})$ samples from the distribution, there is no tester that can solve $\textup{GapSupp}_{\ell, s}$ with proof of arbitrary length.
\end{theorem}

\begin{proof}
    Given any proof $\pi$, take a  random distribution $\mu_L \in \cF_\pi$ where $L$ is the support of $\mu_L$. Let $\nu_0$ be the distribution on $t$ samples that a tester sees when first sample $\mu_L \in \cF_\pi$, then sample $t$ elements from $\mu_L$.

    Consider an adversarial strategy. Given proof $\pi$, the adversary chooses $\mu_L\in \cF_\pi$; then he chooses $S\subseteq L$ with $|S|=s$. Provide the distribution $\mu_S\in \cE_1$ to the tester. Let $\nu_1$ be the distribution of what tester sees when sample $t$ elements from such experiment.

    It's easy to see that as long as $t\ll \sqrt{s}$, $\nu_0$ is statistically close to $\nu_1$.
\end{proof}

Note that this argument can be adapted even if we consider interactive proof rather than the $\propMA$ model. Suppose there is some interactive proof type tester $\Pi$ that accepts the uniform distribution with high probability. 
Now for any distribution $\mu_S$, where $S$ is of small support. As long as the sample complexity $s \ll \sqrt{N}$ is in the non-collision regime. There is a trivial adversary strategy:
Inductively, say at round $t$, $\tau_t$ is the transcript communicated between the verifier and the prover so far. Maybe the verifier will make a few additional sample $S_t\sim\mu$ and roll some dice $R_t$, and send message $m_t = m_t(\tau_t, R_1, R_2,\ldots, R_t,  S_1, S_2, \ldots, S_t)$. 
The adversary prover receiving $m_t$ would simply respond pretending the underlying distribution is uniform.

\begin{corollary}
    To test uniformity, the sample complexity is $\Omega(\sqrt N)$ even in the $\textup{PropIP}$ model.
\end{corollary}

%% file: deMerlinization.tex
\section{Testing with a General Quantum Proof: Yet Another Fiasco}\label{sec:prop_propqma_collapse}

In~\cref{sec:classical-lower-bound}, we discussed at length that for the $\textup{GapSupp}$ problem, assisted with a standard $\MA$ type proof, or even some very general structured certificates which include $\AM$ type proof, would not improve the testability unless the certificates are exponentially long. 

One might ask how general such phenomenon is, if it holds for any property, and for any certificates even structured ones. Our answer is two-fold: There are good news and bad news. For $\MA$ type proof, in fact, for $\QMA$ type proof, the proof would not improve the testability significantly for any property of interest. We prove that quantum proof does not increase testability, and it subsumes the classical proof.
This fact is adapted from the de-Merlinization ideas of Aaronson~\cite{A06} and the follow-up work of Harrow~\etal~\cite{HLM17}. 

Let us review some well-known facts on quantum information.
\begin{fact}[Gentle Measurement Lemma~\cite{winter1999coding}]
\label{fact:gentle-measurement}
    Let $\rho$ be a quantum state and $0\preceq \Lambda \preceq I$. Then
    \begin{equation}
        \left\|\rho - \frac{\sqrt \Lambda \rho \sqrt \Lambda }{\trace\Lambda \rho}\right\|_{\trace}^2 \le \sqrt{\trace(I - \Lambda)\rho}.
    \end{equation}
\end{fact}
\begin{fact}[Quantum Union Bound~\cite{A06}]
\label{fact:quantum-union}
    Suppose $\cM_1 =\{\Lambda_1, I-\Lambda_1\}, \cM_2=\{\Lambda_2, I-\Lambda_2\}, \ldots, \cM_n=\{\Lambda_n, I-\Lambda_n\}$ are some 2-outcome POVM measurements, where $\Lambda_i$ corresponds to accept. Let $\rho$ be some (mixed or pure) quantum state, such that
    \begin{equation*}
        \trace \Lambda_i \rho \ge 1-\epsilon,\qquad i\in \{1,2,\ldots, n\}.
    \end{equation*}
    Apply $\cM_1, \cM_2,\ldots, \cM_n$ to $\rho$ sequentially, then
    \begin{equation*}
        \Pr[\exists i,\; \cM_i \text{ rejects}] \le n\sqrt\epsilon.
    \end{equation*}
\end{fact}

\paragraph{Gap Amplification for $\propBQP,\propQMA$.}
We collect some facts on gap amplification for $\propBQP$ and $\propQMA$.
\begin{theorem}[Gap amplification]
\label{thm:prop-gap-amp}
For any $0<\alpha<\beta<1$, and $\epsilon\ge 0$ such that $1-2\epsilon > \beta-\alpha$, let $\gamma= \beta-\alpha$. Then,
    \begin{align}
        \propBQP_{\alpha \mcom \beta}[k,w] &\subseteq \propBQP_{1-\epsilon\mcom \epsilon}\left[ O\left(\frac{k}{\gamma^2}\log \frac{1}{\epsilon}\right)\mcom O\left(\frac{w}{\gamma^2}\log\frac{1}{\epsilon} \right) \right]. 
        \label{eq:amp-naive-bqp}
        \\
        \propQMA_{\alpha\mcom \beta}[k,w] &\subseteq \propQMA_{1-\epsilon\mcom \epsilon}\left[ O\left(\frac{k}{\gamma^2}\log \frac{1}{\epsilon}\right)\mcom O\left(\frac{w}{\gamma^2}\log\frac{1}{\epsilon} \right) \right]. 
        \label{eq:amp-naive}
    \end{align}
For any $0<\epsilon<1/2$,
    \begin{align}
        \propQMA_{1-\epsilon \mcom \epsilon}[k,w] &\subseteq \propQMA_{1-t\sqrt{\epsilon} \mcom \epsilon^t}[kt, w].
        \label{eq:amp-quantum-union}
    \end{align}
\end{theorem}

\begin{proof}
    (\ref{eq:amp-naive-bqp}) is the standard amplification by running the tester multiple of times and apply chernoff bound; and (\ref{eq:amp-naive}) is the standard amplification for $\QMA$ model by asking for more copies of the state and longer proof using the fact that in the soundness case, no proof can fool the verifier. 
    
    (\ref{eq:amp-quantum-union}) is the amplification that use fresh copies of quantum states but reuse the quantum proof. In particular, suppose $T$ is a tester for $\propQMA_{1-\epsilon, \epsilon}[k,w]$, then the new tester $T'$ is as follows 
    \begin{enumerate}[label=(\roman*)]
        \item Run $T$ for $t$ times, each time with a fresh $k$ copies of the target state and reuse the quantum witness.
        \item Accept only if all the $t$ run of $T$ accepts.
    \end{enumerate}    
    Clearly the new tester $T'$ uses $kt$ copies of the state and asks for a witness of length $w$. Suppose the given state $\ket\psi\not\in\calP$, then soundness of $T'$ becomes $\epsilon^t$ by the soundness of $T$. As for the completeness, it follows  ~\cref{fact:quantum-union}, the quantum union bound, that a single $T$ rejects with probability at most $\epsilon$, thus the $t$ sequential run has at least one reject with probability at most $t\sqrt{\epsilon}$. 
\end{proof}

\paragraph{De-Merlinization for $\propQMA$ in the Small-soundness Regime.}
Next we show that suppose the soundness is exponentially small in $\propQMA$, then the quantum witness can be removed.
\begin{theorem}[Partial de-Merlinization]
\label{thm:demerlin-special}
For any $0<\epsilon<1$, and $0 < \delta < 1-\epsilon$,
    \begin{align}
        \propQMA_{1-\epsilon \mcom \delta}[k,w]
        &~\subseteq~ \propBQP_{\frac{(1-\epsilon)^2}{4}\mcom \frac{2\delta \cdot 2^w}{1-\epsilon}}[k].\label{eq:amp-quantum-or}
    \end{align}
\end{theorem}
\begin{proof}
    Suppose we have a measurement $\cM = \{\Lambda, I-\Lambda\}$ corresponding to a tester for $\propQMA_{1-\epsilon, \delta}[k,w]$. Then we can write down the implications of what it means, 
    \begin{description}
        \item[(Yes)] If $\ket\phi\in \calP$, then there is some $\ket\pi\in \C^{2^w}$,
        \begin{align}
            \trace \left(\Lambda (\phi^{\otimes k} \otimes \pi)\right) \ge 1-\epsilon.\label{eq:demerlin-witness}
        \end{align}
        \item[(No)] If $\ket\phi\not\in \calP$, then for any $\ket\pi\in \C^{2^w}$,
        \begin{align}
            \trace \left(\Lambda (\phi^{\otimes k} \otimes \pi)\right) \le \delta. \label{eq:demerlin-qma-soundess}
        \end{align}
    \end{description}
    Call the register indicating the proof $\pi$ as $P$, let $e_1,e_2,\ldots, e_{2^w}$ be the computational basis for $P$, take 
    \[
     d:= 2^w, \quad 
     \rho := \phi^{\otimes k},\quad 
     \Lambda_i := \bra{e_i}_P \Lambda \ket{e_i}_P, \quad 
     \tilde\Lambda = \Exp_i \Lambda_i.\] 
    Let $\Pi$ be the projector onto 
    \[
        \cH:=\Span\left\{\ket\sigma: \tilde\Lambda \sigma > \frac{1-\epsilon}{2d}  \right\},
    \]
    and $\Pi^\perp$ denote the projector onto the subspace orthogonal to $\cH$. Immediately,
    \begin{equation}\label{eq:demerlin-Pi-bound}
        \frac{2d}{1-\epsilon}\tilde \Lambda \succeq \Pi.
    \end{equation}
    We show that $\{\Pi, I-\Pi\}$ is the ideal measurement for $\propBQP$.
    
    In the (Yes) case, we have the witness $\ket\pi$  satisfying (\ref{eq:demerlin-witness}). Let 
    \[
        \tilde\rho:=\frac{\Pi^\perp \rho \Pi^\perp}{\trace \Pi^\perp \rho}.
    \]
    Then,
    \begin{align*}
        \frac{1-\epsilon}{2d} 
        &\ge \trace \tilde\Lambda \tilde\rho 
        & (\text{By definition of }\Pi^\perp)
        \\
        & =\frac{\sum_i \trace \Lambda (\tilde\rho\otimes |e_i\rangle\langle e_i |)}{d} 
        & (\text{By definition of }\tilde\Lambda)
        \\
        &= \frac{\trace \Lambda (\tilde\rho\otimes I)}{d} 
        \\
        &\ge \frac{\trace\Lambda (\tilde\rho \otimes \pi)} {d} & (I\succeq\pi)
        \\
        &\ge \frac{\trace \Lambda (\rho\otimes \pi) - \|\rho\otimes \pi - \tilde\rho\otimes \pi\|_{\trace}}{d}. & (\textup{By~\cref{fact:trace_norm_acc}})
    \end{align*}
    Combining the above calculation with (\ref{eq:demerlin-witness}), we obtain
    \begin{align*}
        \left(\frac{ 1-\epsilon}{2} \right)^2
        & \le  \|\rho\otimes \pi - \tilde\rho \otimes \pi \|_{\trace}^2  =  \|\rho - \tilde\rho \|_{\trace}^2 
        \le \trace \Pi \rho,
    \end{align*}
    where the last step follows the gentle measurement lemma~\cref{fact:gentle-measurement}.

    We now turn to the (No) case.
    \begin{align*}
        \trace\Pi\rho 
        &\le \frac{2d}{1-\epsilon} \trace \tilde\Lambda \rho
        & (\text{By }(\ref{eq:demerlin-Pi-bound}))
        \\
        &=\frac{2d}{1-\epsilon} \Exp_i \trace \Lambda(\rho \otimes \ket{e_i}\bra{e_i})
        \\
        &\le \frac{2d\delta}{1-\epsilon}.
        &(\textup{By}~(\ref{eq:demerlin-qma-soundess}))
    \end{align*}
    Since the verifier can construct $\Pi$ based on $\cM$. That finishes our proof.
\end{proof}

\paragraph{A Full De-Merlinization for $\propQMA$.}
Putting the previous results together, one can deMerlinize any $\propQMA$ tester.

\begin{corollary}[Full de-Merlinization]
\label{cor:demerlin-full}
For any $0<\alpha<\beta<1$, and $\epsilon\ge 0$ such that $1-2\epsilon > \beta-\alpha$, let $\gamma= \beta-\alpha$. Then    \begin{align}
        \propQMA_{\alpha \mcom \beta}[k,p] 
        ~\subseteq~ \propBQP_{1-\epsilon \mcom \epsilon}\left[ 
            kp\log p \cdot  O\left(
               \frac{1}{(\beta-\alpha)^2}\log\frac{1}{\epsilon} 
            \right)
        \right].
    \end{align} 
\end{corollary}

\begin{proof}
    The theorem now follows (\ref{eq:amp-naive})-(\ref{eq:amp-quantum-or}) by setting the parameters properly. In particular, choose some large enough constant $C$, set
    \begin{equation*}
        \gamma=\beta - \alpha,
        \quad
        \delta = C^{-2} p^{-2},
        \quad 
        t =  \frac{Cp}{2\gamma^2}, 
        \quad
        W = O\left(\frac{p}{\gamma^2}\log\frac{1}{\delta}\right),
    \end{equation*}
    Then
    \begin{align*}
        \propQMA&_{\alpha \mcom \beta}[k,p] 
        \\
        &\subseteq~ \propQMA_{1-\delta \mcom \delta}\left[ O\left(\frac{k}{\gamma^2}\log \frac{1}{\delta}\right),~ W \right]
        & (\textup{By (\ref{eq:amp-naive}) in~\cref{thm:prop-gap-amp}})
        \\
        &\subseteq~ \propQMA_{\frac{1}{2} \mcom \delta^{t}}\left[ O\left(\frac{kt}{\gamma^2}\log \frac{1}{\delta}\right),~ W \right]
        & (\textup{By (\ref{eq:amp-quantum-union}) in~\cref{thm:prop-gap-amp}})
        \\
        &\subseteq~ \propBQP_{\frac{1}{16} \mcom 4\delta^t 2^W}\left[ O\left(\frac{kt}{\gamma^2}\log \frac{1}{\delta}\right)\right].
        & (\textup{By~\cref{thm:demerlin-special}})
    \end{align*}
    For large enough $C$, we have $4\delta^t 2^W < 1/32$.
    Now theorem follows as
    \begin{align*}
        \propQMA_{\alpha \mcom \beta}[k,p] 
        ~&\subseteq~
        \propBQP_{\frac{1}{16} \mcom \frac{1}{32}}\left[ O\left(\frac{kt}{\gamma^2}\log \frac{1}{\delta}\right)\right]
        \\
        &\subseteq~
        \propBQP_{1-\epsilon,\epsilon}\left[
            O\left(\frac{kt}{\gamma^2}\log\frac{1}{\delta}\log\frac{1}{\epsilon}\right)
        \right],
    \end{align*}
    where the last inclusion is due to (\ref{eq:amp-naive-bqp}).
\end{proof}

%% file: supp_prot.tex
\section{Testing with Structured Quantum Certificates: A Triumph}
\label{sec:flat-cert}
In the previous section, we announced the bad news that general $\QMA$ proofs would not improve testability for any properties significantly. 
We now turn to the good news: With naturally structured quantum certificates, one can increase the testability for testing subset state of different sizes dramatically.
This is in sharp contrast to the classical problem $\textup{GapSupp}$, as in~\cref{subsec:structured-classical-certificate}, we proved that structured certificates in an abstract manner as convex subset of $\Delta_\calS$ gives almost no gain in terms of testability when the certificates is short.

We focus on one natural restriction: Restricting the certificates to be subset states. Namely, both honest and adversary prover can only send subset state. We define the corresponding property testing model $\propQMA_{\textup{subset}}(m)_{a,b}[k, pm]$ analogously to~\cref{def:qma-prop}.

\begin{definition}[Quantum Property Testing with Certificates]\label{def:qma-flat-prop}
  For $d\in \mathbb N$, let $k=k(d),p=p(d) : \mathbb{N} \to \mathbb{N}$, $1 \ge a > b \ge 0$. A property $\calP = \sqcup \calP_d$ belongs to $\propQMA_{\textup{subset}}(m)_{a,b}[k,mp]$
  if there exists a verifier $V$ such that for every $\ket{\psi} \in \mathbb{C}^d$,
  \begin{enumerate}
    \item if $\ket{\psi} \in \calP_d$, then there exist $m$ subset states $\ket{\phi_1},\ldots,\ket{\phi_m} \in \mathbb{C}^{2^p}$ such that $V(\ket{\psi}^{\otimes k}\otimes \ket{\phi_1} \otimes \cdots \otimes \ket{\phi_m})$ accepts with probability at least $a$, and
    \item if $\ket{\psi}$ is $\epsilon$-far from $\calP_d$ (in trace distance), then then for any subset states $\ket{\phi_1},\ldots,\ket{\phi_m} \in \mathbb{C}^{2^p}$, $V(\ket{\psi}^{\otimes k}\otimes \ket{\phi_1} \otimes \cdots \otimes \ket{\phi_m})$
          accepts with probability at most $b$.
  \end{enumerate}
\end{definition}

This section aims to prove that (adversarial) flat quantum certificates allow us to
obtain a multiplicative approximation to the support size of a flat state. More precisely, we
show the following.

\begin{theorem}[Effective Quantum Certified Testing]\label{thm:eff_q_cert_testing}
  With just polynomially many (\ie $\poly(n)$) copies and certificates of flat amplitudes, a polynomial time quantum tester can, 
  \begin{enumerate}
      \item  either certify the support size  of an arbitrary subset state (the target state) is $s$ within a $(1\pm \epsilon)$ multiplicative factor for any constant
      $\epsilon > 0$,\footnote{Or even any $\epsilon=\Omega(1/\poly(n))$.}
    \item   or detect that the certificates are malicious.
  \end{enumerate}
\end{theorem}

Note that it suffices to consider the case where there is only one copy of the target state. We show that if the target state has the correct support size $s$, which is part of the proof, the certificates are accepted with a probability exponentially close to 1; while if the prover tries to fool the verifier that the target state has support size $s$ which is $\epsilon$-far from the correct one, then the verifier will reject with a probability constant away from 1.
\cref{thm:eff_q_cert_testing} implies that with polynomial size certificates (at least restricted), the ensembles studied in~\cref{sec:quantum-indist} are distinguishable, using even just a single copy of the state.

For simplicity, assume that $2^n/s$ is a power of 2.
 Our testing strategy is simple. However, a lot of care needs to be taken to handle the adversarial situation. Therefore, we start by presenting the overall idea.

\subsection{Proof Outline}
Let $\rho=\phi_T$ be the target subset state with support $T$, for which we want to certify its support size. The prover will send classical $\ell$ supposedly equal to $n -\log s$. 
Furthermore, the prover will be asked to provide for $i=1,2,\ldots, \ell$ states $\phi_i$ supposedly equal to some subset state $\phi_{S_i}$, such that 
    \begin{align}\label{eq:cascading}
        & T = S_\ell \subseteq S_{\ell-1} \subseteq \cdots \subseteq S_1 \subseteq S_0= [2^n],
        \nonumber\\
        & |S_i| = 2 |S_{i+1}|, & i = 0, 1, 2, \ldots, \ell-1.
        \nonumber
    \end{align}
The task reduces to testing whether indeed the support of the given states halves each time.
This motivates a key technical lemma establishing that the support size of two subset
states, $\ket{\phi_H}$ and $\ket{\phi_S}$, satisfies $S =  |H| / 2$ . The intent is that $S \subset H$ and $\abs{S}/\abs{H} \approx 1/2$. 

\begin{lemma}[Support Halving Lemma (informal)]
  Suppose $\ket{\phi_H}, \ket{\phi_S}, \ket{\phi_{S'}}$ are subset states satisfying
  \begin{enumerate}
       \item $\abs{\braket{\phi_S}{\phi_{S'}}}^2 = \frac{1}{\poly(n)}$,
       \item $\left\lvert \abs{\braket{\phi_H}{\phi_S}}^2 -\frac{1}{2}\right\rvert = \frac{1}{\poly(n)}$, 
       \item $\left\lvert \abs{\braket{\phi_H}{\phi_{S'}}}^2-\frac{1}{2}\right\rvert = \frac{1}{\poly(n)}$.
  \end{enumerate}
  Then, we have
  \begin{align*}
    \frac{\abs{S}}{\abs{H}} = \frac{1}{2} \pm \frac{1}{\poly(n)}\mper
  \end{align*}
\end{lemma}

It is evident how the above lemma will be used
to design an algorithm to approximate the support size of a subset state $\ket{\phi_T}$. We will apply the support halving lemma iteratively. Start with $i=0$, $S_0 = \set{0,1}^n$; for any $i\ge 0$, $S_{i+1}$ can be any subset satisfying
$T \subset S_{i+1} \subset S_i$ and 
\begin{equation}\label{eq:halving}
    \frac{\abs{S_{i+1}}}{\abs{S_{i}}} = \frac{1}{2}\pm\frac{1}{\poly(n)},
\end{equation}  
which will be guaranteed using a test that implements \cref{lem:key_subset_mass}.
Then, we obtain the telescoping multiplication
\begin{align*}
  \abs{S_\ell} \approx \abs{S_0} \frac{\abs{S_1}}{\abs{S_0}} \frac{\abs{S_2}}{\abs{S_1}} \cdots \frac{\abs{S_{\ell}}}{\abs{S_{\ell-1}}} = \frac{1}{2^\ell} \left(1 \pm \frac{1}{\poly(n)}\right)^\ell \abs{S_0}\mper
\end{align*}
As long as the $1/\poly(n)$ is small enough, $(1+1/\poly(n)) ^\ell = 1\pm\epsilon$. Therefore $T = S_\ell$ has size about $2^{n-\ell}$.

To obtain good estimates on 
\[
    \langle \phi_{S_i}, \phi_{S_{i+1}}\rangle, 
    \quad
    \langle \phi_{S_i}, \phi_{ S'_{i+1}} \rangle,
    \quad
    \langle \phi_{S_{i+1}}, \phi_{ S'_{i+1}} \rangle
\]
for $i=0,1,2,\ldots, \ell-1$, we will ask to prover to provide many copies of each of these subset states. In particular, the prover should supply collections of copies of states
\begin{align*}
    \Phi_i, \Psi_i,\qquad i=1,2,\ldots, \ell,
\end{align*}
where $\Phi_i$ corresponds to $m$ copies of states $\phi_{S_i}$, and $\Psi_i$ corresponds to $m$ copies of states $\phi_{ S'_i}$. Supposedly, $S'_i = S_{i-1} \setminus S_i $.
Then Chernoff bound tells us that with probability at most $\exp(-\Omega( m / \poly(n)))$, the estimate differs from the actual correlation by at most $1/\poly(n)$. Hence choose $m = \poly(n)$ would suffice.

\subsection{$\delta$-tilted States and Symmetry Test}
The first challenge to face is that once we are dealing with proofs that are supposed to supply copies of the same states, we need to deal with the adversarial situation where the proofs are not the same states. Here we bring the $\delta$-\emph{tilted states} and \emph{symmetry test} from~\cite[Section 4]{JW23}.

\begin{definition}[$\delta$-tilted states]
A collection of states $|\psi_1\rangle, |\psi_2\rangle, \ldots, |\psi_k \rangle$ defined on a same space is an $\delta$-tilted state
if there is a subset $R\subseteq[k]$ such that $|R|\ge (1-\delta)k$ and for any $i,j \in R$,
\[
\TD(|\psi_i\rangle, |\psi_j\rangle) \le \sqrt\delta.
\]
Furthermore,  we call $ \ket{\psi_i}$ a \emph{representative state} for any $i\in R$, and the subset $\{ |\psi_i\rangle: i\in R\}$ the
\emph{representative set}.
\end{definition}

The symmetry test are used to test if the collection of states $\Phi$ provided by the prover are essentially the same, i.e., being $\delta$-tilted, or not.

\noindent\fbox{
\begin{myalg}[Symmetry Test]\label{proc:sym-test}\ignorespacesafterend
\textbf{Input:} $ \Phi = \{ \phi_{1},\phi_{2},\ldots,\phi_{m} \}$, a collection of pure states for some even number $m$.
\begin{enumerate}
\item Sample a random matching $\pi$ within $1, 2, \ldots, m$.
\item SwapTest on the pairs based on the matching $\pi$.
\end{enumerate}
\emph{Accept} if all SwapTests accept. 

\end{myalg}}

\begin{theorem}[Symmetry Test~\cite{JW23}]\label{thm:symmetry-test}
Suppose $\Psi$ is not an $\delta$-tilted state.
Then the symmetry test passes with probability at most $\exp(-\Theta(\delta^{2}m)).$
On the contrary, for $0$-tilted state $\Psi$, the symmetry test accepts with probability $1$.
\end{theorem}

For a collection of states $\Phi$, one can view it as a mixed state, or simply a set. We will take both perspectives. Next, we collect some facts about $\delta$-tilted collection $\Phi$. The first one states that we can combine two $\delta$-tilted states into one in the natural way.

\begin{proposition}[Tensorization of tilted states~\cite{JW23}]\label{prop:tensor-tilted-states}
If $\Psi$ is an $\delta$-tilted state and $\Phi$ is a $\gamma$-tilted state, and $|\Psi|=|\Phi|=m$. Then $\Psi\otimes\Phi$ is an $(\delta+\gamma)$-tilted state, where
\begin{equation*}
    \Psi\otimes\Phi = \{\psi_i \otimes \phi_i: 
        \Psi = \{\psi_1, \psi_2, \ldots, \psi_m\},
        \Phi = \{\phi_1, \phi_2, \ldots, \phi_m\}
    \}.
\end{equation*}
\end{proposition}

The second fact states that the bahavior of $\delta$-tilted states is like that of any representative state. This fact has two folds: If $\Psi$ is viewed as mixed states, $\Psi$ is close to its representative state in trace distance; if $\Psi$ is viewed as sets, then concentration holds.
\begin{proposition}[\cite{JW23}]\label{prop:tilted-state-approx}
For any quantum algorithm $\cA$, let $\cA(\ket{\psi})$ denote the probability that $\cA$ accepts $\ket{\psi}$. Let $\Psi$ be an $\delta$-tilted state, and $\ket{\psi}$ any representative state of $\Psi$. Then
\begin{equation}\label{eq:approx-tilted-state}
    |\cA(\ket{\psi}) - \cA(\Psi)| \le 3\sqrt \delta. 
\end{equation}
Furthermore, when apply $\cA$ to $\Psi$, let $\alpha$ be the fraction of accepted executions of $\cA$. Then
\begin{equation}\label{eq:concentration-tilted-state}
    \Pr[ |\alpha - \cA(\Psi)| \ge \sqrt \delta] \le \exp(-\delta|\Psi|/2),
\end{equation}
and therefore,
\begin{equation}\label{eq:approx-reprensetative-state}
    \Pr[ |\alpha - \cA(\ket{\psi})| \ge 4\sqrt \delta] \le \exp(-\delta|\Psi|/2).
\end{equation}
\end{proposition}

\subsection{Subset Test}
We now proceed to prove the support halving lemma, and present the subset test which will help us test~(\ref{eq:halving}).
\begin{lemma}[Support Halving Lemma]\label{lem:key_subset_mass}
  Let $\mu \in (0,1)$ be a constant and $\delta \in (0, C \mu^4)$, where $C > 0$ is universal constant.
  Suppose $\ket{\phi_H}, \ket{\phi_S}, \ket{\phi_{S'}}$ are subset states satisfying
  \begin{enumerate}
       \item $\abs{\braket{\phi_S}{\phi_{S'}}}^2 \le \delta$,
       \item $\left\lvert \abs{\braket{\phi_H}{\phi_S}}^2 -\mu\right\rvert \le \delta$, 
       \item $\left\lvert \abs{\braket{\phi_H}{\phi_{S'}}}^2-(1-\mu)\right\rvert \le \delta$.
  \end{enumerate}
  Then, we have
  \begin{align*}
    \frac{\abs{S}}{\abs{H}} = (\mu \pm O(\delta^{1/4}))\mper
  \end{align*}
\end{lemma}
\begin{proof}
  Using the first assumption, we have
  \begin{align}\label{eq:supp:intersec_ub}
    \delta  \ge \abs{\braket{\phi_S}{\phi_{S'}}}^2 = \frac{\abs{S \cap S'}^2}{\abs{S}\abs{S'}} \mper
  \end{align}
  The second assumption states that
  \begin{align*}
    \delta \ge \left\lvert \abs{\braket{\phi_H}{\phi_S}}^2 -\mu\right\rvert =  \left\lvert \frac{\abs{S \cap H}^2}{\abs{S}\abs{H}} - \mu\right\rvert\mcom
  \end{align*}
  in particular, it implies
  \begin{align}\label{eq:supp:ratio_s_lb}
   \mu - \delta \le \abs{\braket{\phi_H}{\phi_S}}^2  =  \frac{\abs{S \cap H}^2}{\abs{S}\abs{H}}  \le \min\left\{ \frac{\abs{H}}{\abs{S}}, \frac{\abs{S}}{\abs{H}}, \frac{|S\cap H|}{|H|}
   \right\}\mcom
  \end{align}
  since $\abs{S \cap H} \le \min\left\{ \abs{S}, \abs{H} \right\}$.
  Similarly, from the third assumption
  \begin{align*}
    \delta \ge \left\lvert \abs{\braket{\phi_H}{\phi_{S'}}}^2 -\mu\right\rvert =  \left\lvert \frac{\abs{S' \cap H}^2}{\abs{S'}\abs{H}} - (1-\mu) \right\rvert
  \end{align*}
  we obtain
  \begin{align}\label{eq:supp:ratio_sprime_lb}
   (1-\mu) - \delta \le \abs{\braket{\phi_H}{\phi_{S'}}}^2 \le \min\left\{\frac{\abs{S'}}{\abs{H}}, \frac{\abs{H}}{\abs{S'}},
   \frac{|S'\cap H|}{|H|}
   \right\}\mper
  \end{align}  
  Using bounds from~\cref{eq:supp:ratio_s_lb} and \cref{eq:supp:ratio_sprime_lb} in~\cref{eq:supp:intersec_ub}, we get 
  \begin{align} \label{eq:S-cap-S-prime}
    \frac{\delta}{(\mu -\delta) ((1-\mu) -\delta)}  \ge \frac{\abs{S \cap S'}^2}{\abs{H}^2} \ge \frac{\abs{S \cap S' \cap H}^2}{\abs{H}^2} \mper
  \end{align}    

  Let $\gamma = \abs{S \cap H}/\abs{S}$ and $\gamma' = \abs{S' \cap H}/\abs{S'}$. 
  From the second and third assumptions, we have
  \begin{align}
    \nonumber
    1 \pm 2\delta &= \frac{\abs{S \cap H}^2}{\abs{S}\abs{H}} + \frac{\abs{S' \cap H}^2}{\abs{S'}\abs{H}} \\
                  \nonumber
                  &= \gamma \frac{\abs{S \cap H}}{\abs{H}} + \gamma' \frac{\abs{S' \cap H}}{\abs{H}}\\
                  \nonumber
                  &= \gamma \frac{\abs{(S \setminus S') \cap H}}{\abs{H}} + \gamma' \frac{\abs{(S'\setminus S) \cap H}}{\abs{H}} + (\gamma +\gamma') \frac{\abs{S \cap S' \cap H}}{\abs{H}} \\
                  \nonumber
                  &= \gamma \alpha  + \gamma' \beta + (\gamma +\gamma') \frac{\abs{S \cap S' \cap H}}{\abs{H}}\\
                  &= \gamma \alpha  + \gamma' \beta + O(\sqrt{\delta})\mcom \label{eq:gamma_alpha_beta}
  \end{align}
  where the $O(\sqrt{\delta})$ bound follows from~\cref{eq:S-cap-S-prime}, and we set
  $\alpha = \abs{(S\setminus S') \cap H}/\abs{H}$ and $\beta = \abs{(S'\setminus S) \cap H}/\abs{H}$.
  In view of (\ref{eq:S-cap-S-prime}) and (\ref{eq:supp:ratio_s_lb}), we have 
  \[\alpha \ge \mu - O(\sqrt{\delta})\] 
  and similarly, using (\ref{eq:S-cap-S-prime}) and (\ref{eq:supp:ratio_sprime_lb}), 
  \[\beta \ge (1-\mu) - O(\sqrt{\delta}).\]
  Since $1 \ge \alpha + \beta$ and $\alpha,\beta \ge 0$, we decude that $\alpha = \mu \pm O(\sqrt{\delta})$ and $\beta = (1-\mu) \pm O(\sqrt{\delta})$.
  Using (\ref{eq:gamma_alpha_beta}) for the first equality, and since $\mu = \Omega(\delta^{1/4})$, and $\gamma, \gamma' \le 1$,
  \begin{align*}
    1 \pm O(\sqrt{\delta}) = \gamma \alpha  + \gamma' \beta = \gamma \mu  + \gamma' (1-\mu) \pm O(\sqrt{\delta}) \le 1 - \mu (1-\gamma) \pm O(\sqrt{\delta})\mcom
  \end{align*}
  we conclude that $\gamma \ge 1 - O(\delta^{1/4})$.
  From this, we get $\abs{S \cap H} = (1-O(\delta^{1/4})) \abs{S}$.
  Using the second assumption, we get
  \begin{align*}
    \delta \ge \left\lvert \frac{\abs{S \cap H}^2}{\abs{S}\abs{H}} - \mu\right\rvert = \left\lvert (1-O(\delta^{1/4}))^2 \frac{\abs{S}}{\abs{H}} - \mu\right\rvert
    = \left\lvert (1-O(\delta^{1/4})) \frac{\abs{S}}{\abs{H}} - \mu\right\rvert
    \mcom
  \end{align*}
  or
  \begin{align*}
    O(\delta) \ge \left\lvert \frac{\abs{S}}{\abs{H}} - (1\pm O(\delta^{1/4})) \mu\right\rvert\mcom
  \end{align*}
  as desired.
\end{proof}

Below is a formal description of the subset test that incorporates our setting and implements~\cref{lem:key_subset_mass} as a test for~(\ref{eq:halving}). The parameter $\gamma$ will be specified later.

\noindent\fbox{
\begin{myalg}[SubSet Test]\label{proc:subset-test}\ignorespacesafterend
\textbf{Input: } Collections of states $\Phi_1, \Phi_2, \Psi_2$, and a target density some constant $\mu\in(0,1)$. Supposedly, the three collections corresponds to some subset state $\phi_H, \phi_S, \phi_{S'}$, respectively. 

Take the following steps:
\begin{enumerate}
    \item 
    Partition the each of collections into two parts of equal size:
    \begin{align*}
        \Phi_1 = \Phi'_1 \sqcup \Phi''_1,\quad
        \Phi_2 = \Phi'_2 \sqcup \Phi''_2,\quad
        \Psi_2 = \Psi_2' \sqcup \Psi''_2.
    \end{align*}    

    \item Estimate $|\langle\phi_H \mid \phi_S\rangle|^2$: Applying $m/2$ swap tests on $\{\Phi_1'\}\otimes\{\Phi_2'\}$. Let the fraction of accepted pairs be $\alpha$.
       
    \item Estimate $|\langle\phi_H \mid \phi_{S'}\rangle|^2$: Applying $m/2$ swap tests on $\{\Phi_1''\}\otimes\{\Psi_2'\}$. Let the fraction of accepted pairs be $\beta$.
       
    \item Estimate $|\langle\phi_S \mid \phi_{S'}\rangle|^2$: Applying $m/2$ swap tests on $\{\Phi_2''\}\otimes\{\Psi_2''\}$. Let the fraction of accepted pairs be $\zeta$. 
\end{enumerate}
\emph{Accept} if all the inequalities hold:
$    |(2\alpha -1) - \mu| \le \gamma; 
    |(2\beta-1) - (1-\mu) | \le \gamma; 
    |2\zeta-1| \le \gamma.
$
\end{myalg}}

\subsection{Subset Support Certification Algorithm and Analysis}

Now we present a formal description of the algorithm that used to certify support size of a given target state.
Set the parameters:
\begin{description}
    \item $\epsilon$, the estimate error tolerance parameter in \cref{thm:eff_q_cert_testing}, $\le 1/2$.
    \item $\delta$, the symmetry test error tolerance parameter as used in~\cref{lem:key_subset_mass},  $=\epsilon^{16} /(320^2 n^8)$.
    \item $\gamma$, the subset test error tolerance parameter, = $\epsilon^8 / (80 n^4)$.
    \item $m$, the size of each collections $\Phi_i, \Psi_i$, $=O(n^{16} / \epsilon^{32}).$
\end{description}
\noindent\fbox{
\begin{myalg}[Subset State Support Test]\label{alg:support-test}\ignorespacesafterend
\textbf{Input}:$\rho, \Phi_0,  \Phi_1, \Psi_1, \Phi_2, \Psi_2 \ldots \Phi_\ell, \Psi_\ell$

Apply one of the following tests:
    \begin{enumerate}
        \item Symmetry Test on $(\Phi_i, \Psi_i)$,  for all $i$;
        \item (Even) Subset Test on $(\Phi_{2i}, \Phi_{2i+1}, \Psi_{2i+1})$, for all $0 < i<\ell/2$;
        \item (Odd) Subset Test on $(\Phi_{2i+1}, \Phi_{2i+2}, \Psi_{2i+2})$, for all $0\le i<\ell/2$;
        \item Swap Test on $\rho$ and a random state $\phi\in\Phi_\ell$.
    \end{enumerate}
\emph{Accept} if the chosen test accepts.
\end{myalg}}

\begin{proof}[Proof of~\cref{thm:eff_q_cert_testing}]
    
Now we show the above Subset State Support Test satisfies~\cref{thm:eff_q_cert_testing}. The completeness is straightforward. The (i) symmetry test and (iv) swap test will pass with probability 1. In (ii) and (iii), each Subset Test will passes will probability $1-\exp(-\Omega(\gamma^2 m ))$. Overall, the test pass with probability $1-\ell \exp(-\Omega(\gamma^2 m )) = 1 - \exp(-\Omega(n^8))$, certifying the subset state has size $N/2^\ell$.

In the adversarial case, we consider the possible ways that the adversary may cheat and show eventually the testing will catch these cheats.

\underline{Case 1: Attack caught by (i) Symmetry Test}. If any collection $\Phi_i, \Psi_i$ of the states is not $\delta$-tilted, then by~\cref{thm:symmetry-test}, with probability at most $\exp(-\Omega(\delta^2 m))=\exp(-\Omega(1))$, it passes the symmetry test. From now on, assume that all the collections are $\delta$-tilted.

\underline{Case 2: Attack caught by (iv) Swap Test}. Consider $\Phi_\ell$, take any representative state $\phi\in\Phi_\ell$. Suppose that $\phi$ corresponds to some subset state $\phi_S$ with $|S| \not \in (1\pm \epsilon^2) |T|$, then
    \begin{equation*}
        |\langle \phi_S\mid \rho\rangle|^2 \le 1-\epsilon^2 ~\implies~
        \Pr[\text{Swap Test accepts } (\phi_S, \rho)] \le 1-\frac{\epsilon^2}{2}.
    \end{equation*}
Suppose $\phi$ is a representative state. Then by definition with probability at least $1-\delta$, a random state $\phi'$ in $\Phi_\ell$ satisfies $\TD(\phi',\phi_S)\le\sqrt\delta$, or in other words $|\langle \phi_S, \phi'\rangle|^2 \ge 1-\delta.$
Therefore, $\phi'$ corresponds to a subset state of size $\not\in (1\pm\epsilon^2)(1\pm\delta)|T|$.
Hence the probability that any other representative state passes the swap test is at most $\frac{1+(1-\epsilon^2)(1-\delta)}{2}$.
We can conclude, with probability at least 
\[
    (1-\delta)\cdot\left(1 - \frac{1+(1-\epsilon^2)(1-\delta)}{2} 
    \right) \ge \frac{\epsilon^2}{2}- \delta,
\]
the swap test rejects. Consequently, from now on we further assume that all the representative states in $\Phi_\ell$ corresponds to a subset state of support size $s\in (1\pm \epsilon^2)|T|$, as otherwise the accepting probability will be a constant away from 1.

\underline{Case 3: Attack caught by (ii)-(iii) Subset Test}. Pick some arbitrary representative state $\phi_i$  of $\Phi_i$ for $i=0,1,\ldots,\ell$, let $s_i$ be the support size for each $\phi_i$, then
\begin{align} \label{eq:telescoping-prod}
    s_0\cdot \frac{s_1}{s_0}\cdot \frac{s_2}{s_1}\cdots\frac{s_\ell}{s_{\ell-1}} = s_\ell \in (1\pm \epsilon^2)|T|.
\end{align} 
To cheat, the adversary can tell a wrong estimate of $|T|$, meaning $2^{-\ell} |N| \not \in (1\pm \epsilon)|T|$.
\begin{claim}\label{claim:bad-index}
    If the adversary tells a wrong estimate of $|T|$. Then for some $i\ge 1$, either,
\begin{align*}
    \frac{s_i}{s_{i-1}} &\ge \frac{1}{2}+\frac{\ln (1+\epsilon^2)}{\ell},
\end{align*}
or, 
\begin{align*}
    \frac{s_i}{s_{i-1}} \le \frac{1}{2}-\frac{\epsilon^2}{2\ell}.
\end{align*}
\end{claim}
\begin{proof}
For the purpose of contradiction, suppose the claim is false. That is for all $i$, the fraction between $s_i$ and $s_{i-1}$ is very close to $1/2$.
Consider two possible situation, 
first, if $2^{-\ell} N > (1+\epsilon)|T|$, then
    \begin{align*}
    s_\ell 
    &\ge \left(
        \frac{1}{2}-\frac{\epsilon^2}{2\ell}
    \right)^\ell s_0 
    \ge (1-\epsilon^2)(1+\epsilon)2^{-\ell}s_0
    > (1-\epsilon^2)(1+\epsilon) |T|
    > (1+\epsilon^2) |T|.
\end{align*}
This contradicts (\ref{eq:telescoping-prod}).
Second, if $2^{-\ell} N < (1-\epsilon)|T|$, then
\begin{align*}
    s_\ell 
    &\le \left(
        \frac{1}{2}+\frac{\ln (1+\epsilon^2)}{2\ell}
        \right)^\ell s_0 
        \le 2^{-\ell}s_0\exp(\ln(1+\epsilon^2)) = (1+\epsilon^2)2^{-\ell}s_0< (1-\epsilon^2)|T|,
\end{align*}
again, contradicting (\ref{eq:telescoping-prod}).
\end{proof}
W.l.o.g., say $i=1$ is an index satisfying the above claim. We show that the (odd) subset test rejects with high probability. 
In the subset test, each collection is partitioned into two parts of equal size,
\begin{align*}
    \Phi_1 = \Phi'_1 \sqcup \Phi''_1,\quad
    \Phi_2 = \Phi'_2 \sqcup \Phi''_2,\quad
    \Psi_2 = \Psi_2' \sqcup \Psi''_2.
\end{align*}
By definition, each part will be a $2\delta$-tilted state. Furthermore, by~\cref{prop:tensor-tilted-states} they form three collections of $4\delta$-tilted states,
\begin{align*}
    \Gamma_0:=\Phi'_1\otimes \Phi_2',\quad
    \Gamma_1:=\Phi''_1\otimes \Psi_2',\quad
    \Gamma_2:=\Phi_2''\otimes\Psi_2''.
\end{align*}
In subset test, $\Gamma_0, \Gamma_1, \Gamma_2$ will be fed to swap test and estimate $s_2/s_1$. By~\cref{claim:bad-index},
\[
    \frac{s_2}{s_1} \not\in \frac{1}{2} \pm \frac{\epsilon^2}{2\ell}.
\]
Let $\kappa= \Theta(\epsilon^2/(2\ell))$, then one of the following must be true by~\cref{lem:key_subset_mass},
\begin{enumerate}
    \item $|\langle \phi_2, \psi_2 \rangle|^2 \ge \kappa^4, $
    \item $||\langle \phi_1, \phi_2 \rangle|^2 - \frac{1}{2}|\ge \kappa^4, $
    \item $||\langle \phi_1, \psi_2 \rangle|^2 - \frac{1}{2}|\ge \kappa^4. $
\end{enumerate}
Without loss of generality say (i) hold. Then
\begin{equation}\label{eq:repr-swap-bound}
    \left|\Pr[\text{SwapTest}(\phi_1,\phi_2)\text{ accept}] - \frac{1}{2} \right| \ge \kappa^4/2.
\end{equation}
By~\cref{prop:tilted-state-approx}, the estimate $\zeta$ from $\Gamma_0$ will be
\begin{align}
    &\Pr_\zeta \left[|\zeta- \Pr[\text{SwapTest}(\phi_1,\phi_2)\text{ accept} |]  \ge 4\sqrt{4\delta}\right] = \exp(-\Omega(\delta m))
    \nonumber \\
    &\qquad\stackrel{(\ref{eq:repr-swap-bound})}{\Longrightarrow}\quad
    \Pr_\zeta \left[|2\zeta- 1 | \le \kappa^4   - 16\sqrt{\delta}\right] = \exp(-\Omega(\delta m)).\label{eq:zeta-bound}
\end{align}
Choose suitable parameters that satisfy,
\[
    \gamma \le     
    \kappa^4-16\sqrt{\delta}, ~\kappa^4\ge 17\sqrt{\delta}.
\]
Hence, the probability that the subset test accepts is $\exp(-\Omega(\delta m))=\exp(-\Omega(n^8)).$ 
\end{proof}

\subsection{Discussion}
Let us review some of our indistinguishability results in~\cref{sec:quantum-indist} and~\cref{sec:classical-lower-bound}, we want to emphasize some remarkable perspectives of our upper bound result. 

First in~\cref{sec:quantum-indist}, we pointed out that there is no tester with any advantage for testing productness with a single copy of a given state $\ket\psi$. 
What if proofs are allowed? It is immediate that the honest prover can give another copy of $\ket\psi$ as the proof, then by product test~\cite{HM13}, the verifier can distinguish product states and those far from being product. 
The unsatisfying feature is that the role of the proof is very limited, it serves as just another copy of the state. 
Quantitatively, if we count the total number of resources used in the algorithm, i.e., the proof complexity plus the given state $\ket\psi$, then proofs  gain us nothing! 
Because with two copies of the state, one can carry out product test anyways. 
Therefore the more interesting question would be to demonstrate properties for which without proof the copy complexity is super-polynomial but with polynomial-size proofs the copy complexity becomes polynomial. This is what our example illustrates.

Second, think about the lower bound that we presented in~\cref{sec:classical-lower-bound} for distinguishing the flat distribution of support size $s$ and $2s$. That would be the classical counterpart of the quantum property for which we demonstrate the power of proofs. However, in this classical setting, the power of proofs is completely gone! To see this note that $\Theta(\sqrt{s})$ samples are necessary and sufficient for distinguishing flat distribution of support size either $s$ or $2s$. With certificates, on the other hand, our lower bound in~\cref{thm:subset-distribution-proof} shows that only when the certificate length is $\Omega(\sqrt{s})$, the certificate can reduce the sample complexity. However, the total resources needed, i.e., certificate length plus sample complexity is $\Omega(\sqrt{s})$, match that without certificates. In fact, recall that in~\cref{re:tightness-classical}, the optimal proof strategy is to send some extra samples.

Finally, for the ``productness'' example in~\cref{sec:quantum-indist}, the best prover strategy would make proof an additional copy, not really providing any extra power; for the subset state, with flat certificates we can estimate the support size using exponentially smaller amount of resources. What about the other ensembles with different support size. It is not hard to see that a proof helps using analogous strategy: Ask the prover to provide certificates that will be subset state of size $\approx p d$, which as we have seen is testable with with flat certificate. 
In fact, in the dense regime, the flat certificate can be further relaxed to nonnegative amplitudes certificates~\cite{JW23} using their sparsity test.

\subsection{Lower Bounds for Quantum-to-Quantum State Transformation}

We now discuss how the study of quantum property testing protocols even under very strong assumptions on
the structure of the proofs can have interesting consequences for quantum-to-quantum state transformation. To this end,
we will suppose that we have obtained the following results for some quantum property $\calP$.
\begin{description}
  \item[(i.)] \textbf{Testability under Certificates.} We managed to design a tester with quantum proofs for property $\calP$ using ``few'' copies of the input state, but assuming
  that the honest proofs satisfy the strong promise of being a chosen function of the state being tested.
  \item[(ii.)] \textbf{$\propBQP$ Hardness.} We also managed to show that property $\calP$ requires ``many'' copies to be tested (using only copies of the
  state to be tested).
\end{description}
Now we can consider a quantum-to-quantum transformation that takes a certain number of copies of a state and produces a state that
is (close to) the chosen function (mentioned above) of the input states. Note that combining this hypothetical transformation with the tester for $\calP$
with the promised structured proofs (from the first item above) yields a tester (using only copies the input state) for property $\calP$. We illustrate this
scenario in~\cref{fig:quantum-to-quantum-lb}. By appealing to the second item above, we would deduce that ``many'' copies are needed to implement this
quantum-to-quantum transformation. Curiously, these considerations also illustrate that the study of quantum-to-classical results can have implications for
quantum-to-quantum results.

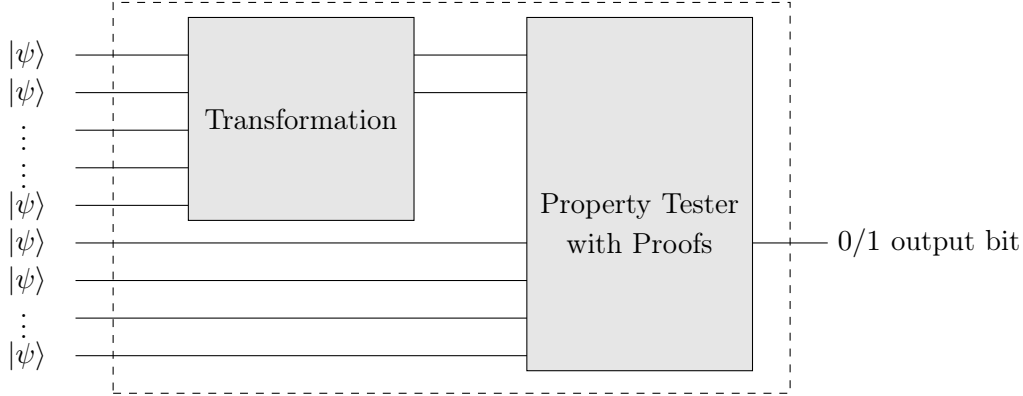
\begin{figure}[h!]
  \centering
\begin{tikzpicture}

    \draw[draw, dashed] (0.5,0.7) rectangle (9.5,-4.5);
    
    \node[left] at (-0.25,0) {$\ket{\psi}$};
    \node[left] at (-0.25,-0.5) {$\ket{\psi}$};
    \node[left] at (-0.5,-1) {$\vdots$};
    \node[left] at (-0.5,-1.5) {$\vdots$};
    \node[left] at (-0.25,-2.0) {$\ket{\psi}$};
    \node[left] at (-0.25,-2.5) {$\ket{\psi}$};
    \node[left] at (-0.25,-3.0) {$\ket{\psi}$};
    \node[left] at (-0.5,-3.5) {$\vdots$};
    \node[left] at (-0.25,-4.0) {$\ket{\psi}$};

    \draw (0,0) -- (9,0);
    \draw (0,-0.5) -- (9,-0.5);
    \draw (0,-1) -- (1.5,-1);
    \draw (0,-1.5) -- (1.5,-1.5); 

    \draw (0,-2.0) -- (1.5,-2.0);
    \draw (0,-2.5) -- (9,-2.5);
    \draw (0,-3.0) -- (9,-3.0);

    \draw (0,-3.5) -- (9,-3.5);
    \draw (0,-4.0) -- (9,-4.0);

    \draw[draw, fill=gray!20] (1.5,0.5) rectangle (4.5,-2.2);
    \node at (3.0, -0.85) {Transformation};

    \draw[draw, fill=gray!20] (6,0.5) rectangle (9,-4.2);
    \node at (7.5, -2) {Property Tester};
    \node at (7.5, -2.5) {with Proofs};

    \draw (9,-2.5) -- (10,-2.5);
    \node[right] at (10,-2.5) {$0/1$ output bit};

\end{tikzpicture}
\caption{A pictorial representation of combining a quantum-to-quantum state transformation that generates suitable proofs with a property tester promised to receive
  copies of the input quantum state, as well as, these suitable (structured) proofs. This combination yields a property tester (depicted as the dashed enclosing box)
  using only copies of the input state and no proofs.}
\label{fig:quantum-to-quantum-lb}
\end{figure}

For instance, using the above template, we can deduce the following quantum-to-quantum transformation lower bounds. The first about transforming
the amplitudes of a quantum state into their absolute values.

\begin{corollary}[Hardness of Absolute Amplitudes Transformation]
\label{theo:hardness_of_abs_value}
  Any transformation that takes $k$ copies of an arbitrary $n$-qubit quantum state $\ket{\psi} = \sum_{x \in \set{0,1}^n} \alpha_x \ket{x}$ and produces
  a single $n$-qubit output state at least $0.001$ close to $\ket{|\psi|}:=\sum_{x \in \set{0,1}^n} \abs{\alpha_x} \ket{x}$ requires $k = 2^{\Omega(n)}$.
\end{corollary}
For the absolute function, we consider distinguishing the subset state and the random binary phase state. The $\propBQP$ hardness follows~\cref{subsec:dense_indist} where we showed that to distinguish the subset state and the ``two-mode'' state requires $2^{\Omega(n)}$ copies.
The testability under certificates is simple. For any subset state $\ket\psi$, take $\ket\phi$ such that
$
    \TD(\ket\phi, \ket{|\psi|})\le\epsilon,
$
then swap test accepts $\ket\psi$ and $\ket\phi$ with probability at least $1-\epsilon^2/2.$ However for a random two-mode state $\psi$, swap test accepts $\ket\psi, \ket{\abs\psi}$ w.p. $1/2+o(1)$ almost surely. So for any $\ket\phi$ that is $\epsilon$ close to $\ket{\abs\psi}$, $\TD(\psi,\abs\psi)\le 1/2+\epsilon+o(1)$. Meaning that if we can construct the absolute-value state using a small number of copies, we can distinguish subset state and random binary phase state with a small number of copies.

The second transformation lower bound is for mapping the amplitudes to their complex conjugates follows~\cref{thm:subset-state-PRS}.

\begin{corollary}[Hardness of Amplitude Conjugation Transformation (Informal)]\label{theo:hardness_of_conj}
  Any transformation that takes $k$ copies of an arbitrary $n$-qubit quantum state $\ket{\psi} = \sum_{x \in \set{0,1}^n} \alpha_x \ket{x}$ and produces
  a single $n$-qubit output state at least $0.001$ close to $\sum_{x \in \set{0,1}^n} \alpha_x^* \ket{x}$ requires $k = 2^{\Omega(n)}$.
\end{corollary}

%% file: hierarchy.tex
\section{Property Testing Complexity Classes and Hierarchy}
\label{sec:hierarchy}

Finally, we take the opportunity to define some obvious property testing complexity classes regarding the information theoretic sample/copy complexity. Consider some property $\calP=\sqcup \calP_N$, where $\calP_N$ can be a subset of $\fP(\mathbb{C}^N)$ or $\Delta_N$. In the following, we only give a subset of the complexity classes for properties of quantum states, and the complexity classes regarding properties of classical distributions can be defined totally analogous. 
\begin{definition}[Property Testing Complexity Class] Let $n=\log(N)$, fix some  constant $c = 2/3, s= 1/3$,
    \begin{align*}
        \propBQP &:= \propBQP_{c,s}\left[\poly(n)\right] ;\\
        \propEXP &:= \propBQP_{c,s}\left[2^{\poly(n)}\right];\\
        \textup{PropMA} &:= \textup{PropMA}_{c,s}\left[\poly(n),\poly(n) \right];\\
        \propMAexp &:= \propMA_{c,s}\left[\poly(n),2^{\poly(n)}\right];\\
        \propQMA(k) &:= \propQMA(k)_{c,s}[\poly(n), \poly(n)];\\
        \propQMA &:= \propQMA(1);\\
        \propQMAexp &:= \propQMA(1)_{c,s}\left[\poly(n),2^{\poly(n)}\right];\\
        \propAM(m, r) &:= \propAM(m, r)_{c,s}[\poly(n), \poly(n)];\\
        \propIP &:= \propIP(1,\poly(n))_{c,s}[\poly(n),\poly(n)];\\
        \propMIP &:= \propIP(\poly(n),\poly(n))_{c,s}[\poly(n),\poly(n)];\\
        &\quad \vdots
    \end{align*}
\end{definition}

For notations,  we do not make any distinction between testing classical distribution or quantum states. It is normally very clear from the context if the problem of interest is to test classical distribution or to test quantum states.
This notation could be less standard for the context of sublinear algorithms, as for us $\propBPP$ is polynomial with respect to $n=\log(N)$, where $N$ is the size of the discrete probability space in the case of property testing for classical distributions. Our choice is more natural in the context of this paper, as testing probability distribution can be viewed as a degenerated version of testing quantum states.

Some of the notations $\propBPP,\propQMA,$ etc. are used for both the property testing models as well as the property testing complexity classes. This is a somewhat common abuse of notation. To give an alert for the unfamiliar readers, consider the following two statement, for some property $\calP$,
\begin{enumerate}
    \item $\calP\in \propQMA$,
    \item $\propQMA(\calP)=\exp(\Omega(n))$.
\end{enumerate}
In the first case, $\propQMA$ is a complexity class. $\calP\in \propQMA$ is an upper bound result, i.e., $\calP$ can be tested using $\poly(n)$ copies assisted with a $\QMA$ type prover. On the other hand, in the second case, $\propQMA$ is the property testing model, and $\propQMA(\calP)=\exp(\Omega(n))$ is a lower bound result, meaning that in the $\propQMA$ model, one needs $\exp(\Omega(n))$ copies to test $\calP.$

Now we collect some obvious relationship regarding these complexity classes either follows easily from the definition or from the results proved in the previous sections. Let $\textup{ALL}$ denote the set of all properties.
\begin{proposition}\label{prop:trivial-upper-bound}
For both classical distribution and quantum state properties
\begin{equation}
    \propEXP = \mathrm{ALL}.
\end{equation}
\end{proposition}
\begin{proof}
     The statement holds for both quantum and classical properties, because exponentially many samples/copies are sufficient for learning to quantum states and classical distribution~\cite{OW16tomo,haah2016tomo,delaVegaK07}.
\end{proof}
This proposition justifies our definition of the property testing class $\propMAexp$ and $\propQMAexp$, where the number of samples is polynomial instead of exponential. It turns out that $\propQMAexp$ is also a trivial upper bounds for any quantum state properties.
\begin{proposition}\label{prop:quantum-trivial-upper-bound}
    For quantum state properties,
    \begin{equation}
        \propQMAexp = \mathrm{ALL}.
    \end{equation}
\end{proposition}
\begin{proof}
    The statement holds because the the prover can send a classical description of the state $\ket\psi$ to test. Then the verifier can prepare a state $\ket\phi$ based on the classical description and use swap test to check if the classical description is the correct. 
    In particular, set $\delta=1/\exp(\poly(d)).$ Estimate the overlap $|\langle \phi \mid \psi \rangle|^2$. 
    
    For $\ket\psi\in\calP$, the honest prover sends a correct description of $\ket\psi$ up to the precision $\delta=1/\exp(\poly(d))$, thus $|\langle \psi\mid\phi\rangle|^2 \ge 1-\delta$. Therefore, $\ket\phi$ is $\delta$ close to $\calP$. Given $k=\polylog(d)$ many copy of $\ket\psi$, the probability that all $k$ swap test passes is
    \[
        (1-O(\delta))^k = 1-1/\exp(\poly(d)).
    \]

    On the other hand, for $\ket\psi$ $\epsilon$-far from $\calP$, if the verifier lies by giving some $\ket\phi$ that is $\delta$ close to $\calP$, then $|\langle\psi\mid\phi\rangle|^2\le 1-\epsilon+o(\epsilon)$. Therefore all the swap test passes with probability at most
    \[
        (1-\epsilon+o(\epsilon))^k= \exp(-\epsilon k),
    \]
    which is tiny for any constant $\epsilon$.
\end{proof}

One may wonder if an analogous statement is true for classical distribution properties. As we see in~\cref{thm:subset-distribution-proof}, it is not: There are untestable classical properties regardless of the proof length (which certainly can be a classical description of some quantum state). Therefore, this is another example of how quantum coherence enlarges the testability.
\begin{fact}
    For  property testing of classical distribution, 
    \[
        \propMAexp\not= \textup{ALL}.
    \]
\end{fact}
An interesting question we left for future investigation is whether there are other interesting upper bound for an arbitrary quantum properties. In view of~\cref{prop:quantum-trivial-upper-bound} and the classical result $\classfont{MIP}=\NEXP$, a concrete question is the following
\begin{problem}
    Is it true that
    \[
        \propMIP = \textup{ALL}?
    \]
\end{problem}

In terms of the proof system, the seminar work of Goldwasser and Sipser proved a surprising result $\IP=\AM(1, \poly(n))$~\cite{GS86}, i.e., public-coin interactive proof system is as powerful as the private-coin proof system. However, in view of~\cref{thm:IP-upper-bound} and~\cref{cor:AM-k-lower-bound}, private-coin is significantly more powerful in property testing for both classical distribution and quantum states, 
\begin{equation}\label{eq:prop-public-vs-private-coin}
        \propAM(1,\poly(n))\subsetneq \propIP.
\end{equation}
Finally, \cref{cor:demerlin-full} implies the following collapses
\begin{fact} For property testing of quantum states and distributions, $\MA$ and $\QMA$ type proof does not increase testability,
   \begin{align}
        \propQMA = \propBQP.\label{eq:propqma-propbqp}\\
        \propMA = \propBPP.
        \label{eq:proma-propbpp}
\end{align} 
\end{fact}
In the context of property testing for distributions, (\ref{eq:prop-public-vs-private-coin}) and (\ref{eq:proma-propbpp})
are proved by Chiesa and Gur~\cite{CG18} using different arguments. In fact, the key ingredient to~\cref{cor:demerlin-full} is a witness preserving gap amplification happening in~\cref{thm:demerlin-special}, and $\propAM(\poly(n),\poly(n))$ admits such gap amplification since the interaction between the verifier and the provers is independent with the samples from the classical distribution or measurements on the copies of quantum state, therefore, we actually have a very strong collapse in terms of complexity classes
\begin{align}
    \propAM(\poly(n),\poly(n)) = \propBPP.\label{eq:AM-k-collapse-classical}\\
    \propAM(\poly(n),\poly(n)) = \propBQP.\label{eq:AM-k-collapse-quantum}
\end{align}

Interestingly, the question whether multiple unentangled provers help property testing for quantum states remains. Can we de-Merlinize $\propQMA(2)$ as well?
\begin{problem}
    Is it true that
    \[
    \propQMA = \propQMA(2)?
    \]
\end{problem}
We remark that a negative answer rules out input dimension efficient disentanglers, and it could potentially lead to a full resolution of the
disentangler conjecture (depending on the strength of the parameters it achieves).\footnote{A formal statement of the disentangler conjecture can be found in~\cite{ABDFS08}. Roughly speaking, it says that a quantum channel that maps quantum states to approximately separable states that is approximately surjective must have its input dimension exponential on the output dimension.} This happens because if such a
dimension efficient disentangler existed, then we would have the collapse $\propQMA = \propQMA(2)$, since we would
be able to simulate $\propQMA(2)$ in $\propQMA$ by ``breaking'' the entanglement of the proof.
The advantage of the property testing model compared to the other information theoretical model (e.g. black box model),
is we actually know $\propQMA=\propBQP$. Showing lower bounds for $\propBQP$ could potentially be a much easier task.
A positive answer is also extremely interesting. Note that naively, the de-Merlinization strategy does not work for $\propQMA(2)$.
So a positive answer can provide deeper quantum information insight on separable states, which may lead to progress in the problem regarding the power
of $\QMA(2)$ itself.

\subsection{Information Theoretic versus Computation Constrained Models}
The study of property testing can be broadly divided into two main categories: information theoretic
and computation constrained testing. In the former category, no computation assumption is made about the tester, one can think the tester is computationally unbounded. In the latter category, we impose
that a tester has to be generated uniformly by a Turing machine, and it has to obey the computation resource
constraints of the corresponding complexity class (\eg in this model a tester for $\propQMA$ is required to be a $\BQP$ verifier).
In particular, these models can capture the following behavior regarding property testing and computation complexity.

\begin{enumerate}
  \item Information theoretic: captures the inherent limitations imposed by quantum/classical information theory regardless of
                               any computation limitations on a tester.
  \item computation constrained with 
        \begin{itemize}
            \item[$\bullet$] quantum input states: captures decision problems with quantum inputs under resource constraints.
            \item[$\bullet$] classical input states: captures standard complexity classes.
        \end{itemize}
\end{enumerate}

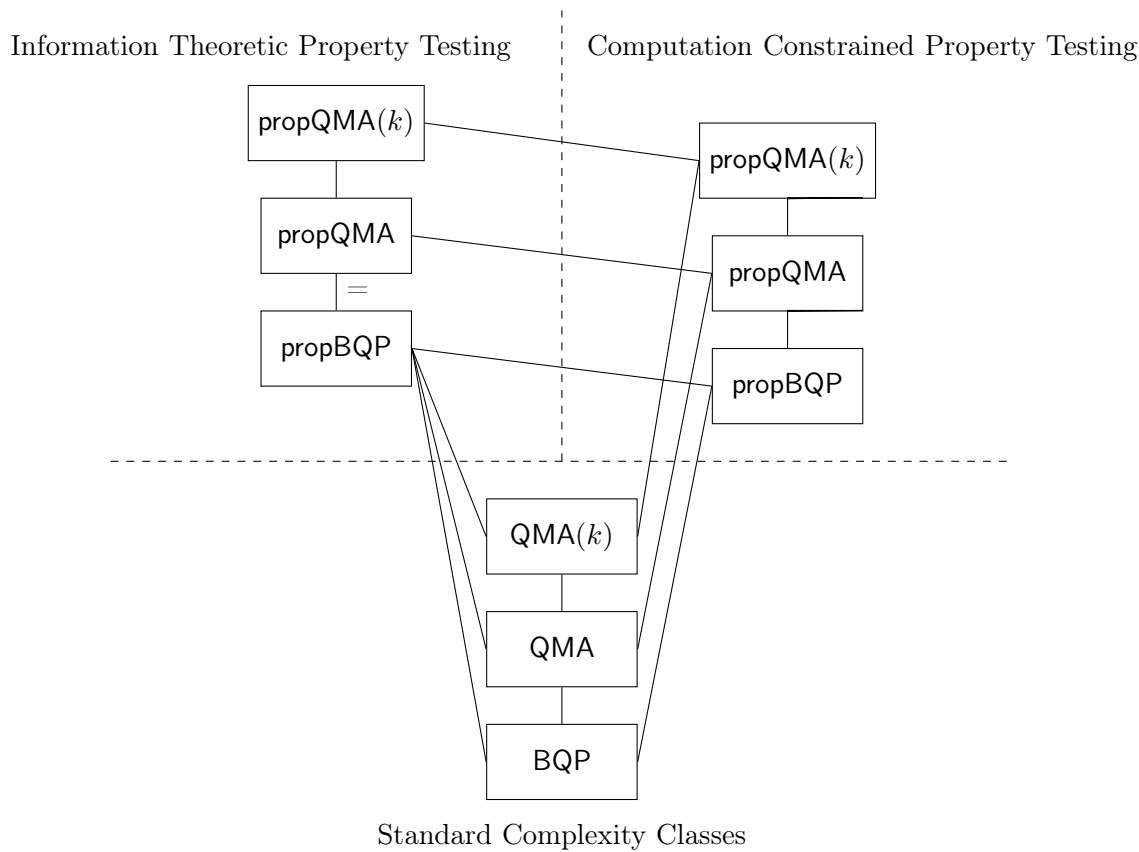
\begin{figure}[h!]
\begin{tikzpicture}

\node[draw, rectangle, minimum width=2cm, minimum height=1cm] (Prop) at (-3,0) {$\propBQP$};
\node[draw, rectangle, minimum width=2cm, minimum height=1cm] (propQMA) at (-3,1.5) {$\propQMA$};
\node[draw, rectangle, minimum width=2cm, minimum height=1cm] (propQMAk) at (-3,3) {$\propQMA(k)$};

\draw[-] (Prop.north) -- node[right] {=} (propQMA.south);  
\draw[-] (propQMA.north) -- (propQMAk.south);

\node[draw, rectangle, minimum width=2cm, minimum height=1cm] (Prop2) at (3,-0.5) {$\propBQP$};
\node[draw, rectangle, minimum width=2cm, minimum height=1cm] (propQMA2) at (3,1) {$\propQMA$};
\node[draw, rectangle, minimum width=2cm, minimum height=1cm] (propQMAk2) at (3,2.5) {$\propQMA(k)$};

\draw[-] (Prop2.north) -- +(0,0.5) -- +(1,0.5) -- (propQMA2.south);
\draw[-] (propQMA2.north) -- +(0,0.5) -- +(1,0.5) -- (propQMAk2.south);

\draw[-] (Prop.east) -- (Prop2.west);       
\draw[-] (propQMA.east) -- (propQMA2.west); 
\draw[-] (propQMAk.east) -- (propQMAk2.west); 

\node[draw, rectangle, minimum width=2cm, minimum height=1cm] (BQP) at (0,-2.5) {$\QMA(k)$};
\node[draw, rectangle, minimum width=2cm, minimum height=1cm] (QMA) at (0,-4) {$\QMA$};
\node[draw, rectangle, minimum width=2cm, minimum height=1cm] (QMAk) at (0,-5.5) {$\BQP$};

\draw[-] (BQP.south) -- (QMA.north);
\draw[-] (QMA.south) -- (QMAk.north);

\draw[-] (BQP.east) -- (propQMAk2.west);

\draw[-] (QMA.east) -- (propQMA2.west);

\draw[-] (QMAk.east) -- (Prop2.west);

\draw[-] (BQP.west) -- (Prop.east);

\draw[-] (QMA.west) -- (Prop.east);

\draw[-] (QMAk.west) -- (Prop.east);

\draw[dashed] (0,-1.5) -- (0,4.5);  
\draw[dashed] (-6,-1.5) -- (6,-1.5);  

\node at (0,-6.5) {Standard Complexity Classes};

\node at (-4,4) {Information Theoretic Property Testing};
\node at (4,4) {Computation Constrained Property Testing};

\end{tikzpicture}
\caption{We depict the relationships among property testing both in the information theoretic models (on the upper left),
         the computation constrained models (on the upper right), and the standard complexity classes (on the bottom) for the case
         of BQP, QMA, and QMA(k). Line segments from bottom to top indicate containments (\ie the model on top can test at least all the
         properties its connecting bottom model can).}\label{fig:hiearchy}
\end{figure}

\begin{remark}
  Any language or promise problem in a complexity (or computability) class with classical inputs
  gives rise to two disjoint collection of bit strings $L_{\textup{yes}}$ and $L_{\textup{no}}$ consisting in yes and no instances,
  respectively. We remark that the information theoretic version of the class $\propBPP$ and $\propBQP$ trivially capture them.
\end{remark}

\begin{remark}
  By considering classical input states in the complexity constrained models of property testing, we can ask whether a classical bit
  string (given as input state to be tested) is a yes or no instance of a language or promise problem. Therefore, these models capture
  standard complexity classes. Under the assumption $\BQP \ne \QMA$, we have $\propBQP \ne \propQMA$ for their
  computation constrained models.
\end{remark}

\begin{remark}
  In contrast, for the information theoretic models, we have the collapse $\propBQP = \propQMA$. This means that a general
  quantum proof cannot substantially improve information theoretic testability of quantum properties (they can at best reduce polynomially
  the number of copies of the input state, or improve the efficiency of the tester).
\end{remark}

We summarize the above remarks in~\cref{fig:hiearchy}.


\subsection{Summary of Our Information Theoretic in terms of Property Testing Classes}

In~\cref{fig:lower_bounds}, we provide a visual summary of some of our information theoretic results for support
size~\cref{theo:failure_coherence,theo:failure_am_testing,theo:subset_advantage} in terms of the property testing classes.

\usetikzlibrary{positioning}
\usetikzlibrary{calc}

\begin{figure}[h]
  \centering
\begin{tikzpicture}[node distance=1cm and 1.5cm, every node/.style={draw, minimum width=3cm, minimum height=1cm, align=center}]

    \node[draw=none] (coherentInputLabel) at (-2,0) {Coherent States};
    \node[below=of coherentInputLabel, yshift=-1.5cm] (PropQMAk) {$\propQMA_{\textup{subset}}(\poly(n))$};
    \node[below=of PropQMAk, dashed] (PropQMA) {$\propQMA$};
    \node[below=of PropQMA, dashed] (PropBQP) {$\propBQP$};

    \node[draw=none] (classicalDistributionLabel) at (4,0) {Classical Distribution};
    \node[below=of classicalDistributionLabel, yshift=-1.5cm, dashed] (PropMAk) {$\propMA_{\textup{flat}}\left(2^{\Omega(n)}\right)$};
    \node[below=of PropMAk, dashed] (PropMA) {$\propMA$};
    \node[below=of PropMA, dashed] (PropBPP) {$\propBPP$};

    \node[above=of PropMAk, dashed] (PropAMk) {$\propAM\left(2^{\Omega(n)}\right)$};

    \draw[-] ($(PropAMk.north)!0.5!(coherentInputLabel.south)+(0,1.5)$) -- ($(PropBQP.south)!0.5!(PropBPP.south)-(0,1.5)$);

    \draw[->] ($(PropBPP.south east)+(2,-0.5)$) -- ($(PropAMk.north east)+(2,0.5)$) node[midway, draw=none, right] {{\footnotesize Model Strength}};
    
\end{tikzpicture}

\caption{A pictorial representation of the limitations in distinguishing support size of flat coherent quantum states (depicted on the left column) and
  flat classical distributions (on the right column). Dashed boxes indicate that the model fails in this task whereas a solid box indicates that
  the model succeeds.}\label{fig:lower_bounds}
\end{figure}
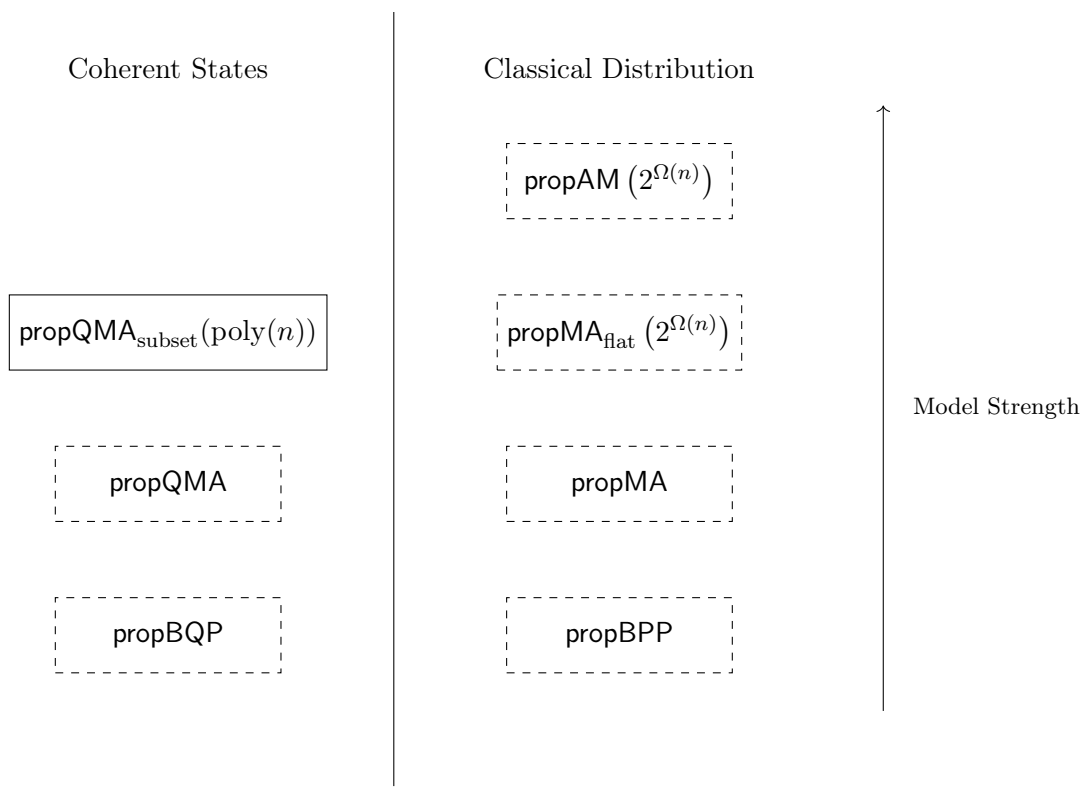

%% file: appendix.tex
\appendix
\section{Divergence Contraction}

In this section, we finish the proof of~\cref{lem:div-contraction}. 
It remains to establish~\cref{claim:chain-bound}.
Before we prove~\cref{claim:chain-bound}, we need to introduce some definitions and recall several facts about elementary symmetric polynomials and KL-divergence.
\paragraph{Auxiliary Definitions and Facts.}
Recall that the \emph{downwalk operator} from $\binom{[N]}{s}$ to $\binom{[N]}{t}$ is defined
as follows
\begin{align*}
  D_{s \to t}(S,T) \coloneqq \begin{cases}
      \frac{1}{\binom{s}{t}}, & T\subseteq S;\\
      0, & \text{otherwise},
  \end{cases}
\end{align*}
for every  $S \in \binom{[N]}{s}$ to $T \in \binom{[N]}{t}$. Thus, viewing the distributions $\lambda_i, \mu_i$ in~\cref{lem:div-contraction} as row vectors, then $\lambda_i = \mu_i D_{s\to t} $
\begin{definition}[Generating polynomial]
    Given a distribution $\mu$ on $\binom{[N]}{s}$, its generating polynomial $P_\mu\in\R[X_1,X_2,\ldots, X_N]$ is 
    \[
        P_\mu(X):= \sum_{S\subseteq[N]: |S|=s} \mu(S) \prod_{i\in S}X_i.
    \]
\end{definition}

The well-known MacLaurin's inequality on elementary symmetric polynomials~\cite{maclaurin1730iv} reads: For nonnegative $X_1,X_2,\ldots, X_N$, and integers $s \ge t>0$,
\begin{equation}\label{eq:log-concave}
    \left(\Exp_{S\in \binom{[N]}{s}} \prod_{i\in S}X_i\right)^{1/s} \le \left(\Exp_{T\in \binom{[N]}{t}}\prod_{i\in T}X_i\right)^{1/t}.
\end{equation}
An immediate corollary is that for the uniform distribution $\mu$, its generating polynomial $P_\mu$ is log-concave on nonnegative inputs.

The next lemma about KL-divergence minimization follows from duality theory of convex optimization.
\begin{lemma}[See Appendix B of~\cite{singh2014entropy}]\label{lem:sing14}
    Given a distribution $\mu$ on $\binom{[N]}{s}$, and a distribution $q:[N]\to \R$, then
    \begin{equation}
        \inf_{\nu:\binom{[N]}{s}\to \R} \{\KL{\nu}{\mu}: q=\nu D_{s\to 1}\} = -\log\left(
            \inf_{x_1,x_2,\ldots, x_N >0} \frac{P_\mu(x)}{(x_1^{q(1)}x_2^{q(2)}\cdots x_N^{q(N)})^s}
        \right).
    \end{equation}
\end{lemma}

\paragraph{Proof of~\cref{claim:chain-bound}.} Now we are ready to prove~\cref{claim:chain-bound}. Without loss of generality, say $x_i=i$. Let $\mu_1', \mu_0'$ be the induced distribution of $\mu, \mu_1$ on $\binom{\{i,i+1,\ldots,N\}}{s-i+1}$ conditioning on $[i]\subseteq S, S'$. Then $\mu_0'$ is uniform on $\binom{\{i,i+1,\ldots, N\}}{s-i+1}$. 
Let $q:= \mu_1' D_{s-i+1\to 1}$, then
\begin{align}
    &\KLfrac{X_i X_{i+1} \ldots X_s \mid X_{<i}=x_{<i} }{Y_i Y_{i+1} \ldots Y_s \mid Y_{<i}=x_{<i}} 
        \nonumber \\
        &\qquad\qquad= \KL{\mu_1'}{\mu_0'}
        \nonumber \\
        &\qquad\qquad \ge \inf_{\mu_1''} \{\KL{\mu_1''}{\mu_0'}: \mu_1'' D_{s-i+1\to 1} = q\}
        \nonumber   \\
        &\qquad\qquad = -\log\left(
            \inf_{z_i,z_{i+1},\ldots, z_N >0} \frac{P_{\mu_0'}(z)}{(z_i^{q(i)}z_{i+1}^{q(i+1)}\cdots z_N^{q(N)})^{s-i+1}}
        \right)
        \nonumber   \\
        &\qquad\qquad \ge -\log\left(
            \inf_{z_i,z_{i+1},\ldots, z_N >0} \left(\frac{\Exp_{j\in\{i,i+1,\ldots,N\}} z_j}{z_i^{q(i)}z_{i+1}^{q(i+1)}\cdots z_N^{q(N)})}\right)^{s-i+1}
        \right).
        \nonumber
\end{align}
where the second step is due to~\cref{lem:sing14}; the third step is due to MacLaurin's inequality.
Set $z_i = (N-s+1)q_i$, then
\begin{align}
    \KL{\mu'_1}{\mu'_0} &\ge -(s-i+1)\sum_{j=i}^N q(j)\log\frac{q_j}{1/(N-s+1)} 
    \nonumber\\
    & = (s-i+1)\KL{\mu_1' D_{s-i+1\to 1}}{\mu_0' D_{s-i+1\to 1}}
    \nonumber\\
    & = (s-i+1)  \KLfrac{X_i\mid X_{<i}=x_{<i}}{Y_i\mid Y_{<i}=x_{<i}}.
    \nonumber \qedhere
\end{align}

\section{Spectra of $\cD_{t}$ from Johnson Scheme}

The spectra of $\cD_{t}$ is known \cite{delsarte1973algebraic}. In particular,
fix any $0\le t\le k-1,$ there are $k+1$ distinct eigenvalues $\lambda_{0},\lambda_{1},\ldots,\lambda_{k},$
such that
\begin{align*}
 & \lambda_{0}=\binom{k}{t}\binom{d-k}{k-t},\\
 & \lambda_j = \sum_{\ell=\max\{0,j-t\}}^{\min\{j, k-t\}}(-1)^\ell\binom{j}{\ell}\binom{k-j}{k-t-\ell}\binom{d-k-j}{k-t-\ell},& j = 1,2,\ldots, k,
\end{align*}
with multiplicity
\begin{align*}
& m_{0}=1,\\
 & m_{j}=\binom{d}{j}-\binom{d}{j-1}, & j=1,2,\ldots,k.
\end{align*}
For us, we simplify $\lambda_j$ for $k=O(\sqrt d)$,
\begin{align*}
 & |\lambda_{j}|\lesssim\frac{\binom{k-j}{t-j}}{\binom{k}{t}}\lambda_{0}, & j=1,2,\ldots,t,\\
 & |\lambda_{j}|\lesssim\frac{\binom{j}{t}(k-t)!}{\binom{k}{t}(k-j)!}\cdot\frac{1}{d^{j-t}}\lambda_{0}, & j=t+1,\ldots,k.
\end{align*}
We bound $\|\cD_{t}\|_{1}$ as follows 
\begin{align*}
\|\cD_{t}\|_{1} & =\lambda_{0}+\sum_{j=1}^{k}m_{j}|\lambda_{j}|\\
 & \lesssim\lambda_{0}\left(1+\sum_{j=1}^{t}\frac{d^{\underline{j}}}{j!}\frac{\binom{k-j}{t-j}}{\binom{k}{t}}+\sum_{j=t+1}^{k}\frac{d^{\underline{t}}}{j!}\frac{\binom{j}{t}(k-t)^{\underline{j-t}}}{\binom{k}{t}}\right)\\
 & \lesssim\lambda_{0}\left(1+\frac{d^{\text{\ensuremath{\underline{t}}}}}{k^{\underline{t}}}+\sum_{j=t+1}^{k}\frac{d^{\underline{t}}}{k^{\underline{t}}}\binom{k-t}{j-t}\right)\\
 & \le\lambda_{0}\left(1+\frac{d^{\underline{t}}}{k^{\underline{t}}}2^{k-t}\right)\\
 & =\binom{k}{t}\binom{d-k}{k-t}
    \left(1+ \frac{d^{\underline{t}}}{k^{\underline{t}}}2^{k-t} \right)
    \\
 &\lesssim \binom{k}{t}\binom{d-k}{k-t} \frac{d^{\underline{t}}}{k^{\underline{t}}}2^{k-t}
 =\binom{d}{t}\binom{d-k}{k-t}2^{k-t}.
\end{align*}
This proves Fact~\ref{fact:johnson-trace}.